\documentclass[runningheads]{article}
\usepackage[margin=2cm]{geometry}
\usepackage[T1]{fontenc}
\newcommand*{\email}[1]{\href{mailto:#1}{\nolinkurl{#1}} } 

\usepackage{graphicx}

\usepackage{booktabs} % For formal tables
\usepackage{float}
\usepackage{diagbox}
\usepackage[table,xcdraw,rgb]{xcolor}
\usepackage[normalem]{ulem}
\useunder{\uline}{\ul}{}

\usepackage{amssymb, amsfonts, amsmath}
\usepackage{amsthm}
\usepackage[nocomma]{optidef}

\usepackage{enumitem}

\usepackage{hyperref}
\usepackage{cleveref}

\usepackage{algorithm}
\usepackage{algpseudocode}

\usepackage{color}

\usepackage{tikz} 
\usetikzlibrary{arrows.meta}
\usetikzlibrary{calc}
\usetikzlibrary{positioning}
\usetikzlibrary{quotes}
\usetikzlibrary{snakes}

\usetikzlibrary{hobby,arrows,backgrounds,calc,trees}

\usepackage[dvipsnames,svgnames,x11names,]{xcolor}

\definecolor{my_green}{HTML}{99c04c}
\definecolor{my_red}{HTML}{ff8ea2}
\definecolor{my_blue}{HTML}{6c94b0}
\definecolor{my_grey}{HTML}{bdc3c7}

\newcommand{\degree}{\text{deg}}
\newcommand{\Eends}{\text{ends}}
\newcommand{\type}{\text{type}}

\makeatletter
\def\namedlabel#1#2{\begingroup
    #2
    \def\@currentlabel{#2}
    \phantomsection\label{#1}\endgroup
}
\makeatother

\newtheorem{observation}{Observation}
\newtheorem{claim}{Claim}
\newtheorem{theorem}{Theorem}[section]
\newtheorem{lemma}[theorem]{Lemma}

\begin{document}

\title{Online Algorithm for Fractional Matchings with Edge Arrivals in Graphs of Maximum Degree Three}

\author{Kanstantsin Pashkovich$^1$ \and
Thomas Snow$^2$}

\date{
    $^1$University of Waterloo, Waterloo ON N2L 3G1, Canada\\[1ex]%
    \email{kpashkovich@uwaterloo.ca}\\[1ex]%
    $^2$University of Toronto, Toronto ON M5S 1A1, Canada\\[1ex]%
    \email{tsnow@cs.toronto.edu}\\[2ex]%
}

\maketitle

\begin{abstract}
We study online algorithms for maximum cardinality matchings with edge arrivals in graphs of low degree. Buchbinder, Segev, and Tkach showed that no online algorithm for maximum cardinality fractional matchings can achieve a competitive ratio larger than $4/(9-\sqrt 5)\approx 0.5914$ even for graphs of maximum degree three. The negative result of Buchbinder et al. holds even when the graph is bipartite and edges are revealed according to vertex arrivals, i.e. once a vertex arrives, all edges are revealed that include the newly arrived vertex and one of the previously arrived vertices. In this work, we complement the negative result of Buchbinder et al. by providing an online algorithm for maximum cardinality fractional matchings with a competitive ratio at least $4/(9-\sqrt 5)\approx 0.5914$ for graphs of maximum degree three. We also demonstrate that no online algorithm for maximum cardinality integral matchings can have the competitive guarantee $0.5807$, establishing a gap between integral and fractional matchings for graphs of maximum degree three. Note that the work of Buchbinder et al. shows that for graphs of maximum degree two, there is no such gap between fractional and integral matchings, because for both of them the best achievable competitive ratio is $2/3$. Also, our results demonstrate that for graphs of maximum degree three best possible competitive ratios for fractional matchings are the same in the vertex arrival and in the edge arrival models.

\end{abstract}

\section{Introduction}

Matchings constitute an extensively studied area of mathematics and theoretical computer science with various practical applications.
Indeed, matchings arise in different areas of our everyday lives: job placements for students, assigning riders to drivers on a ride-sharing platform, providing advertisement spots, etc. Some of these matchings have an online nature since the edges or vertices in the underlying graph appear at certain timepoints and are available only in a specific time frame.

In this work, we study maximum cardinality matchings in an adversarial edge arrival model. In this model, at every timepoint a new edge arrives. In the integral matching case, upon the arrival of a new edge, we need to immediately and irrevocably decide whether to include this new edge in our current matching. In the fractional matching case, we need to irrevocably select a value for each new edge such that for every vertex, the sum of values on incident edges is always at most one. To make these decisions, we rely on online algorithms. To estimate the performance of an online algorithm, we select as a benchmark the cardinality of a maximum matching in the already "arrived" graph.

In this paper, we focus mainly on the adversarial edge arrival model in graphs of maximum degree three. We determine the best competitive ratio of online algorithms for fractional matchings in these graphs. To do this, we provide an online algorithm that achieves the guarantee $4/(9-\sqrt 5)\approx 0.5914$ on these graphs, where $4/(9-\sqrt 5)\approx 0.5914$ equals the corresponding upper bound obtained in~\cite{Buchbinder}. Due to the construction in~\cite{Buchbinder}, for graphs of maximum degree three, the best possible competitive ratio remains the same regardless of whether one considers general or bipartite graphs, and whether one considers the vertex arrival or edge arrival models. Additionally, we show that the guarantee $4/(9-\sqrt 5)\approx 0.5914$ cannot be achieved on graphs of maximum degree four; we also show that the above guarantee cannot be achieved for integral matchings in general graphs of maximum degree three.

In general, we know that the vertex arrival and edge arrival models lead to different competitive ratios for fractional matchings. Indeed, the results of~\cite{Wang} show that for general graphs, online algorithms can achieve a competitive ratio~$0.526$ in the vertex arrival model. The results of~\cite{Gamlath} show that no online algorithm can achieve a guarantee larger than~$0.5$ in the edge arrival model, even for bipartite graphs. Thus, at a certain value of the maximum degree, the best competitive ratio for the vertex arrival order is strictly larger than the competitive ratio for the edge arrival order; our work shows that this degree should be at least four.

\subsection{Our Results}

In our work, we focus on online algorithms for fractional matchings in the adversarial edge arrival model. \cite{Buchbinder} established a series of results for online matchings in the edge arrival model when the underlying graph has a bounded maximum degree. In particular, Buchbinder, Segev, and Tkach showed that no online algorithm for maximum cardinality fractional matchings can achieve a competitive ratio larger than $4/(9-\sqrt 5)\approx 0.5914$ even for graphs of maximum degree three. The negative result in~\cite{Buchbinder} holds even when the graph is a forest and edges are revealed according to vertex arrivals. In this work, we provide an online algorithm for maximum cardinality fractional matchings with a competitive ratio at least $4/(9-\sqrt 5)\approx 0.5914$ for graphs of maximum degree three, thus showing that $4/(9-\sqrt 5)\approx 0.5914$ is the best competitive ratio for graphs of degree three. So, for maximum cardinality fractional matchings, our result demonstrates that for graphs of maximum degree three, the competitive ratios are the same for the edge arrival and vertex arrival models. We also show that no online algorithm can achieve a guarantee larger than $0.5807$ for integral matchings in graphs of maximum degree three. Thus, unlike for graphs of maximum degree two~\cite{Buchbinder}, our results establish a gap between the best achievable competitive ratios for fractional and integral matchings in graphs of maximum degree three.

Next, we show that the guarantee of $4/(9-\sqrt 5)\approx 0.5914$ is not achievable in the graphs of maximum degree four. To show this, we provide an instance such that no online algorithm for maximum cardinality fractional matching can achieve a guarantee larger than $\approx 0.58884$ on it in the edge arrival model.

Another important contribution in~\cite{Buchbinder} is an elegant algorithm, so called MinIndex Algorithm. Buchbinder et al. show that MinIndex achieves the best possible guarantees both in the case of fractional and integral matchings when the maximum degree is at most two. We show that the guarantee achieved by the MinIndex algorithm is at most $5/9\approx 0.555$ in forests with maximum degree three. Our result improves on the upper bound $4/7\approx 0.571$ shown by Buchbinder et al. Note that Buchbinder et al. showed that the competitive ratio of the MinIndex algorithm equals $5/9$ for both integral and fractional matchings on forests, but their upper bound construction involves graphs with maximum degree four.

\subsection{Related Work}
The seminal paper~\cite{Karp-Paper} studied online matchings in the setting where the graph is bipartite and the vertices in one part appear over time. Each time a vertex appears, all of its incident edges are revealed, and one needs to make an irrevocable decision on which one of these edges to include in the matching, if any. \cite{Karp-Paper} provided a ranking algorithm that achieves the best possible competitive ratio of~$(1-1/e)$.

In the general adversarial edge arrival model, Gamlath et al. \cite{Gamlath} showed that no online algorithm has a competitive ratio larger than $1/2 + 1/(2d + 2)$ when the maximum degree is $d$, even on bipartite graphs. Thus, \cite{Gamlath} showed that no online algorithm can beat the greedy algorithm's competitive ratio $1/2$, even in bipartite graphs. These results hold for both fractional and integral matchings. 

In \cite{Buchbinder}, an algorithm with the competitive guarantee $2/3$ was provided for graphs with maximum degree two, which was shown to be optimal.   \cite{Buchbinder} showed that no online algorithm can achieve a competitive ratio larger than $4/(9-\sqrt 5)\approx 0.5914$ even on forests with maximum degree three in the vertex arrival model.  Further upper bounds were obtained  in~\cite{Epstein}, \cite{Huang19}.

In the edge arrival model for bipartite graphs, where all edges appear in $s$ batches~\cite{Lee} developed an algorithm with a guarantee $1/2+1/(2^{s+2}-2)$ for both integral and fractional matchings, where $s$ is the number of batches. For $s=2$, the competitive ratio becomes $2/3$, and it is also optimal. \cite{Guru} developed an online algorithm with a competitive ratio larger than the competitive ratio of the greedy algorithm 
for bipartite graphs and random uniform edge arrival orders. Online stochastic matchings with oblivious adversarial edge arrival order in bipartite graphs were studied in~\cite{Gravin}. In~\cite{Gravin}, an algorithm was developed that achieves a guarantee of $0.503$ in the above stochastic model, and they complement this result with an upper bound of $2/3$ on any achievable guarantee.

There was an extensive study of the edge arrival models under the assumption of free edge disposal, i.e. an already selected edge can be disposed of at later timepoints. For the weighted version of the problem, a deterministic algorithm with guarantee $1/(3+2\sqrt 2)$ was provided in~\cite{McGregor}; moreover, this guarantee was shown to be optimal among deterministic algorithms~\cite{Varadaraja}. Later, \cite{Epstein} provided a randomized algorithm with a guarantee of $0.1867$ for this model. There was further progress on upper bounds for possible guarantees of randomized algorithms in this model, see~\cite{Epstein}, \cite{Huang}.

The degree of the underlying graph was also studied in the context of online matching algorithms for rounding fractional matchings~\cite{Cohen}, \cite{Wajc}, \cite{Blikstad}.

For a comprehensive overview of results on online matchings, we refer the readers to the surveys~\cite{Mehta}, \cite{Devanur}, and to a recent survey~\cite{Huang}.

\subsection{Our Techniques}

Our online algorithm and its analysis demonstrate that the upper bound $4/(9-\sqrt 5)\approx 0.5914$ from~\cite{Buchbinder} is the best possible for graphs with maximum degree three. Our online algorithm is inspired by the construction from~\cite{Buchbinder}. Indeed, to obtain their upper bound, \cite{Buchbinder} construct instances such that every online algorithm with the guarantee $4/(9-\sqrt 5)\approx 0.5914$ on them should maintain a certain fractional matching. We refer to these instances as "consistent instances". We use the structure of the fractional matchings from~\cite{Buchbinder} on consistent instances as "building blocks" in our algorithm. We partition the edges from the consistent instances into two types of edges "path edges" and "spokes". Our algorithm attempts to greedily construct consistent instances from arriving edges, identifying some of the arrived edges as path edges and some as spokes. Naturally, the algorithm is not able to group all edges into consistent instances, and thus, we identify the remaining edges as "bridges".

For the path edges and spokes, our algorithm attempts to keep their values close to the values as in fractional matchings from~\cite{Buchbinder}. For the bridges, we need to consider several cases to carefully assign the value of the resulting fractional matching. 

Our algorithm keeps both a primal solution and a dual solution, i.e. it keeps both the values of a fractional matching and the values of a fractional vertex cover. The values of the fractional cover are used mainly for the analysis, and with the exception of spokes, the assignment of values for the fractional matching does not rely on them.

\section{Our Algorithm}
In this section, we first provide the intuition behind our algorithm that comes from the upper bound construction in~\cite{Buchbinder}. Afterwards, we state our algorithm in full detail and provide all necessary notions.

\subsection{Consistent Instances}\label{sec: consistent instances}

First, let us introduce a particular hard instance for graphs of maximum degree three which was constructed in~\cite{Buchbinder}.
    A \emph{consistent} instance with $n$ rounds contains the following edges
    \begin{enumerate}
        \item \label{consistent: first path edge} $e_1 = v^l_1 v^r_1$.
        \item \label{consistent: other path edges} $e^l_i = v^l_{i - 1} v^l_i$ and $e^r_i = v^r_{i - 1} v^r_i$ for $i = 2, \dots, n$.
        \item \label{consistent: spokes} $\hat{e}^l_i = v^l_i \hat{v}^l_i$ and $\hat{e}^r_i = v^r_i \hat{v}^r_i$ for $i = 1, \dots, n - 2$.
    \end{enumerate}
    Here, first the edge $e_1$ arrives, then with each further $i=2, \dots, n$ the edges $e^l_i$ and $e^r_i$ arrive. After that, the edges $\hat{e}^l_i$ and $\hat{e}^r_i$ arrive for $i = 1, \dots, n - 2$. See Figure~\ref{graph: paper instance n=4} for an example of a consistent instance with $n=4$.

    Buchbinder et al. show that any online algorithm for the fractional matching problem achieves the guarantee at most $c:=4/(9-\sqrt 5)\approx 0.5914$ on the consistent instances defined above. Since the consistent instances correspond to bipartite graphs, the same upper bound $c=4/(9-\sqrt 5)\approx 0.5914$ holds for the guarantee of online algorithms for integral matchings.

\begin{figure}
    \centering
    \begin{tikzpicture}[node distance={16mm}, main/.style = {draw, circle, thin}, line width=2pt] 
        \node[main] (1l) [text=black] {$v^l_1$}; 
        \node[main] (1r) [right of=1l] {$v^r_1$};
        \node[main] (2r) [right of=1r] {$v^r_2$}; 
        \node[main] (3r) [right of=2r] {$v^r_3$};
        \node[main] (4r) [right of=3r] {$v^r_4$};
        \node[main] (2l) [left of=1l] {$v^l_2$}; 
        \node[main] (3l) [left of=2l] {$v^l_3$};
        \node[main] (4l) [left of=3l] {$v^l_4$};
        
        \node[main] (1ls) [above of=1l] {$\hat{v}^l_1$};
        \node[main] (1rs) [above of=1r] {$\hat{v}^r_1$};
        \node[main] (2ls) [above of=2l] {$\hat{v}^l_2$};
        \node[main] (2rs) [above of=2r] {$\hat{v}^r_2$};
        
        \draw[my_green] (1l) -- node[midway, below, pos=0.5,text=black] {$e_1$} (1r);
        \draw[my_green] (1l) -- node[midway, below, pos=0.5,text=black] {$e^l_2$} (2l);
        \draw[my_green] (1r) -- node[midway, below, pos=0.5,text=black] {$e^r_2$} (2r);
        \draw[my_green] (2l) -- node[midway, below, pos=0.5,text=black] {$e^l_3$} (3l);
        \draw[my_green] (2r) -- node[midway, below, pos=0.5,text=black] {$e^r_3$} (3r);
        \draw[my_green] (3l) -- node[midway, below, pos=0.5,text=black] {$e^l_4$} (4l);
        \draw[my_green] (3r) -- node[midway, below, pos=0.5,text=black] {$e^r_4$} (4r);
        
        \draw[my_red, decorate,decoration=snake] (1l) -- node[midway, left, pos=0.5,text=black] {$\hat{e}^l_1$} (1ls);
        \draw[my_red, decorate,decoration=snake] (1r) -- node[midway, right, pos=0.5,text=black] {$\hat{e}^r_1$} (1rs);
        \draw[my_red, decorate,decoration=snake] (2l) -- node[midway, left, pos=0.5,text=black] {$\hat{e}^l_2$} (2ls);
        \draw[my_red, decorate,decoration=snake] (2r) -- node[midway, right, pos=0.5,text=black] {$\hat{e}^r_2$} (2rs);
    \end{tikzpicture}
    \caption{Example of a consistent instance with $n = 4$ rounds.} 
    \label{graph: paper instance n=4}
\end{figure}
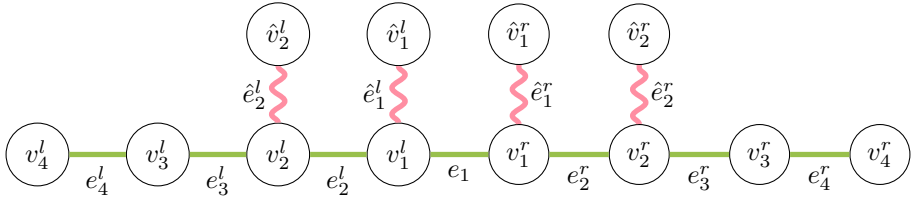

\subsection{Edge Types}
\label{sec: edge types}
Our algorithm tries to greedily construct consistent instances from the arriving edges. To do that systematically, we define three types of edges, so each edge is assigned one of these types upon its arrival:
\begin{itemize}
    \item type 1, \emph{path edges}.
    \item type 2, \emph{spokes}.
    \item type 3, \emph{bridges}.
\end{itemize}
Intuitively, the \emph{path edges} are the edges associated with~\ref{consistent: first path edge} and~\ref{consistent: other path edges} in the definition of consistent instances, while  \emph{spokes} are the edges associated with~\ref{consistent: spokes} in the definition of consistent instances; see Figure~\ref{graph: paper instance n=4}. Roughly speaking, \emph{bridges} are the edges that run between two different consistent instances that our algorithm constructed so far. We would like to note that bridges and spokes are more tricky objects for our algorithm than the above intuition may suggest.

To keep the exposition concise, we define the function $\type (\cdot)$. For each subset of edges $A\subseteq E$, the value  $\type (A)$ equals $(t_1, t_2, t_3)$, where $t_1$, $t_2,$ and $t_3$ equal the number of edges in $A$ of type 1, type 2, and type 3, respectively.

The most challenging case for us is to identify bridges and to assign them appropriate values. For these purposes, our algorithm is looking for special combinations of types when an edge $e=uv$ arrives. In particular our algorithms relies on the following set \[\left\{\type \left(\delta(u) \setminus \{e\}\right),\, \type \left(\delta(v) \setminus \{e\}\right)\right\}\] in the current graph $G$ after the edge $e$ arrived. If the above set lies in 
    \[\mathcal{B} := \{\{(1,0,0)\}, \{(1,0,0), (0,1,0)\}, \{(1,1,0), (1,0,0)\}, \{(0,2,0), (1,0,0)\}\}\]
then an arrived edge $e$ is considered to be a bridge by our algorithm.

For the sake of exposition, we also define the function $\Eends(\cdot)$. For each edge $e=uv\in E$, we have $\Eends(uv) := \{u,v\}$.

\subsection{Determining Types for the Arriving Edge}

Note that our algorithm assumes that the graph has maximum degree at most three at every timepoint. Let us describe how we assign the type to a newly arrived edge $e=uv$. Here, we work with the graph $G$ that refers to the graph after the arrival of the edge $e$. So, $\delta(u)$ stands for the edges incident to the vertex $u$ in $G$, and $\deg(u)$ stands for the degree of $u$ in $G$, etc. For each $f$ in $E\setminus\{e\}$, the value $y_f$ represents the value assigned to the edge $f$ in the fractional matching constructed before the arrival of $e$.

 Table~\ref{tab: edge-partitions} illustrates $\type(e)$ for the arriving edge $e=uv$ assigned by our algorithm. As follows from the table, the type of $e=uv$ depends on $\type\left(\delta(u) \setminus \{e\}\right)$ and $\type(\delta(v) \setminus \{e\})$. 
 
 There are several special cases. In Table~\ref{tab: edge-partitions} these cases are represented by enclosing the type of $e=uv$ in a box. In these special cases, both $u$ and $v$ have degree $3$ in the graph $G$, so without loss of generality, in these cases, we assume 
\[1 - \sum_{f \in \delta(u) \setminus \{e\}} y_f \leq 1 - \sum_{f \in \delta(v) \setminus \{e\}} y_f\,.\]

The empty cells in Table~\ref{tab: edge-partitions} correspond to impossible combinations of \\
$\type\left(\delta(u) \setminus \{e\}\right)$ and $\type(\delta(v) \setminus \{e\})$.

\begin{table}[H]
\centering
\caption{Type assignment for the newly arrived edge $e=uv$.
}
\setlength\extrarowheight{4pt}
\resizebox{\textwidth}{!}{
\begin{tabular}{|c|c|c|c|c|c|c|c|c|c|c|}
\hline
 \diagbox{$\type(\delta(u)\setminus \{e\})$}{$\type(\delta(v)\setminus \{e\})$} &
  \multicolumn{1}{c|}{(0,0,0)} &
  \multicolumn{1}{c|}{(1,0,0)} &
  \multicolumn{1}{c|}{(0,1,0)} &
  \multicolumn{1}{c|}{(0,0,1)} &
  \multicolumn{1}{c|}{(1,1,0)} &
  \multicolumn{1}{c|}{(1,0,1)} &
  \multicolumn{1}{c|}{(0,1,1)} &
  \multicolumn{1}{c|}{(2,0,0)} &
  \multicolumn{1}{c|}{(0,2,0)} &
  (0,0,2) \\[3pt] \hline
(0,0,0) &
  \cellcolor[HTML]{b9d485}1 &
  \cellcolor[HTML]{b9d485}1 &
  \cellcolor[HTML]{b9d485}1 &
  \cellcolor[HTML]{cbcbcb} &
  \cellcolor[HTML]{b9d485}1 &
  \cellcolor[HTML]{FFC1CC}2 &
  \cellcolor[HTML]{FFC1CC}2 &
  \cellcolor[HTML]{FFC1CC}2 &
  \cellcolor[HTML]{b9d485}1 &
  \cellcolor[HTML]{cbcbcb} \\[3pt] \hline
(1,0,0) &
  \cellcolor[HTML]{b9d485}1 &
  \cellcolor[HTML]{8dacc2}3 &
  \cellcolor[HTML]{8dacc2}3 &
  \cellcolor[HTML]{cbcbcb} &
  \cellcolor[HTML]{8dacc2}3 &
  \cellcolor[HTML]{FFC1CC}2 &
  \cellcolor[HTML]{FFC1CC}2 &
  \cellcolor[HTML]{FFC1CC}2 &
  \cellcolor[HTML]{8dacc2}3 &
  \cellcolor[HTML]{cbcbcb} \\[3pt] \hline
(0,1,0) &
  \cellcolor[HTML]{b9d485}1 &
  \cellcolor[HTML]{8dacc2}3 &
  \cellcolor[HTML]{b9d485}1 &
  \cellcolor[HTML]{cbcbcb} &
  \cellcolor[HTML]{b9d485}1 &
  \cellcolor[HTML]{FFC1CC}2 &
  \cellcolor[HTML]{FFC1CC}2 &
  \cellcolor[HTML]{FFC1CC}2 &
  \cellcolor[HTML]{b9d485}1 &
  \cellcolor[HTML]{cbcbcb} \\[3pt] \hline
(0,0,1) &
  \cellcolor[HTML]{cbcbcb} &
  \cellcolor[HTML]{cbcbcb} &
  \cellcolor[HTML]{cbcbcb} &
  \cellcolor[HTML]{cbcbcb} &
  \cellcolor[HTML]{cbcbcb} &
  \cellcolor[HTML]{cbcbcb} &
  \cellcolor[HTML]{cbcbcb} &
  \cellcolor[HTML]{cbcbcb} &
  \cellcolor[HTML]{cbcbcb} &
  \cellcolor[HTML]{cbcbcb} \\[3pt] \hline
(1,1,0) &
  \cellcolor[HTML]{b9d485}1 &
  \cellcolor[HTML]{8dacc2}3 &
  \cellcolor[HTML]{b9d485}1 &
  \cellcolor[HTML]{cbcbcb} &
  \cellcolor[HTML]{b9d485}1 &
  \cellcolor[HTML]{b9d485}{\color[HTML]{333333} \fbox{ \textbf{1}}} &
  \cellcolor[HTML]{b9d485}{\color[HTML]{333333} \fbox{ \textbf{1}}} &
  \cellcolor[HTML]{b9d485}{\color[HTML]{333333} \fbox{ \textbf{1}}} &
  \cellcolor[HTML]{b9d485}1 &
  \cellcolor[HTML]{cbcbcb} \\[3pt] \hline
(1,0,1) &
  \cellcolor[HTML]{FFC1CC}2 &
  \cellcolor[HTML]{FFC1CC}2 &
  \cellcolor[HTML]{FFC1CC}2 &
  \cellcolor[HTML]{cbcbcb} &
  \cellcolor[HTML]{FFC1CC}{\color[HTML]{333333} \fbox{ \textbf{2}}} &
  \cellcolor[HTML]{FFC1CC}2 &
  \cellcolor[HTML]{FFC1CC}2 &
  \cellcolor[HTML]{FFC1CC}2 &
  \cellcolor[HTML]{FFC1CC}{\color[HTML]{333333} \fbox{ \textbf{2}}} &
  \cellcolor[HTML]{cbcbcb} \\[3pt] \hline
(0,1,1) &
  \cellcolor[HTML]{FFC1CC}2 &
  \cellcolor[HTML]{FFC1CC}2 &
  \cellcolor[HTML]{FFC1CC}2 &
  \cellcolor[HTML]{cbcbcb} &
  \cellcolor[HTML]{FFC1CC}{\color[HTML]{333333} \fbox{ \textbf{2}}} &
  \cellcolor[HTML]{FFC1CC}2 &
  \cellcolor[HTML]{FFC1CC}2 &
  \cellcolor[HTML]{FFC1CC}2 &
  \cellcolor[HTML]{FFC1CC}{\color[HTML]{333333} \fbox{ \textbf{2}}} &
  \cellcolor[HTML]{cbcbcb} \\[3pt] \hline
(2,0,0) &
  \cellcolor[HTML]{FFC1CC}2 &
  \cellcolor[HTML]{FFC1CC}2 &
  \cellcolor[HTML]{FFC1CC}2 &
  \cellcolor[HTML]{cbcbcb} &
  \cellcolor[HTML]{FFC1CC}{\color[HTML]{333333} \fbox{ \textbf{2}}} &
  \cellcolor[HTML]{FFC1CC}2 &
  \cellcolor[HTML]{FFC1CC}2 &
  \cellcolor[HTML]{FFC1CC}2 &
  \cellcolor[HTML]{FFC1CC}{\color[HTML]{333333} \fbox{ \textbf{2}}} &
  \cellcolor[HTML]{cbcbcb} \\[3pt] \hline
(0,2,0) &
  \cellcolor[HTML]{b9d485}1 &
  \cellcolor[HTML]{8dacc2}3 &
  \cellcolor[HTML]{b9d485}1 &
  \cellcolor[HTML]{cbcbcb} &
  \cellcolor[HTML]{b9d485}{\color[HTML]{333333} 1} &
  \cellcolor[HTML]{b9d485}{\color[HTML]{333333} \fbox{ \textbf{1}}} &
  \cellcolor[HTML]{b9d485}{\color[HTML]{333333} \fbox{ \textbf{1}}} &
  \cellcolor[HTML]{b9d485}{\color[HTML]{333333} \fbox{ \textbf{1}}} &
  \cellcolor[HTML]{b9d485}{\color[HTML]{333333} 1} &
  \cellcolor[HTML]{cbcbcb} \\[3pt] \hline
(0,0,2) &
  \cellcolor[HTML]{cbcbcb} &
  \cellcolor[HTML]{cbcbcb} &
  \cellcolor[HTML]{cbcbcb} &
  \cellcolor[HTML]{cbcbcb} &
  \cellcolor[HTML]{cbcbcb} &
  \cellcolor[HTML]{cbcbcb} &
  \cellcolor[HTML]{cbcbcb} &
  \cellcolor[HTML]{cbcbcb} &
  \cellcolor[HTML]{cbcbcb} &
  \cellcolor[HTML]{cbcbcb} \\[3pt] \hline
\end{tabular}
}
\label{tab: edge-partitions}
\end{table}

\subsection{Fractional Matching in Consistent Instances}\label{sec:fractional matching in consistent instances}
Now that we have provided an idea of how our algorithm assigns a type to the arriving edge, let us provide a general idea of how we intend to construct a fractional matching.

The upper bound proof by Buchbinder et al. showed that for an online algorithm to achieve a competitive ratio on consistent instances defined in Section~\ref{sec: consistent instances}, the algorithm has to output a very specific fractional matching on these instances.

Let us define the values that should appear in the resulting fractional matching as per~\cite{Buchbinder}. For this, let us define the following values
\[
\widetilde{y}_1 := c,\, \widetilde{y}_2 := \frac{c}{2}\quad \text{and}\quad\widetilde{y}_3 := \frac{5c - 2}{2}
\]
and for natural $n$, $n\geq 4$ let us define
\[
\widetilde{y}_n := \frac{(3F_n + F_{n-2} -2)c - 2F_n + 2}{2}\,,
\]
where $\phi := \frac{1 + \sqrt{5}}{2}$, $\psi := 1 - \phi$ and $F_n := \frac{\phi^n - \psi^n}{\sqrt{5}}$. So $\phi$ is the golden ratio, and $F_n$ is the $n$-th Fibonacci number. We note the following useful property, the proof of which can be found in Appendix~\ref{appendix: fractional matching in consistent instances}, property \eqref{lma 1.7} of Lemma~\ref{lma fractional matching consistent instances}, for all natural $n$ we have $1 - \widetilde{y}_n - \widetilde{y}_{n+1} = c - \widetilde{y}_{n+2}$. In particular, we use $1 - \widetilde{y}_n - \widetilde{y}_{n+1}$ and $c - \widetilde{y}_{n+2}$ interchangeably.

The proof for the upper bound in~\cite{Buchbinder} showed that for the algorithm to achieve a competitive ratio $c$ on consistent instances, the algorithm needs to assign the following edge values (subject to symmetry breaking), see Section~\ref{sec: consistent instances}:
\begin{enumerate}
        \item  $e_1 = v^l_1 v^r_1$ has to be assigned $\widetilde{y}_1$.
        \item  $e^l_i = v^l_{i - 1} v^l_i$ and $e^r_i = v^r_{i - 1} v^r_i$ for $i = 2, \dots, n$ have to be assigned $\widetilde{y}_n$.
        \item $\hat{e}^l_i = v^l_i \hat{v}^l_i$ and $\hat{e}^r_i = v^r_i \hat{v}^r_i$ for $i = 1, \dots, n - 2$ have to be assigned $1-\widetilde{y}_n-\widetilde{y}_{n+1}$.
\end{enumerate}
Our algorithm tries to follow these value assignments on path edges and spokes, but the presence of bridges requires us to select more nuanced assignments even on path edges and spokes. In particular, we make more careful value assignments for edges that are bridges or are incident to a bridge upon their arrival.

\subsection{Algorithm}

We defer the formal definition of Algorithm~\ref{alg: one} to Appendix~\ref{appendix: algorithm and analysis}, which relies heavily on the primal-dual methodology.

Without loss of generality, we assume that for the arriving edge $e =uv$ that $\deg(u)\geq \deg(v)$.
Algorithm~\ref{alg: one} is a primal-dual algorithm. The algorithm produces a fractional matching by assigning each arrived edge $e\in E$ a nonnegative value $y_e$ such that at every timepoint and for every vertex $w$ we have $\sum_{f\in \delta(w)}y_f$ is at most $1$. Moreover, the algorithm produces a fractional vertex cover by keeping and updating nonnegative values $x_w$ for each vertex $w\in V$. At every timepoint and for every arrived edge $e=uv$, we have that $x_u+x_v$ is at least $c$. Thus, at every timepoint, the values $x_w$, $w\in V$ scaled by $1/c$ produce a fractional vertex cover for the graph $G$. 

\subsubsection{Value Assignments in Consistent Instances}

Thus, to prove that $c$ is indeed the guarantee of our online algorithm, it suffices to prove that at every timepoint the sum of all $y_e$, $e\in E$ equals the sum of all $x_w$, $w\in V$. Let us provide an intuition about how the algorithm preserves this equality by updating $x_w$, $w\in V$. Figure~\ref{fig: cover construction} demonstrates the increases in the values $x_w$, $w\in \left\{v^l_{i}, v^l_{i-1}, v^l_{i+1}, \hat v^l_{i}\right\}$ after the arrivals of their incident edges in the consistent instances from Section~\ref{sec: consistent instances}. For example, after the edge $v^l_{i}v^l_{i+1}$ arrives, the edge $v^l_{i}v^l_{i+1}$ gets value $\widetilde{y}_{i}$ and the values $x_w$, $w\in \left\{v^l_{i},  v^l_{i+1}\right\}$ are increased by $\widetilde{y}_{i}-(c-\widetilde{y}_{i+1})$ and $c-\widetilde{y}_{i+1}$, respectively. We would like to emphasize that Figure~\ref{fig: cover construction} depicts an ideal situation for constructing the fractional matchings and updating the fractional vertex cover. Our algorithm attempts to mimic this ideal behavior upon the arrival of path edges and spokes.

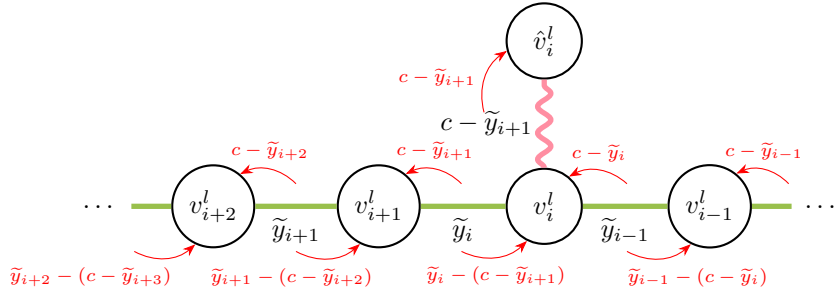
\begin{figure}
    \centering
    \begin{tikzpicture}[node distance={22mm}, main/.style = {draw, circle}] 
        \begin{scope}[every node/.style={circle, thick, draw, minimum size=10mm}, scale=.7]
            \node[main] (1l) {$v^l_{i}$}; 
            \node[main] (1r) [right of=1l] {$v^l_{i-1}$};
            \node[main] (2l) [left of=1l] {$v^l_{i+1}$}; 
            \node[main] (3l) [left of=2l] {$v^l_{i+2}$};
            \node[main] (1ls) [above of=1l] {$\hat v^l_{i}$};
        \end{scope}

        \begin{scope}[node distance={15mm}, every node/.style={fill=white,circle}, line width=2pt]
            \node (4l) [left of=3l] {$\dots$};
            \node (2r) [right of=1r] {$\dots$};

            \draw[my_green] (3l) -- (4l);
            \draw[my_green] (1r) -- (2r);
        \end{scope}

        \begin{scope}[>={Stealth[black]},
              every edge/.style={draw=black}, line width=2pt]
            \draw[my_green] (1l) -- node[midway, below, pos=0.5, text=black] {$\widetilde{y}_{i-1}$} (1r);
            \draw[my_green] (1l) -- node[midway, below, pos=0.5, text=black] {$\widetilde{y}_{i}$} (2l);
            \draw[my_green] (2l) -- node[midway, below, pos=0.5, text=black] {$\widetilde{y}_{i+1}$} (3l);
            
            \draw[my_red, decorate,decoration=snake] (1l) -- node[midway, left, pos=0.5, text=black] {$c - \widetilde{y}_{i+1}$} (1ls);

            % below label coords
            \coordinate[yshift=-6mm] (e1b) at ($(1l)!0.5!(1r)$);
            \coordinate[yshift=-6mm] (e2lb) at ($(1l)!0.5!(2l)$);
            \coordinate[yshift=-6mm] (e3lb) at ($(2l)!0.5!(3l)$);

            \coordinate[yshift=-6mm] (e4lb) at ($(3l)!0.7!(4l)$);

            %above label coords
            \coordinate[yshift=3mm] (e1u) at ($(1l)!0.5!(1r)$);
            \coordinate[yshift=3mm] (e2lu) at ($(1l)!0.5!(2l)$);
            \coordinate[yshift=3mm] (e3lu) at ($(2l)!0.5!(3l)$);

            \coordinate[yshift=3mm] (e0u) at ($(1r)!0.7!(2r)$);

            \coordinate[left=8mm of 1l] (e1lhat) at ($(1l)!0.5!(1ls)$);

        \end{scope}

        %dual construction
        \begin{scope}[>={Stealth[red]},
              every edge/.style={draw=red}]
  
            \draw[red] [->] (e1b.east) to ["$\widetilde{y}_{i-1} - (c - \widetilde{y}_i)$" {below, font=\fontsize{7}{0}\selectfont, shift={(0.5cm, 0)}}, bend right=35] (1r);
            \draw[red] [->] (e1u.west) to ["$c - \widetilde{y}_i$" {above, font=\fontsize{7}{0}\selectfont}, bend right=35] (1l);

            %left path edges dual updates
            \draw[red] [->] (e2lb.east) to ["$\widetilde{y}_{i} - (c - \widetilde{y}_{i+1})$" {below, font=\fontsize{7}{0}\selectfont}, bend right=35] (1l);
            \draw[red] [->] (e2lu.west) to ["$c - \widetilde{y}_{i+1}$" {above, font=\fontsize{7}{0}\selectfont}, bend right=35] (2l);

            \draw[red] [->] (e3lb.east) to ["$\widetilde{y}_{i+1} - (c - \widetilde{y}_{i+2})$" {below, font=\fontsize{7}{0}\selectfont, shift={(-0.5cm, 0)}}, bend right=35] (2l);
            \draw[red] [->] (e3lu.west) to ["$c - \widetilde{y}_{i+2}$" {above, font=\fontsize{7}{0}\selectfont}, bend right=35] (3l);

            \draw[red] [->] (e4lb.east) to ["$\widetilde{y}_{i+2} - (c - \widetilde{y}_{i+3})$" {below, font=\fontsize{7}{0}\selectfont, shift={(-1cm, 0)}}, bend right=35] (3l);

            \draw[red] [->] (e0u.west) to ["$c - \widetilde{y}_{i-1}$" {above, font=\fontsize{7}{0}\selectfont}, bend right=35] (1r);

            %left spokes
            \draw[red] [->] (e1lhat.north) +(0, 1.5mm) to ["$c - \widetilde{y}_{i+1}$" {left, font=\fontsize{7}{0}\selectfont}, bend left=35] (1ls);
        \end{scope}

    \end{tikzpicture}
    \caption{The construction of the fractional matching and increases in the fractional vertex cover, when Algorithm~\ref{alg: one} has a consistent instance as an input.}
    \label{fig: cover construction}
\end{figure}

\subsubsection{Position Indicators and Endpoints' distinction in Consistent Instances} \label{sec: position indicators}

Even in the ideal situation depicted in Figure~\ref{fig: cover construction}, to assign values to path edges, Algorithm~\ref{alg: one} relies on determining the exact position of such edges in the consistent instance. To make sure that Algorithm~\ref{alg: one} has access to these positions, we keep a position indicator $n_f$ for each path edge $f$ in the graph. For example, in a consistent instance when $f=v^l_{i}v^l_{i+1}$ we have $n_f=i+1$; and when $f=v^l_{1}v^r_{i}$ we have $n_f=1$.

Already in the consistent instances, the endpoints of the arriving edges could have different properties. For example, if an arriving edge is identified as a path edge in a consistent instance, then one of its endpoints has degree one and the other endpoint has degree two (immediately after this arrival). To keep track of these different properties, for some edges $e$, the algorithm identifies one of the endpoints as $z(e)$ and another endpoint as $w(e)$. In particular, this is crucial for the analysis of path edges and spokes. For example, in Figure~\ref{fig: cover construction} for the path edge $f=v^l_iv^l_{i+1}$ we have $z(f)=v^l_i$ and $w(f)=v^l_{i+1}$, similarly for the path edge $f=v^l_{i+1}v^l_{i+2}$ we have $z(f)=v^l_{i+1}$ and $w(f)=v^l_{i+2}$. For the spoke $f=v^l_i\hat{v}^l_i$ in Figure~\ref{fig: cover construction}, we have $z(f)=v^l_i$ and $w(f)=\hat{v}^l_i$.

\subsubsection{Partition into Consistent Instances and Bridges}

With the cover construction depicted in Figure~\ref{fig: cover construction} in mind, Algorithm~\ref{alg: one} utilizes this construction by loosely partitioning the arriving edges into subgraphs of the consistent instance, with some exceptions; namely, bridges. The purpose of bridges is to connect, not necessarily distinct, partitions. To see how the fractional cover changes, consider, for example, the updates in the fractional cover upon the arrival of a path edge in lines~\ref{line: x path edge z} and \ref{line: x path edge w} of Algorithm~\ref{alg: one}. These updates are identical to the ones depicted in~\ref{fig: cover construction} except that in certain cases we cannot use the value~$\widetilde{y}_{n_e}$ from the ideal case depicted in Figure~\ref{fig: cover construction} but we have to use the actually assigned value~$y_e$. In a similar way, we can see the updates in the fractional cover upon the arrival of a spoke, see line~\ref{line:  spoke cover change} of Algorithm~\ref{alg: one}.

Recall from Section~\ref{sec: edge types} that an edge $e=uv$ is assigned to be a bridge if and only if $\{\delta(u)\setminus\{e\}, \, \delta(v)\setminus\{e\}\}$ is in the set $\mathcal B$.
Consider Figure~\ref{fig: potential edge partitions}, where we assume that all edges but $b_1$, $b_2$, $b_3$, $b_4$, $b_5$, $s_1$, $s_2$
 arrive first in some specific order, and then the edges $b_1$, $b_2$, $b_3$, $b_4$, $b_5$ arrive, and then $s_1$, $s_2$ arrive. Now, before the arrival of $b_5$, both of the endpoints of $b_5$ are incident only to one path edge each. Thus since $\{(1,0,0)\}$ is in $\mathcal B$, the edge $b_5$ is assigned to be a bridge.

Even though it is helpful to think about Algorithm~\ref{alg: one} as an algorithm partitioning the arrived path edges and spokes into consistent instances, this intuition does not always extend to spokes. In particular, some of the edges are assigned to be a spoke even though they do not "naturally fit" in any consistent instance. For example, before $s_1$ arrives, one of its endpoints is incident to a bridge $b_1$ and a path edge, while the other endpoint is incident to no edges. Due to the presence of the bridge $b_1$, the edge $s_1$ is not assigned to be a path edge but to be a spoke by Algorithm~\ref{alg: one}. Generally, an edge becomes such a spoke when, after its arrival, precisely one of its endpoints has degree $3$ and is incident to a bridge.

%Partitions
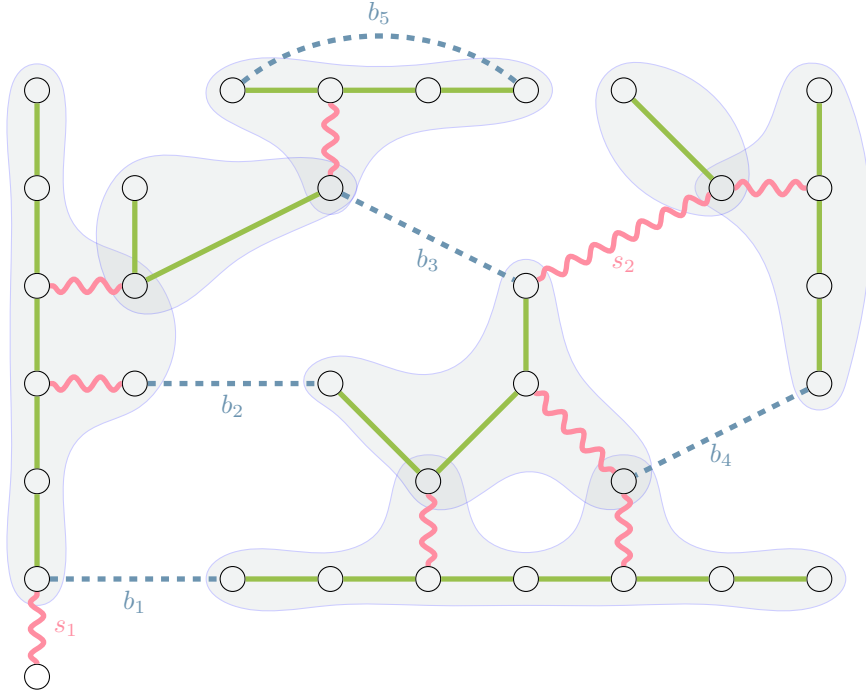
\begin{figure}
    \centering
    \begin{tikzpicture}[node distance={13mm}, main/.style = {draw, circle, thin}, line width=2pt, glow/.style={%
    preaction={draw,line cap=round,line join=round,
    opacity=0.4,line width=10pt,#1}},glow/.default=yellow,
    transparency group, yscale=.8] 
        %partition 1
        \node[main] (v1_p1) {}; 
        \node at (v1_p1) [name=fake_v1_p1,outer sep=5pt,inner sep=5pt]{};
        \node[main] (v2_p1) [above of=v1_p1] {};
        \node at (v2_p1) [name=fake_v2_p1,outer sep=5pt,inner sep=5pt]{};
        \node[main] (v3_p1) [above of=v2_p1] {};
        \node at (v3_p1) [name=fake_v3_p1,outer sep=5pt,inner sep=5pt]{};
        \node[main] (v4_p1) [above of=v3_p1] {};
        \node at (v4_p1) [name=fake_v4_p1,outer sep=5pt,inner sep=5pt]{};
        \node[main] (v5_p1) [above of=v4_p1] {};
        \node at (v5_p1) [name=fake_v5_p1,outer sep=5pt,inner sep=5pt]{};
        \node[main] (v6_p1) [above of=v5_p1] {};
        \node at (v6_p1) [name=fake_v6_p1,outer sep=5pt,inner sep=5pt]{};

        \node[main] (v7_p1) [right of=v3_p1] {};
        \node at (v7_p1) [name=fake_v7_p1,outer sep=5pt,inner sep=5pt]{};
        \node[main] (v8_p1) [right of=v4_p1] {};
        \node at (v8_p1) [name=fake_v8_p1,outer sep=5pt,inner sep=5pt]{};
            
        \draw[my_green] (v1_p1) -- node[midway, below, pos=0.5] {} (v2_p1);
        \draw[my_green] (v2_p1) -- node[midway, below, pos=0.5] {} (v3_p1);
        \draw[my_green] (v3_p1) -- node[midway, below, pos=0.5] {} (v4_p1);
        \draw[my_green] (v4_p1) -- node[midway, below, pos=0.5] {} (v5_p1);
        \draw[my_green] (v5_p1) -- node[midway, below, pos=0.5] {} (v6_p1);

        \draw[my_red, decorate,decoration=snake] (v7_p1) -- node[midway, below, pos=0.5] {} (v3_p1);
        \draw[my_red, decorate,decoration=snake] (v8_p1) -- node[midway, below, pos=0.5] {} (v4_p1);

        %partition 2
        \node[main, draw=none] (temp_p2) [right of=v1_p1] {}; 
        \node[main] (v1_p2) [right of=temp_p2] {};
        \node at (v1_p2) [name=fake_v1_p2,outer sep=5pt,inner sep=5pt]{};
        \node[main] (v2_p2) [right of=v1_p2] {};
        \node at (v2_p2) [name=fake_v2_p2,outer sep=5pt,inner sep=5pt]{};
        \node[main] (v3_p2) [right of=v2_p2] {};
        \node at (v3_p2) [name=fake_v3_p2,outer sep=5pt,inner sep=5pt]{};
        \node[main] (v4_p2) [right of=v3_p2] {};
        \node at (v4_p2) [name=fake_v4_p2,outer sep=5pt,inner sep=5pt]{};
        \node[main] (v5_p2) [right of=v4_p2] {};
        \node at (v5_p2) [name=fake_v5_p2,outer sep=5pt,inner sep=5pt]{};
        \node[main] (v7_p2) [right of=v5_p2] {};
        \node at (v7_p2) [name=fake_v7_p2,outer sep=5pt,inner sep=5pt]{};
        \node[main] (v8_p2) [right of=v7_p2] {};
        \node at (v8_p2) [name=fake_v8_p2,outer sep=5pt,inner sep=5pt]{};

        \node[main] (v6_p2) [above of=v3_p2] {};
        \node at (v6_p2) [name=fake_v6_p2,outer sep=5pt,inner sep=5pt]{};

        \node[main] (v9_p2) [above of=v5_p2] {};
        \node at (v9_p2) [name=fake_v9_p2,outer sep=5pt,inner sep=5pt]{};

        \draw[my_green] (v1_p2) -- node[midway, below, pos=0.5] {} (v2_p2);
        \draw[my_green] (v2_p2) -- node[midway, below, pos=0.5] {} (v3_p2);
        \draw[my_green] (v3_p2) -- node[midway, below, pos=0.5] {} (v4_p2);
        \draw[my_green] (v4_p2) -- node[midway, below, pos=0.5] {} (v5_p2);
        \draw[my_green] (v7_p2) -- node[midway, below, pos=0.5] {} (v5_p2);
        \draw[my_green] (v8_p2) -- node[midway, below, pos=0.5] {} (v7_p2);

        \draw[my_red, decorate,decoration=snake] (v6_p2) -- node[midway, below, pos=0.5] {} (v3_p2);

        \draw[my_red, decorate,decoration=snake] (v9_p2) -- node[midway, below, pos=0.5] {} (v5_p2);
        
        %partition 3
        \node[main] (v1_p3) [left of=v6_p2, above of=v6_p2] {};
        \node at (v1_p3) [name=fake_v1_p3,outer sep=5pt,inner sep=5pt]{};
        \node[main] (v2_p3) [right of=v6_p2, above of=v6_p2] {};
        \node at (v2_p3) [name=fake_v2_p3,outer sep=5pt,inner sep=5pt]{};
        \node[main] (v3_p3) [above of=v2_p3] {};
        \node at (v3_p3) [name=fake_v3_p3,outer sep=5pt,inner sep=5pt]{};

        \draw[my_green] (v1_p3) -- node[midway, below, pos=0.5] {} (v6_p2);
        \draw[my_green] (v6_p2) -- node[midway, below, pos=0.5] {} (v2_p3);
        \draw[my_green] (v2_p3) -- node[midway, below, pos=0.5] {} (v3_p3);

        \draw[my_red, decorate,decoration=snake] (v2_p3) -- node[midway, below, pos=0.5] {} (v9_p2);

        %partition 4
        \node[main, draw=none] (temp_p4) [right of=v6_p1] {}; 
        \node[main] (v1_p4) [right of=temp_p4] {}; 
        \node at (v1_p4) [name=fake_v1_p4,outer sep=5pt,inner sep=5pt]{};
        \node[main] (v2_p4) [right of=v1_p4] {}; 
        \node at (v2_p4) [name=fake_v2_p4,outer sep=5pt,inner sep=5pt]{};
        \node[main] (v3_p4) [right of=v2_p4] {}; 
        \node at (v3_p4) [name=fake_v3_p4,outer sep=5pt,inner sep=5pt]{};
        \node[main] (v4_p4) [right of=v3_p4] {}; 
        \node at (v4_p4) [name=fake_v4_p4,outer sep=5pt,inner sep=5pt]{};

        \node[main] (v5_p4) [below of=v2_p4] {}; 
        \node at (v5_p4) [name=fake_v5_p4,outer sep=5pt,inner sep=5pt]{};

        \draw[my_green] (v1_p4) -- node[midway, below, pos=0.5] {} (v2_p4);
        \draw[my_green] (v2_p4) -- node[midway, below, pos=0.5] {} (v3_p4);
        \draw[my_green] (v3_p4) -- node[midway, below, pos=0.5] {} (v4_p4);

        \draw[my_red, decorate,decoration=snake] (v2_p4) -- node[midway, below, pos=0.5] {} (v5_p4);

        %partition 5

        \node[main] (v1_p5) [above of=v8_p1] {};
        \node at (v1_p5) [name=fake_v1_p5,outer sep=5pt,inner sep=5pt]{};
        
        \draw[my_green] (v8_p1) -- node[midway, below, pos=0.5] {} (v5_p4);
        \draw[my_green] (v8_p1) -- node[midway, below, pos=0.5] {} (v1_p5);

        %partition 6
        \node[main, draw=none] (temp_p6) [above of=v8_p2] {};
        \node[main] (v1_p6) [above of=temp_p6] {};
        \node at (v1_p6) [name=fake_v1_p6,outer sep=5pt,inner sep=5pt]{};
        \node[main] (v2_p6) [above of=v1_p6] {};
        \node at (v2_p6) [name=fake_v2_p6,outer sep=5pt,inner sep=5pt]{};
        \node[main] (v3_p6) [above of=v2_p6] {};
        \node at (v3_p6) [name=fake_v3_p6,outer sep=5pt,inner sep=5pt]{};
        \node[main] (v4_p6) [above of=v3_p6] {};
        \node at (v4_p6) [name=fake_v4_p6,outer sep=5pt,inner sep=5pt]{};

        \node[main] (v5_p6) [left of=v3_p6] {};
        \node at (v5_p6) [name=fake_v5_p6,outer sep=5pt,inner sep=5pt]{};

        \draw[my_green] (v1_p6) -- node[midway, below, pos=0.5] {} (v2_p6);
        \draw[my_green] (v2_p6) -- node[midway, below, pos=0.5] {} (v3_p6);
        \draw[my_green] (v3_p6) -- node[midway, below, pos=0.5] {} (v4_p6);

        \draw[my_red, decorate,decoration=snake] (v5_p6) -- node[midway, below, pos=0.5] {} (v3_p6);

        %partition 7
        \node[main] (v1_p7) [above of=v5_p6, left of=v5_p6] {};
        \node at (v1_p7) [name=fake_v1_p7,outer sep=5pt,inner sep=5pt]{};

        \draw[my_green] (v1_p7) -- node[midway, below, pos=0.5] {} (v5_p6);

        %Bridges
        \begin{scope}   [dashed]
            \draw[my_blue] (v1_p1) -- node[midway, below, pos=0.5] {$b_1$} (v1_p2);
            \draw[my_blue] (v7_p1) -- node[midway, below, pos=0.5] {$b_2$} (v1_p3);
            \draw[my_blue] (v5_p4) -- node[midway, below, pos=0.5] {$b_3$} (v3_p3);
            \draw[my_blue] (v1_p6) -- node[midway, below, pos=0.5] {$b_4$} (v9_p2);

            \draw[my_blue] (v1_p4) to[out=45,in=135] node[at start, above, pos=0.5] {$b_5$} (v4_p4);
        \end{scope}
        
        %spokes ~ bridges
        \node[main] (spoke_1) [below of=v1_p1] {};

        \draw[my_red, decorate,decoration=snake] (spoke_1) -- node[midway, right, xshift=2pt, pos=0.5] {$s_1$} (v1_p1);
        \draw[my_red, decorate,decoration=snake] (v5_p6) -- node[midway, below, yshift=-3pt, pos=0.5] {$s_2$} (v3_p3);

        %Partitions
        \begin{pgfonlayer}{background}
            % partition 1
            \draw[blue,fill=my_grey,opacity=0.2](fake_v1_p1.west) 
                to[closed,curve through={
                (fake_v2_p1.west) .. 
                (fake_v3_p1.west) ..
                (fake_v4_p1.west) ..
                (fake_v5_p1.west) ..
                (fake_v6_p1.west) ..
                (fake_v6_p1.north) ..
                (fake_v6_p1.east) ..
                (fake_v5_p1.east) ..
                ($(v5_p1.south east)!0.5!(v8_p1.north west)$) ..
                (fake_v8_p1.north) ..
                (fake_v8_p1.east) ..
                (fake_v7_p1.east) ..
                (fake_v7_p1.south) ..
                ($(v7_p1.south west)!0.5!(v2_p1.north east)$) ..
                (fake_v2_p1.east) ..
                (fake_v1_p1.east) ..
                (fake_v1_p1.south) 
                }] (fake_v1_p1.west);
            
            % partition 2
            \draw[blue,fill=my_grey,opacity=0.2](fake_v1_p2.north) 
                to[closed,curve through={
                (fake_v2_p2.north) .. 
                ($(v2_p2.north east)!0.5!(v6_p2.south west)$) ..
                (fake_v6_p2.north) ..
                ($(v6_p2.south west)!0.5!(v4_p2.north east)$) ..
                (fake_v4_p2.north) ..
                ($(v4_p2.north east)!0.5!(v9_p2.south west)$) ..
                (fake_v9_p2.north) ..
                ($(v9_p2.south west)!0.5!(v7_p2.north east)$) ..
                (fake_v7_p2.north) ..
                (fake_v8_p2.north) ..
                (fake_v8_p2.east) ..
                (fake_v8_p2.south) ..
                (fake_v7_p2.south) ..
                (fake_v5_p2.south) ..
                (fake_v4_p2.south) ..
                (fake_v3_p2.south) ..
                (fake_v2_p2.south) ..
                (fake_v1_p2.south) ..
                (fake_v1_p2.west) 
                }] (fake_v1_p2.north);

            % partition 3
            \draw[blue,fill=my_grey,opacity=0.2](fake_v1_p3.west) 
                to[closed,curve through={
                (fake_v1_p3.north) ..
                ($(v1_p3.north east)!0.5!(v2_p3.north west)$) ..
                (fake_v3_p3.west) ..
                (fake_v3_p3.north) ..
                (fake_v3_p3.east) ..
                ($(v3_p3.east)!0.5!(v9_p2.north)$) ..
                (fake_v9_p2.east) ..
                (fake_v9_p2.south) ..
                (fake_v9_p2.south west) ..
                ($(v6_p2.north east)!0.5!(v9_p2.north east)$) ..
                (fake_v6_p2.south east) ..
                (fake_v6_p2.south) ..
                (fake_v6_p2.west) ..
                ($(fake_v6_p2.west)!0.5!(fake_v1_p3.south)$) ..
                (fake_v1_p3.south)
                }] (fake_v1_p3.west);

            % partition 4
            \draw[blue,fill=my_grey,opacity=0.2](fake_v1_p4.west) 
                to[closed,curve through={
                (fake_v1_p4.north) ..
                (fake_v2_p4.north) ..
                (fake_v3_p4.north) ..
                (fake_v4_p4.north) ..
                (fake_v4_p4.east) ..
                (fake_v4_p4.south) ..
                ($(v3_p4.south)!0.5!(v5_p4.east)$) ..
                (fake_v5_p4.east) ..
                (fake_v5_p4.south) ..
                (fake_v5_p4.west) ..
                ($(v5_p4.west)!0.5!(v1_p4.south)$) ..
                }] (fake_v1_p4.west);

            % partition 5
            \draw[blue,fill=my_grey,opacity=0.2](fake_v5_p4.east) 
                to[closed,curve through={
                (fake_v5_p4.south) ..
                ($(fake_v5_p4.south)!0.5!(fake_v8_p1.south)$) ..
                (fake_v8_p1.south east) ..
                (fake_v8_p1.south) .. 
                (fake_v8_p1.west) ..
                (fake_v1_p5.west) ..
                (fake_v1_p5.north) ..
                ($(fake_v1_p5.north)!0.5!(fake_v5_p4.north)$) ..
                (fake_v5_p4.north)
                }] (fake_v5_p4.east);

            % partition 6
            \draw[blue,fill=my_grey,opacity=0.2](fake_v1_p6.south) 
                to[closed,curve through={
                (fake_v1_p6.east) ..
                (fake_v4_p6.east) ..
                (fake_v4_p6.north) ..
                ($(v4_p6.west)!0.5!(v5_p6.north)$) ..
                (fake_v5_p6.north) ..
                (fake_v5_p6.west) ..
                (fake_v5_p6.south) ..
                (fake_v1_p6.west)
                }] (fake_v1_p6.south);

            % partition 7
            \draw[blue,fill=my_grey,opacity=0.2](fake_v1_p7.north) 
                to[closed,curve through={
                (fake_v1_p7.west) ..
                (fake_v1_p7.south west) ..
                (fake_v5_p6.south west) ..
                (fake_v5_p6.south) ..
                (fake_v5_p6.north east) ..
                (fake_v1_p7.north east)
                }] (fake_v1_p7.north);
                
        \end{pgfonlayer}
        
    \end{tikzpicture}
    \caption{An example of a potential partition indirectly maintained by Algorithm~\ref{alg: one} (subject to arrival order) into subgraphs of the consistent instances. Here, the straight green edges represent path edges, the wavy red edges represent spokes, and the dashed blue edges represent bridges.}
    \label{fig: potential edge partitions}
\end{figure}

\subsubsection{The Difficulty of Bridges} \label{sec: difficulty of bridges}

Bridges are divided into four classes determined by the incident edges to their endpoints, see the definition based on $\mathcal{B}$. Intuitively, an edge $e$ is assigned a bridge when both endpoints are already in consistent instances in the current partition, and it is not clear to which consistent instance the edge $e$ should be added. In this case, the edge $e$ is assigned a bridge and $e$ attempts to fulfill the role it would be given, as if it were assigned to each partition individually. For instance, in Figure~\ref{fig: potential edge partitions} as the bridge $b_1$ arrives it is not immediately clear which consistent instance $b_1$ should join. In this case, $b_1$ prevents the paths' ``growth" in these two consistent instances beyond the endpoints of $b_1$. After the arrival of $b_1$, the algorithm needs to account for the possibility of future edges incident to $b_1$, and to do that the algorithm needs to update $x$ and $y$ appropriately. The main obstacle for finding an appropriate $x$ and $y$ update is the possibility of future spokes incident to $b_1$. Since, as explained above, $b_1$ prevented the paths' ``growth" in two consistent instances, we might want the $x$ update to happen as in the case where $b_1$ is just a new path edge in both consistent instances. However, this is not always possible. For instance, consider the case where the two path edges adjacent to $b_1$ are assigned $\widetilde{y}_1 = c$ and $\widetilde{y}_2 = c/2$, that is, the first and second edges in their respective paths. So if we were to adhere to the structure in Figure~\ref{fig: cover construction}, we would require the dual solution $x$ to increase by $\widetilde{y}_2 - (c - \widetilde{y}_3)$ for one of the endpoints and $\widetilde{y}_3 - (c - \widetilde{y}_4)$ for the other. However, to do so the assignment to $b_1$ would have to be at least 
\begin{equation*}
    \widetilde{y}_2 - (c - \widetilde{y}_3) + \widetilde{y}_3 - (c - \widetilde{y}_4) = \frac{c}{2} - (c - \frac{5c - 2}{2}) + \frac{5c-2}{2} - \left(c - (4c-2)\right) = \frac{15c}{2} - 4\,.
\end{equation*}
However, as $c + \frac{15c}{2} - 4 > 1$, this assignment to $b_1$ is not feasible. To circumvent this, the algorithm capitalizes on the structure of the future edges incident to $b_1$, and this structure allows the algorithm to assign $b_1$ a substantially smaller value than $c + \frac{15c}{2} - 4$.

Similarly, one can handle the case when one of the endpoints of an arriving bridge is incident to a previously arrived spoke. This case can be seen in bridges $b_2$, $b_3$, and $b_4$ in Figure~\ref{fig: potential edge partitions}. Let $u$ be the common endpoint of the arriving bridge and existing spoke. In this situation, our algorithm guarantees that the value of $x_u$ after the arrival of $b_1$ is at least the value of $x_u$ at the moment when the spoke arrived. This ensures that $x$ remains a feasible dual solution; in particular, that the corresponding constraint for the spoke is satisfied by $x$ even after the arrival of $b_1$.

\section{Algorithm Definition and Main Properties}\label{appendix: algorithm and analysis}

\begin{algorithm}
    \caption{Online Algorithm for Maximum Cardinality Fractional Matchings in Graphs of Maximum Degree Three}\label{alg: one}
    \begin{algorithmic}[1]
    \State While $e = uv$ arrives (Assume $\degree(v) \leq \degree(u)$)
        
    \If{$\{\type(\delta(u) \setminus \{e\}), \type(\delta(v) \setminus \{e\})\} \in \mathcal{B}$}
        \State $\type(e) \gets 3$
        \If{$\{\type(\delta(u) \setminus \{e\}), \type(\delta(v) \setminus \{e\})\} = \{(1,1,0), (1,0,0)\}$}
            \State Let $f_1, f_2 \in \delta(u)$, $\type(f_1) = 1, \type(f_2) = 2$, and $f_3 \in \delta(v)$, $\type(f_3) = 1$
            \State \label{line: b4 y assigned} $y_e \gets \max\{\widetilde{y}_{n_{f_3} + 1} - y_{f_2} - \min\{c - \widetilde{y}_{n_{f_1} + 1}, y_{f_2}\}, 0\}$
            \State \label{line: b4 u cover change} $x_u \gets x_u + y_e - \max\{\widetilde{y}_{n_{f_3} + 1} - y_{f_2}, 0\}$
            \State \label{line: b4 v cover change} $x_v \gets x_v + \max\{\widetilde{y}_{n_{f_3} + 1} - y_{f_2}, 0\}$
        \ElsIf{$\{\type(\delta(u) \setminus \{e\}), \type(\delta(v) \setminus \{e\})\} = \{(0,2,0), (1,0,0)\}$}
            \State Let $f_v \in \delta(v)  
            \setminus \{e\}$ 
            \State \label{line: b3 y assigned} $y_e \gets \max\{\widetilde{y}_{n_{f_v} + 1} - \max\{y_f \mid f \in \delta(u) \setminus \{e\}\}, 0\}$
            \State \label{line: b3 v cover change} $x_v \gets x_v + y_e$
        \ElsIf{$\{\type(\delta(u) \setminus \{e\}), \type(\delta(v) \setminus \{e\})\} = \{(1,0,0)\}$}
            \State \label{line: b1 f_u, f_v} Let $f_u \in \delta(u) \setminus \{e\}$, $f_v \in \delta(v) \setminus \{e\}$
            \State \label{line: b1 z, w} Let $z(e) \in \Eends(e)$ s.t $\widetilde{y}_{n_{f_{z(e)}} + 1} = \min\{\widetilde{y}_{n_{f} + 1}, \mid f \in \{f_u, f_v\}\}$, and $w(e) \in \Eends(e) \setminus \{z(e)\}$
            \State \label{line: b1 y assigned} $y_e \gets \widetilde{y}_{n_{f_{z(e)}} + 1} - (c - \widetilde{y}_{n_{f_{w(e)}} + 1})$
            \State \label{line: b1 z cover change} $x_{z(e)} \gets x_{z(e)} + \widetilde{y}_{n_{f_{z(e)}} + 1} - \min\{\frac{c}{2}, 1 - y_{f_{z(e)}} - y_{e}\}$
            \State \label{line: b1 w cover change} $x_{w(e)} \gets x_{w(e)} + \widetilde{y}_{n_{f_{w(e)}} + 1} - \max\{\frac{c}{2}, c - ( 1 - y_{f_{z(e)}} - y_{e})\}$ 
        \ElsIf{$\{\type(\delta(u) \setminus \{e\}), \type(\delta(v) \setminus \{e\})\} = \{(1,0,0), (0,1,0)\}$}
            \State \label{line: b2 definitions} $f_1, f_2 \in \delta(\Eends(e)), \type(f_i) = i, z(e) \in \Eends(e) \cap \Eends(f_1), w(e) \in \Eends(e) \cap \Eends(f_2)$
            \State \label{line: b2 y assigned} $y_e \gets \max\{\widetilde{y}_{n_{f_1} + 1} - y_{f_2}, 0\}$
            \State \label{line: b2 z cover change} $x_{z(e)} \gets x_{z(e)} + y_e - \max\{(2c-1) - y_{f_2}, 0\}$
            \State \label{line: b2 w cover change} $x_{w(e)} \gets x_{w(e)} + \max\{(2c-1) - y_{f_2}, 0\}$ 
        \EndIf
    \Else 
        \State \[ (z(e), w(e)) \gets \begin{cases}
                (v, u) & \mbox{if $ \degree(v) == 3$ and $1 - \sum_{f \in \delta(v) \setminus \{e\}} y_f < 1 - \sum_{f \in \delta(u) \setminus \{e\}} y_f $}\label{line: choice of z}\\
                (u, v) & \mbox{\text{otherwise }}
            \end{cases}\]
        \If{$\degree(z(e)) == 3$ and $\type(\delta(z(e)) \setminus \{e\}) \notin \{(0,2,0), (1,1,0)\}$}\label{line: test spoke edge}
            \State $y_e \gets c - x_{z(e)}$ \label{line: y spoke edge}
            \State $\type(e) \gets 2$ \label{line: spoke edge assigned}
            \State $x_{w(e)} \gets x_{w(e)} + y_e$ \label{line:  spoke cover change}
        \Else
            \State \label{line: def of n} $n_e \gets \max\left\{\{n_f + 1 \mid f \in \delta(z(e) \setminus \{e\}), \type(f) = 1\} \cup \{1\} \right\}$
            \State \label{line: y path edge} $y_e \gets \min\{\widetilde{y}_{n_e}, 1 - \sum_{f \in \delta(z(e)) \setminus \{e\}} y_f\}$
            \State \label{line: path edge assigned} $\type(e) \gets 1$
            \State \label{line: x path edge z} $x_{z(e)} \gets x_{z(e)} + y_e - (c - \widetilde{y}_{n_e + 1})$
            \State \label{line: x path edge w} $x_{w(e)} \gets x_{w(e)} + (c - \widetilde{y}_{n_e + 1})$
        \EndIf
    \EndIf
    \end{algorithmic}
\end{algorithm}

Having stated Algorithm~\ref{alg: one}, we need to demonstrate two things: that the algorithm outputs a feasible fractional matching and that the algorithm achieves the desired guarantee~$c$. To accomplish this, we prove that at every timepoint the results of computations satisfy all of the properties stated in the next lemma. The crucial properties for the correctness and guarantee are the properties~\ref{P1}, \ref{P2}, and \ref{P7} from the below lemma, while other key properties for technical arguments are stated in Appendix~\ref{sec: proof properties}. We defer the proof of Lemma~\ref{lma main lemma} to Appendix~\ref{sec: proof properties}. 

\begin{lemma}[Main Properties] \label{lma main lemma}
    Let the underlying graph have maximum degree three. At every timepoint the values $y$ and $x$ computed by Algorithm~\ref{alg: one} satisfy the following properties:
    \begin{description}[style=multiline]
        \item[\namedlabel{P1}{P1}] $\sum_{u \in V} x_u = \sum_{e \in E} y_e$.
        \item[\namedlabel{P2}{P2}] for all $e = uv \in E$ we have $x_u + x_v \geq c$.
        \item[\namedlabel{P7}{P3}] for all $u \in V$ we have $\sum_{f \in \delta(u)} y_f \leq 1$.
    \end{description}
\end{lemma}

In particular, we accomplish property~\ref{P1} by adhering to the following rule: after deciding on the value $y_e$ assigned to an arriving edge $e = uv$, we update the cover solution $x$ only at the endpoints of $e$, i.e. we update only $x_u$ and $x_v$. We require the increase of $x_u+x_v$ to be precisely $y_e$. Thus, we can view it as assigning $e$ a value $y_e$, and after that distributing the value $y_e$ to the endpoints of $e=uv$. Note that we do not always increase both $x_u$ and $x_v$; indeed, $x_u$ or $x_v$ can even decrease as long as $x_u+x_v$ is increased by $y_e$.

In Section~\ref{sec: nonnegativity}, we prove that the values $x_u$, $u\in V$ and $y_e, e\in E$ are nonnegative, see Lemma~\ref{lma x non-neg} and Lemma~\ref{lma y non-neg} below. Note, that nonnegativity of $y_e$, $e\in E$ together with the property~\ref{P7} implies that $y_e$, $e\in E$ is a feasible fractional matching. Nonnegativity of $x_u$, $u\in V$ together with the property~\ref{P2} implies that $x_u/c$, $u\in V$ is a feasible fractional cover. Finally, having a feasible fractional cover $x_u/c$, $u\in V$, and a feasible fractional matching $y_e$, $e\in E$, the property~\ref{P1} shows that Algorithm~\ref{alg: one} achieves the guarantee $c$ as desired, leading us to the following theorem.

\begin{theorem}
    For fractional matchings in the adversarial edge arrival model, Algorithm~\ref{alg: one} achieves the guarantee $c$ on graphs of maximum degree three.
\end{theorem}

So, in this work, we show that Algorithm~\ref{alg: one} has a guarantee $c$ under both the adversarial edge arrival model and the adversarial vertex arrival model for graphs of maximum degree three. Moreover, Algorithm~\ref{alg: one} is optimal for graphs of maximum degree three with respect to both these models, which follows directly from the upper bound in~\cite{Buchbinder}.

\subsection{Nonnegativity of Fractional Matching and Fractional Cover} \label{sec: nonnegativity}

In this section, we show that the values $y_e$, $e\in E$ and $x_u$, $u\in V$ computed by Algorithm~\ref{alg: one} are nonnegative. Towards that goal, we need to extend the set of properties that are satisfied by Algorithm~\ref{alg: one}. The next lemma contains all additional key properties, and we defer their proof to Appendix~\ref{sec: proof properties}. In the current section, we make use only of the properties~\ref{P2}, \ref{P3}, \ref{P5}, and \ref{P6}. The analysis of the algorithm requires all of the stated properties, though, and so all of the properties from \ref{P1} to \ref{P6} are proved together in Appendix~\ref{sec: proof properties}.

\begin{lemma}[Additional Properties] \label{lma main lemma additional}
    Let the underlying graph have maximum degree three. At every timepoint the values $y$ and $x$ computed by Algorithm~\ref{alg: one} satisfy the following properties:
    \begin{description}[style=multiline]
        \item[\namedlabel{P3}{P4}] for all $u\in V$ with $\degree(u) = 2$ and $\type\left(\delta(u)\right) \notin \left\{(0,2,0), (1,1,0)\right\}$ we have 
        \[x_u \in \left[2c - 1, \frac{5c - 2}{2}\right] \quad \text{ and }\quad x_u \geq c - 1 + \sum_{f \in \delta(u)} y_f\,,\] where the first statement can be reformulated as $c - x_u \in \left[1 - \frac{3c}{2}, 1 - c\right]$. This property implies that for every spoke $e$, i.e. for every edge $e$ with $\type(e) = 2$, we have  $y_e \in \left[1-\frac{3c}{2}, 1 - c\right]$ due to lines~\ref{line: y spoke edge} and ~\ref{line: spoke edge assigned} of Algorithm~\ref{alg: one}. 
        \item[\namedlabel{P4}{P5}] for all $u\in V$ the value $x_u$ can decrease only upon the arrival of an edge $e$ incident to $u$ such that $\type(\delta(u) \setminus \{e\}) \in \{(0,2,0), (1,1,0)\}$.
        \item[\namedlabel{P5}{P6}] for all $e \in E$ with $\type(e) = 1$ we have $x_{z(e)} \geq \widetilde{y}_{n_e + 1}$ and $x_{w(e)} \geq c - \widetilde{y}_{n_e + 1}$; additionally we have  $x_{w(e)} = c - \widetilde{y}_{n_e + 1}$ whenever $\degree(w(e)) = 1$.
        This property implies that if $n_e$ = 1 then we have $x_{z(e)} \geq c - \widetilde{y}_{n_e + 1}$, because $x_{z(e)} \geq \widetilde{y}_{n_e + 1}$ and  $\widetilde{y}_{n_e + 1} = \widetilde{y}_2 = c - \widetilde{y}_2 = \frac{c}{2}$.
        \item[\namedlabel{P6}{P7}] for all $e \in E$ with $\type(e) = 2$ we have $x_{w(e)} \geq y_e$.
    \end{description}
\end{lemma}

Lemma~\ref{lma main lemma additional} guarantees that a certain structure of $x$ is preserved throughout the algorithm. The properties in Lemma~\ref{lma main lemma additional} are also crucial in the construction of $y$.

To illustrate the importance of these properties, consider a scenario when an edge $e = uv$ arrives and before this arrival, we had $\deg(z(e))=2$ and $\type(\delta(z(e))) \notin \{(0,2,0), (1,1,0)\}$. In such a scenario,   as per Table~\ref{tab: edge-partitions}, we assign $e$ a spoke.  Thus, Property~\ref{P3} ensures that we can assign $e$ a value of $c - x_{z(e)}$, as in Figure~\ref{fig: cover construction} and line~\ref{line: spoke edge assigned} in Algorithm~\ref{alg: one},  while preserving~\ref{P7} and while ensuring that $y_e$ lies in $\left[1 - \frac{3c}{2}, 1 - c\right]$. Furthermore, as $\type(e)=2$ and as $\deg(z(e))=3$ after the arrival of $e$, Properties~\ref{P4} and \ref{P6} ensure that $e$ remains covered by $x$ in the future.

Let us consider a scenario when an arriving edge $e$ gets assigned a path edge and $n_e > 1$.  Property~\ref{P5} ensures that in a consistent instance with ideal value assignments, we have $x_{z(e)}\geq c - \widetilde{y}_{n_e}$, see Figure~\ref{fig: cover construction}. In general instances, if feasible with respect to~\ref{P7}, we assign $e$ the value $\widetilde{y}_{n_e}$; Otherwise, we assign $e$ the largest possible value that is feasible for~\ref{P7}. In the latter case, i.e. in the case when $e$ cannot be assigned $\widetilde{y}_{n_e}$, we show that $\degree(z(e))$ equals three and $x_{z(e)}$ before the arrival of $e$ is sufficiently large to overcome the limitation that the total increase of $x_{z(e)}+x_{w(e)}$ is now smaller than $\widetilde{y}_{n_e}$. The nature of property~\ref{P5} is to ensure the cover construction of each consistent instance is at least that of the ideal assignment case. 

Let us state the observation about the value assignments from Section~\ref{sec:fractional matching in consistent instances}. Note that~\eqref{ob: 3c/2 bound} in Observation~\ref{obs fractional matching consistent instances} follows from~\eqref{lma 1.7} in Lemma~\ref{lma fractional matching consistent instances} and from~\eqref{ob: 2} in Observation~\ref{obs fractional matching consistent instances}.
\begin{observation} 
The following properties hold:
\begin{enumerate}\label{obs fractional matching consistent instances}
    \item \label{ob: 1}
    for all natural $n$ we have $\widetilde{y}_n \leq c$ and $\widetilde{y}_{n + 1} \in \left[\frac{c}{2}, \frac{5c - 2}{2}\right]$.

    \item \label{ob: 2}
    for all natural $n$ we have $c - \widetilde{y}_{n + 1} \in \left[1 - \frac{3c}{2}, \frac{c}{2}\right]$.

    \item \label{ob: 3c/2 bound}
    for all natural $n$ we have $\widetilde{y}_{n+1} + \widetilde{y}_{n} \leq \frac{3c}{2}$. 
\end{enumerate}
\end{observation}

Now we are ready to prove that both the values $x_v$, $v\in V$, and the values $y_f$, $f\in E$ computed by Algorithm~\ref{alg: one} are nonnegative at every timepoint. Recall that for each path edge and spoke, we differentiate between its endpoints. In particular, for each $e \in E$ with $\type(e)\in \{1,2\}$ we defined $z(e)$ and $w(e)$ to be as in line~\ref{line: choice of z} of Algorithm~\ref{alg: one}.

\begin{lemma} \label{lma x non-neg}
    If the properties \ref{P1}, \dots, \ref{P6} hold at every timepoint, then at every timepoint the values $x_v$, $v\in V$ are nonnegative.
\end{lemma}

\begin{proof}[Proof of Lemma~\ref{lma x non-neg}]
    Assume for the sake of a contradiction that there exists $v \in V$ with $x_v < 0$ at some timepoint. If there exists $f \in \delta(v)$ with $\type(f) = 1$ then due to~\ref{P5}  we have that either $x_v \geq \widetilde{y}_{n_{f} + 1}$ or $x_v \geq c - \widetilde{y}_{n_{f} + 1}$. Hence by~\eqref{ob: 1} and~\eqref{ob: 2} in Observation~\ref{obs fractional matching consistent instances} we have that $x_v \geq 0$, contradiction. If there exists $f \in \delta(v)$ with $\type(f) = 2$ and $z(f) = v$ then due to the validity~\ref{P3} (immediately before the edge $f$ arrived) we have that $x_v \in \left[2c - 1, \frac{5c-2}{2}\right]$ and hence $x_v > 0$ contradiction. If there exists $f \in \delta(v)$ with $\type(f) = 2$ and $w(f) = v$ then due to \ref{P6} and \ref{P3}  we have $x_v \geq y_f \geq 1 - \frac{3c}{2} \geq 0$, contradiction.
    Thus, for each $f \in \delta(v)$ we have $\type(f) = 3$, but by  Table~\ref{tab: edge-partitions} this is impossible, a contradiction.
    
\end{proof}

\begin{lemma} \label{lma y non-neg}
    If the properties \ref{P1}, \dots, \ref{P6} hold at every timepoint, then at every timepoint the values $y_f$, $f\in E$ are nonnegative.
\end{lemma}

\begin{proof}[Proof of Lemma~\ref{lma y non-neg}]
    Let $f \in E$ and let us consider three possible types of $f$. If $\type(f) = 1$ then this type was assigned to $f$ in line~\ref{line: path edge assigned}, and so by line~\ref{line: y path edge} we have 
    \[y_f = \min\left\{\widetilde{y}_{n_f},\, 1 - \sum_{f_0 \in \delta(z(f)) \setminus \{f\}} y_{f_0} \right\} \geq 0\,,\] where the inequality holds due to~\ref{P7}.
    If $\type(f) = 2$ then by \ref{P3} we have $y_f \in \left[1 - \frac{3c}{2}, 1-c\right]$ hence $y_f \geq 0$.
    Finally if $\type(f) = 3$ then upon arrival of $f$ we have \[\left\{\type(\delta(u) \setminus \{f\}) \,\mid\,  u  \in \Eends(f)\right\} \in \mathcal{B\,.}\]  If $\{\type(\delta(u) \setminus \{f\}) \mid  u \in \Eends(f)\}$ is not $\{(1,0,0)\}$ then by assignments in Algorithm~\ref{alg: one} we have $y_f \geq 0$. If $\{\type(\delta(u) \setminus \{f\}) \mid  u \in \Eends(f)\}$ is $\{(1,0,0)\}$ then by line~\ref{line: b1 y assigned} in Algorithm~\ref{alg: one} we have $y_f = \widetilde{y}_{n+1} - (c - \widetilde{y}_{m+1})$ for some natural $n$ and $m$. In this case, by part~\eqref{ob: 1} in Observation~\ref{obs fractional matching consistent instances} we have $\widetilde{y}_{n+1} \geq \frac{c}{2}$ and by part~\eqref{ob: 2} in Observation~\ref{obs fractional matching consistent instances} we have $(c - \widetilde{y}_{m+1}) \leq \frac{c}{2}$. Thus in this case we have $y_f = \widetilde{y}_{n+1} - (c - \widetilde{y}_{m+1}) \geq \frac{c}{2} - \frac{c}{2} = 0$, and so $y_f \geq 0$ as required.
    
\end{proof}

\subsection{Observations about Algorithm}
In this section, we collect some observations about Algorithm~\ref{alg: one}. Each of these observations is straightforward by itself, and all of them allow us to argue about Algorithm~\ref{alg: one} efficiently.

The first observation is about general properties for each type of edge: path edges, spokes, and bridges.

\begin{observation}\label{obs edge types}
    The following properties hold:
    \begin{enumerate}
        \item  \label{ob: 4}
    for $e \in E$ with $\type(e) = 1$, we have that  $n_e$ is a natural number and $y_e \leq \widetilde{y}_{n_e}$.  
        \item \label{ob: n_e} for $e \in E$ with $\type(e) = 1$, we have that at the moment when $e$ arrives
     \[ n_e = 
        \begin{cases}
            n_f + 1 & \mbox{ if there exists\quad$ f \in \delta(z(e)) \setminus \{e\},\, \type(f) = 1$}\\
            1 & \mbox{ otherwise }\\
        \end{cases} 
    \,.\]
    \item \label{obs 2.1}   
    for $e \in E$ with $\type(e) = 1$, let us assume  $e=uv$, $\degree(v) \leq \degree(u)$ and let us assume that at the moment when $e$ arrives $\degree(v) < 3$. Then for $ \ f \in \delta(v) \setminus \{e\}$ we have $\type(f) =2$ or $\type(f) =3$ at the moment when $e$ arrives.

        \item \label{ob: 7}
    for $e \in E$ with $\type(e)=2$,  let us assume that after the arrival of $e$ we have $\deg(z(e))=3$. Then, after and at the arrival of $e$ the value $x_{z(e)}$ is not changing.
        \item \label{ob: 3}
    for $e \in E$ with $\type(e) = 3$, we have $\deg(u) \geq 2$ for all $u \in \Eends(e)$.
    \end{enumerate}
\end{observation}

Parts~\eqref{ob: 4} and~\eqref{ob: n_e} follow directly from lines~\ref{line: def of n} and \ref{line: y path edge} in Algorithm~\ref{alg: one}. Part~\eqref{obs 2.1} can be obtained by inspection of the rows and columns corresponding to $(1, 0, 0)$ in Table~\ref{tab: edge-partitions}. Part~\eqref{ob: 7} is due to the treatment of spokes in lines from~\ref{line: y spoke edge} to~\ref{line:  spoke cover change} of Algorithm~\ref{alg: one}. Part~\eqref{ob: 3} can be obtained by inspection of the entries leading to a bridge in Table~\ref{tab: edge-partitions}.

The second observation is about vertices in the graph with a specific structure of the edges incident to them.

\begin{observation}\label{obs vertex neighborhood}
    The following properties hold:
    \begin{enumerate}
        \item \label{ob: 8}
    for $v\in V$ with $\delta(v)\subseteq \{f\,|\, \type(f)=2\}$ and  $|\delta(v)|\leq 2$, we have $x_v=\sum_{f\in\delta(v)}y_f$.
    \item \label{ob: 5} for $u\in V$ with $\deg(u) = 2$, $f_1, f_2 \in \delta(u)$ and $\type(f_1) = \type(f_2) = 1$, we have $|n_{f_1} - n_{f_2}| = 1$. 
    \end{enumerate}
\end{observation}

Part~\eqref{ob: 8} follows from line~\ref{line:  spoke cover change} in Algorithm~\ref{alg: one}. Part~\eqref{ob: 5} follows from the rows and columns corresponding to $(1,0, 0)$ in Table~\ref{tab: edge-partitions} and the line~\ref{line: def of n} in Algorithm~\ref{alg: one}.

Recall that for each path edge and spoke $e$ we differentiate between its endpoints $z(e)$ and $w(e)$ as in line~\ref{line: choice of z} of Algorithm~\ref{alg: one}. The next observation is related to the structure of these endpoints.

\begin{observation} \label{obs special endpoints}
    The following properties hold:
    \begin{enumerate}
        \item \label{remark: 1}
    for $e \in E$ with $\type(e) =1$ or $\type(e) =2$,  we have $\degree(z(e)) \geq \degree(w(e))$.
     \item \label{ob: 6} for $e \in E$ with $\type(e) =1$, at the moment when $e$ arrives we have $|\{f \in \delta(z) \setminus \{e\} \mid \type(f) = 1\}| \leq 1$. 
    \item \label{remark: 2}
    for $e \in E$ with $\type(e) = 1$ and $n_e > 2$, there is a unique $f \in \delta(z(e)) \setminus \{e\}$ such that $\type(f) = 1$ and $z(e) = w(f)$. 

    Note that for $e \in E$ with $\type(e) = 1$ and $n_e = 2$, there is an edge $f \in \delta(z(e)) \setminus \{e\}$ such that $\type(f) = 1$ and $n_f=1$, and we have  $z(e) = w(f)$ or $z(e) = z(f)$.
    \end{enumerate}
\end{observation}

Part~\eqref{remark: 1} is due to the assumption $\degree(v) \leq \degree(u)$ and the definition of $z(e)$ and $w(e)$ in Algorithm~\ref{alg: one}.
To see part~\eqref{ob: 6}, consider the rows and columns corresponding to $(2, 0, 0)$ in Table~\ref{tab: edge-partitions}. To see part~\eqref{remark: 2}, consider the edge $f$ that was used to assign a value to $n_e$ in line~\ref{line: def of n} of Algorithm~\ref{alg: one}. We have $\type(f)=1$ and $n_e=n_f+1$ and so $n_f > 1$. Assume for a contradiction to~\eqref{remark: 2} that we have $z(f) = z(e)$. Due to $n_f > 1$  we have that there exists an edge $f' \in \delta(z(f))$, $\type(f') = 1$ and $n_f = n_{f'} + 1$, contradicting $\type(e) = 1$ as $\degree(z(e)) = 3$ and $\type(\delta(z(e)) \setminus \{e\}) = (2,0,0)$. Thus $\type(e) = 2$ by line~\ref{line: test spoke edge} in Algorithm~\ref{alg: one}, contradiction.

Note that part~\eqref{ob: 6} of  Observation~\ref{obs special endpoints} implies that the parameter $n_e$ is well defined for $e \in E$ with $\type(e) = 1$, i.e. there is only one choice for $n_e$. 

Finally, we prove the following claim regarding vertices of degree two, which are incident to precisely a path edge and a spoke. In the following, $x$ represents the vertex cover produced by Algorithm~\ref{alg: one}.
\begin{claim} \label{claim 1}
    Assume properties~\ref{P1}, \dots, \ref{P6} hold, then for all $v \in V$ such that $\type(\delta(v)) = (1,1,0)$ with $f_p, f_s \in \delta(v)$, $\type(f_p) = 1$ and $\type(f_s) = 2$ we have that $x_v = c - \widetilde{y}_{n_{f_p} + 1} + y_{f_s}$.
\end{claim}
\begin{proof}[Proof of Claim~\ref{claim 1}]
    First as $\degree(v) = 2$ it follows from line~\ref{line: test spoke edge} that $w(f_s) = v$ and hence by line~\ref{line:  spoke cover change} $f_s$ contributes $y_{f_s}$ to $x_{v}$. If $n_{f_p} \neq 1$ then by \eqref{remark: 1} in Observation~\ref{obs special endpoints} we have that $w(f_p) = v$ and hence following line~\ref{line: x path edge w} we have that $f_p$ contributes $c - \widetilde{y}_{n_{f_p} + 1}$ to $x_v$ as required. If $n_{f_p} = 1$ and $w(f_p) = v$ then $f_p$ contributes $c - \widetilde{y}_{n_{f_p} + 1}$ to $x_v$ for the same reasoning as in the case where $n_{f_p} \neq 1$. If $n_{f_p} = 1$ and $z(f_p) = v$ then by \ref{P3} we have that $y_{f_s} \leq 1 - c$ and hence $y_{f_p} = \widetilde{y}_1 = c$; so following line~\ref{line: x path edge z} we have that $f_p$ contributes $y_{f_p} - (c - \widetilde{y}_{n_{f_p} + 1}) = \frac{c}{2} = c - \widetilde{y}_{n_{f_p} + 1}$ to $x_v$ as required.
    
\end{proof}

\subsection{Bridge Assignments} \label{sec: bridge assignments}

In this section, we provide intuition behind the values assigned to bridges and the fractional cover updates for their endpoints in Algorithm~\ref{alg: one}. Consider an arriving edge $e$ that is a bridge, so we have $\{\type(\delta(u) \setminus \{e\}), \type(\delta(v) \setminus \{e\})\} \in \mathcal{B}$. To provide intuition, we consider each case in $\mathcal{B}$ separately.

We use the notation $x^{old}$ to denote the vertex cover before the arrival of $e$, and $x^{new}$ to denote the vertex cover after the arrival of $e$ and its respective updates. For simplicity, we assume that all edges arrived before $e$ are assigned their ideal values as in consistent instances. Similarly, we further assume that the cover $x^{old}$ satisfies the ideal cover construction depicted in Figure~\ref{fig: cover construction}. These assumptions allow us to simplify the exposition; dropping these assumptions requires more nuanced calculations, which we defer to Appendix~\ref{sec: proof properties}.

\subsubsection*{Case 1: $\{\type(\delta(u) \setminus \{e\}), \type(\delta(v) \setminus \{e\})\} = \{(1,0,0)\}$}

Let $f_u \in \delta(u) \setminus \{e\}$ and $f_v \in \delta(v) \setminus \{e\}$, thus we have $\type(f_u) = \type(f_v) = 1$, i.e. both are path edges and  $i \coloneq n_{f_u}$ and $j \coloneq n_{f_v}$. As per our assumptions, we have $y_{f_u} = \widetilde{y}_{i}$ and $y_{f_v} = \widetilde{y}_j$. 

First, let us provide the intuition for the value assigned to the edge $e$. To guarantee the approximation ratio, we ensure $\left(x^{new}_u-x^{old}_u\right)+\left(x^{new}_v-x^{old}_v\right)=y_e$ as in~\ref{P1} and $x^{new}_u+x^{new}_v\geq c$ as in~\ref{P2}. Therefore, we have to ensure 
\begin{equation*}
    y_e + x^{old}_u + x^{old}_v = y_e + (c - \widetilde{y}_{i+1}) + (c - \widetilde{y}_{j+1}) \geq c\,,
\end{equation*}
after rearranging we obtain $y_e \geq \widetilde{y}_{i+1} + \widetilde{y}_{j+1} - c$. Note that $y_e$ is assigned $\widetilde{y}_{i+1} + \widetilde{y}_{j+1} - c$  in line~\ref{line: b1 y assigned} in Algorithm~\ref{alg: one}. Intuitively, this is the minimum value we can assign $e$ to ensure~\ref{P1} and~\ref{P2}.

Now, let us provide the intuition for the assignments $x^{new}_u$ and $x^{new}_v$. Without loss of generality, we assume that $\widetilde{y}_{i+1} \leq \widetilde{y}_{j+1}$, that is for $z(e)$ and $w(e)$ as defined in line~\ref{line: b1 z, w} in Algorithm~\ref{alg: one} we have $z(e) = u$ and $w(e) = v$.
Furthermore, as we are assuming the edge assignments and cover $x^{old}$ followed the values as in consistent instances depicted in Figure~\ref{fig: cover construction}, we have
\begin{equation*}
    x^{old}_{z(e)} = c - \widetilde{y}_{i+1} \qquad \text{and} \qquad x^{old}_{w(e)} = c - \widetilde{y}_{j + 1}\,.
\end{equation*}

Since $(x^{new}_u-x^{old}_u)+(x^{new}_v-x^{old}_v)=y_e$, let us explain how to  ``distribute" $y_e$ to get $x^{new}$ from $x^{old}$. Let us assume $y_e = a + b$ where $x_u^{new} = x_u^{old} + a$ and $x^{new}_v = x^{old}_v + b$. For Property~\ref{P3} to hold we require
\begin{equation*}
    x^{new}_u = x^{old}_u + a = c - \widetilde{y}_{i+1} + a \in \left[ 2c - 1, \frac{5c - 2}{2} \right] \qquad \text{and} \qquad x^{new}_v = x^{old}_v + b = c - \widetilde{y}_{j+1} + b \in \left[ 2c - 1, \frac{5c - 2}{2} \right]
\end{equation*}

As $a = y_e - b$ and $y_e = \widetilde{y}_{i+1} + \widetilde{y}_{j+1} - c$, we have $x^{new}_u = \widetilde{y}_{j+1} - b$. Our strategy  is to ``balance" $x_u$ and $x_v$, i.e. to try to achieve $x^{new}_u = x^{new}_v$ if possible while preserving~\ref{P1}-\ref{P6}. The property which is a ``bottleneck" is Property~\ref{P3}.

On one side, to achieve $x^{new}_u = x^{new}_v$ we need to have $b = \widetilde{y}_{j+1} - \frac{c}{2}$ and $a = \widetilde{y}_{i+1} - \frac{c}{2}$; and  so if this assignment is possible we get $x^{new}_u = x^{new}_v = \frac{c}{2}$. 

On the other side, to preserve Property~\ref{P3} we require
\begin{equation*}
    x^{new}_u = c - \widetilde{y}_{i+1} + a \geq c - 1 + \widetilde{y}_i + y_e = \widetilde{y}_{j+1} - (1 - \widetilde{y}_i - \widetilde{y}_{i+1})\,.
\end{equation*}
Rearranging implies that we need $a \geq \widetilde{y}_{i+1} - c + \widetilde{y}_{j+1} - (1 - \widetilde{y}_i - \widetilde{y}_{i+1}) = y_e - (1 - \widetilde{y}_i - \widetilde{y}_{i+1})$. 
Thus to ``balance" $x_u$ and $x_v$ when possible with respect to Property~\ref{P3}, we let $a = \widetilde{y}_{i+1} - \min\{\frac{c}{2}, 1 - \widetilde{y}_i - y_e\}$ as seen in Figure~\ref{fig: cover construction for bridge 1} and given in line~\ref{line: b1 z cover change} in Algorithm~\ref{alg: one}. Finally, as $b = y_e - a$ we also get
\begin{equation*}
    b = y_e - \widetilde{y}_{i+1} + \min\{\frac{c}{2}, 1 - \widetilde{y}_i - y_e\} = \widetilde{y}_{j+1} - c + \min\{\frac{c}{2}, 1 - \widetilde{y}_i - y_e\} = \widetilde{y}_{j+1} - \max\{\frac{c}{2}, c - (1 - \widetilde{y}_i - y_e)\}\,.
\end{equation*}
as seen in Figure~\ref{fig: cover construction for bridge 1} and  given in line~\ref{line: b1 w cover change} in Algorithm~\ref{alg: one}.

\begin{figure}[H]
    \centering
    \begin{tikzpicture}[node distance={25mm}, main/.style = {draw, circle}] 
        \begin{scope}[every node/.style={circle, thick, draw, minimum size=5mm}]
            \node[main] (1ld) {$u$}; 
            \node[main] (1lu) [above of=1ld] {};
            \node[main] (2lu) [above of=1lu] {}; 

            \node[main, draw=none] (bfake) [right of=1ld] {};

            \node[main] (1rd) [right of=bfake] {$v$}; 
            \node[main] (1ru) [above of=1rd] {};
            \node[main] (2ru) [above of=1ru] {}; 
        \end{scope}

        \begin{scope}[>={Stealth[black]},
              every edge/.style={draw=black}, line width=2pt]
            \draw[my_green] (1ld) -- node[midway, right, pos=0.5, text=black] {$\widetilde{y}_i$} (1lu);
            \draw[my_green] (1lu) -- node[midway, left, pos=0.5, text=black] {} (2lu);

            \draw[my_green] (1rd) -- node[midway, left, pos=0.5, text=black] {$\widetilde{y}_j$} (1ru);
            \draw[my_green] (1ru) -- node[midway, left, pos=0.5, text=black] {} (2ru);

            \draw[my_blue, dashed] (1ld.east) to ["$ \widetilde{y}_{i+1} - (c - \widetilde{y}_{j+1})$" {above, font=\fontsize{7}{0}\selectfont, shift={(0, 0)}}, text = black] (1rd.west);

            \coordinate[xshift=-6mm] (1l) at ($(1ld)!0.5!(1lu)$);
            \coordinate[xshift=6mm] (1r) at ($(1rd)!0.5!(1ru)$);

            \coordinate[yshift=-6mm] (bl) at ($(1ld)!0.45!(1rd)$);
            \coordinate[yshift=-6mm] (br) at ($(1ld)!0.55!(1rd)$);

        \end{scope}

        \begin{scope}[>={Stealth[red]},
              every edge/.style={draw=red}]

            %path edges
            \draw[red, dashed] [->] (1l.north) to ["$c - \widetilde{y}_{i+1}$" {left, font=\fontsize{7}{0}\selectfont, shift={(-0.1cm, 0)}}, bend right=35] (1ld);

            \draw[red, dashed] [->] (1r.north) to ["$c - \widetilde{y}_{j+1}$" {right, font=\fontsize{7}{0}\selectfont, shift={(0.1cm, 0)}}, bend left=35] (1rd);

            %bridge
            \draw[red] [->] (bl.west) to ["$\widetilde{y}_{i+1} - \min\{ \frac{c}{2} \text{,} \ 1 - \widetilde{y}_{i} - y_e \}$" {below, font=\fontsize{7}{0}\selectfont, shift={(-2cm, 0)}}, bend left=35] (1ld);

            \draw[red] [->] (br.east) to ["$\widetilde{y}_{j+1} - \max\{ \frac{c}{2} \text{,} \ c - (1 - \widetilde{y}_i - y_e) \}$" {below, font=\fontsize{7}{0}\selectfont, shift={(2cm, 0)}}, bend right=35] (1rd);
        \end{scope}

    \end{tikzpicture}
    \caption{
    Edge assignment and cover updates upon the arrival of a bridge $e$ with $\{\type(\delta(u) \setminus \{e\}), \type(\delta(v) \setminus \{e\})\} = \{(1,0,0)\}$. The dashed cover assignments represent the values of $x^{old}$ under the assumption that prior to the arrival of $e$, assignments and cover construction are as per Figure~\ref{fig: cover construction}. Whereas, the solid cover assignments represent the contribution of $e$ to the updated cover $x^{new}$. In particular, the value of $x^{new}_{u}$ is the sum of the dashed and solid assignments into $u$, and the same holds for $x^{new}_v$.
    }
    \label{fig: cover construction for bridge 1}
\end{figure}

\subsubsection*{Case 2: $\{\type(\delta(u) \setminus \{e\}), \type(\delta(v) \setminus \{e\})\} = \{(1,0,0), (0,1,0)\}$}

Without loss of generality, we can assume that before the arrival of $e$ we had $f_u\in\delta(u)$ and $f_v\in \delta(v)$ such that $\type(f_u) = 1$ and $\type(f_v) = 2$. Thus as per the assignment of $z(e)$ and $w(e)$ in line~\ref{line: b2 definitions} of Algorithm~\ref{alg: one}, we have $z(e) = u$ and $w(e) = v$. Let $i \coloneq n_{f_u}$ be the position indicator for $f_u$ in its respective consistent instance. Furthermore, as stated above, we assume $y_{f_u} = \widetilde{y}_{i}$ and $y_{f_v} = 1 - \widetilde{y}_j - \widetilde{y}_{j+1}$ for some $j$. Therefore, as in Figure~\ref{fig: cover construction} we have
\begin{equation*}
    x^{old}_u = c - \widetilde{y}_{i+1} \qquad \text{and} \qquad x^{old}_v = y_{f_v}\,.
\end{equation*}
To guarantee~\ref{P2}, we need $x^{new}_u + x^{new}_v\geq c$ and so we require
\begin{equation*}
    x^{new}_u + x^{new}_v=y_e + x^{old}_u + x^{old}_v = y_e + (c - \widetilde{y}_{i+1}) + y_{f_v} \geq c\,,
\end{equation*}
leading to $y_e \geq \widetilde{y}_{i+1} - y_{f_v}$. Note, that by \eqref{ob: 2} in Observation~\ref{obs fractional matching consistent instances}  we have $y_{f_v} \in [1- \frac{3c}{2}, \frac{c}{2}]$. By Observation~\ref{ob: 1} we have $\widetilde{y}_{i+1} \geq c/2$.
Therefore, $\max\{\widetilde{y}_{i+1} - y_{f_v}, 0\} = \widetilde{y}_{i+1} - y_{f_v}$; and so under the current assumptions, assigning $y_e$ the value of $\widetilde{y}_{i+1} - y_{f_v}$ is precisely the same as assigning $y_e$ the value $\max\{\widetilde{y}_{i+1} - y_{f_v}, 0\}$ in line~\ref{line: b2 y assigned} in Algorithm~\ref{alg: one}.

Let us now consider how to ``distribute" $y_e$ to define $x^{new}$. Let us assume $y_e = a + b$ where $x^{new}_u = x^{old}_u + a$ and $x^{new}_v = x^{old}_v + b$. To satisfy Property~\ref{P3}, we require
\begin{equation*}
    x^{new}_u = c - \widetilde{y}_{i+1} + a \in \left[2c - 1, \frac{5c-2}{2}\right] \qquad \text{and} \qquad x^{new}_v = y_{f_v} + b \in \left[2c - 1, \frac{5c-2}{2}\right]\,.
\end{equation*}
 Therefore, we require $y_{f_v} + b \geq 2c-1$, and so $b \geq (2c-1) - y_{f_v}$. To satisfy Property~\ref{P5}, we require $x^{new}_v \geq x^{old}_v$, and so $b\geq 0$. Taking the maximum of these lower bounds on $b$, we get precisely the cover update in line~\ref{line: b2 w cover change} in Algorithm~\ref{alg: one} as seen in Figure~\ref{fig: cover construction for bridge 2}. So we let $b$ be $\max\{(2c-1) - y_{f_v}, 0\}$. Now we can retrieve the appropriate value of $a$ since $a=y_e-b$ by construction.

\begin{figure}[H]
    \centering
    \begin{tikzpicture}[node distance={25mm}, main/.style = {draw, circle}] 
        \begin{scope}[every node/.style={circle, thick, draw, minimum size=5mm}]
            \node[main] (1ld) {$u$}; 
            \node[main] (1lu) [above of=1ld] {};
            \node[main] (2lu) [above of=1lu] {}; 

            \node[main, draw=none] (bfake) [right of=1ld] {};

            \node[main] (1rd) [right of=bfake] {$v$}; 
            \node[main] (1ru) [above of=1rd] {};
            \node[main] (2rl) [left of=1ru] {}; 
            \node[main] (2rr) [right of=1ru] {}; 
        \end{scope}

        \begin{scope}[>={Stealth[black]},
              every edge/.style={draw=black}, line width=2pt]
            \draw[my_green] (1ld) -- node[midway, right, pos=0.5, text=black] {$\widetilde{y}_i$} (1lu);
            \draw[my_green] (1lu) -- node[midway, left, pos=0.5, text=black] {} (2lu);

            \draw[my_red, decorate,decoration=snake] (1rd) -- node[midway, left, pos=0.5, text=black] {$y_{f_v}$} (1ru);
            \draw[my_green] (1ru) -- node[midway, left, pos=0.5, text=black] {} (2rl);
            \draw[my_green] (1ru) -- node[midway, left, pos=0.5, text=black] {} (2rr);

            \draw[my_blue, dashed] (1ld.east) to ["$\widetilde{y}_{i+1} - y_{f_v}$" {above, font=\fontsize{7}{0}\selectfont, shift={(0, 0)}}, text = black] (1rd.west);

            \coordinate[xshift=-6mm] (1l) at ($(1ld)!0.5!(1lu)$);
            \coordinate[xshift=6mm] (1r) at ($(1rd)!0.5!(1ru)$);

            \coordinate[yshift=-6mm] (bl) at ($(1ld)!0.45!(1rd)$);
            \coordinate[yshift=-6mm] (br) at ($(1ld)!0.55!(1rd)$);

        \end{scope}

        %dual construction
        \begin{scope}[>={Stealth[red]},
              every edge/.style={draw=red}]

            %path edges
            \draw[red, dashed] [->] (1l.north) to ["$c - \widetilde{y}_{i+1}$" {left, font=\fontsize{7}{0}\selectfont, shift={(-0.1cm, 0)}}, bend right=35] (1ld);

            \draw[red, dashed] [->] (1r.north) to ["$y_{f_v}$" {right, font=\fontsize{7}{0}\selectfont, shift={(0.1cm, 0)}}, bend left=35] (1rd);

            %bridge
            \draw[red] [->] (bl.west) to ["$y_e - \max\{ (2c - 1) - y_{f_v} \text{,} \ 0 \}$" {below, font=\fontsize{7}{0}\selectfont, shift={(-1cm, 0)}}, bend left=35] (1ld);

            \draw[red] [->] (br.east) to ["$\max\{ (2c - 1) - y_{f_v } \text{,} \ 0 \}$" {below, font=\fontsize{7}{0}\selectfont, shift={(1cm, 0)}}, bend right=35] (1rd);
        \end{scope}

    \end{tikzpicture}
    \caption{
    Edge assignment and cover updates upon the arrival of a bridge $e$ with $\{\type(\delta(u) \setminus \{e\}), \type(\delta(v) \setminus \{e\})\} = \{(1,0,0), (0,1,0)\}$. The dashed cover assignments represent the values of $x^{old}$ under the assumption that prior to the arrival of $e$, assignments and cover construction are as per Figure~\ref{fig: cover construction}. Whereas, the solid cover assignments represent the contribution of $e$ to the updated cover $x^{new}$. In particular, the value of $x^{new}_{u}$ is the sum of the dashed and solid assignments into $u$, and the same holds for $x^{new}_v$.
    }
    \label{fig: cover construction for bridge 2}
\end{figure}

\subsubsection*{Case 3: $\{\type(\delta(u) \setminus \{e\}), \type(\delta(v) \setminus \{e\})\} = \{(1,0,0), (0,2,0)\}$}

Recall that $\degree(v) \leq \degree(u)$ and so we have $\type(\delta(v) \setminus \{e\}) = (1,0,0)$ and $\type(\delta(u) \setminus \{e\}) = (0,2,0)$. Let us assume that $f_v \in \delta(v) \setminus \{e\}$ and $f_u^1, f_u^2 \in \delta(u) \setminus \{e\}$, where $\type(f_v) = 1$ and $\type(f_u^1) = \type(f_u^2) = 2$. Let $i \coloneq n_{f_v}$ be the position indicator for $f_v$ in its respective instance. Furthermore, due to our assumptions  we have $y_{f_v} = \widetilde{y}_{i}$, and also by \eqref{ob: 2} in Observation~\ref{obs fractional matching consistent instances}  we have $y_{f_u^1}, y_{f_u^2} \in [1 - \frac{3c}{2}, \frac{c}{2}]$. So following the cover construction in Figure~\ref{fig: cover construction}, we have
\begin{equation*}
    x^{old}_u = y_{f_u^1} + y_{f_u^2} \qquad \text{and} \qquad x^{old}_v = c - \widetilde{y}_{i+1}.
\end{equation*}
To guarantee~\ref{P2}, we need $x^{new}_u + x^{new}_v\geq c$ and so we require
\begin{equation*}
    x^{new}_u + x^{new}_v=y_e + x^{old}_u + x^{old}_v = y_e + (y_{f_u^1} + y_{f_u^2}) + (c - \widetilde{y}_{i+1}) \geq c\,,
\end{equation*}
leading to $y_e \geq \widetilde{y}_{i+1} - y_{f_u^1} - y_{f_u^2}$. Notice that after the arrival of $e$, the degree of $u$ becomes three, and so no further arriving edges are incident to $u$. This observation motivates us to keep $x_u$ unchanged, i.e. to have $x^{new}_u = x^{old}_u$ while  guaranteeing~\ref{P2} for $f^1_u$ and $f^2_u$. Since we decide to have $x^{new}_u = x^{old}_u$, we ``distribute" the whole $y_e$ to $x_v$, i.e. we have $x^{new}_v = x^{old}_v + y_e$. Now, to adhere to Property~\ref{P3}, we require
\begin{equation*}
    x^{new}_v = x^{old}_v + y_e = c - \widetilde{y}_{i+1} + y_e \in \left[2c -1, \frac{5c - 2}{2}\right]\,.
\end{equation*}

Due to our assumptions on $y_{f_u^1}$ and $y_{f_u^2}$,  we have $c - y_{f_u^1}, c-y_{f_u^2} \in [2c-1, \frac{5c-2}{2}]$ by Property~\ref{P3}. We aim to have $x^{new}_v = c - y_{f_u^1}$ or $x^{new}_v = c - y_{f_u^2}$, because this would guarantee~\ref{P3} with respect to $x^{new}_v$. Thus to guarantee~\ref{P3}, we can select a nonnegative value for $y_e$ such that $y_e \geq \widetilde{y}_{i+1} - \max\{ y_{f_u^1}, y_{f_u^2}\}$. That corresponds precisely to the assignment in line~\ref{line: b3 y assigned} in Algorithm~\ref{alg: one}. 
Note that by Observation~\ref{obs fractional matching consistent instances} we have $\max\{ y_{f_u^1}, y_{f_u^2} \} \leq \frac{c}{2}$ and $\widetilde{y}_{i+1}\geq c/2$. Hence, we can assign $y_e$ precisely $\widetilde{y}_{i+1} - \max\{ y_{f_u^1}, y_{f_u^2}\}$, because this is a nonnegative value.

Now since $y_e = \widetilde{y}_{i+1} - \max\{ y_{f_u^1}, y_{f_u^2} \}$, we have
\begin{equation*}
    x^{new}_v = c - \widetilde{y}_{i+1} + y_e = c - \widetilde{y}_{i+1} + \widetilde{y}_{i+1} - \max\{ y_{f_u^1}, y_{f_u^2} \} = c - \max\{ y_{f_u^1}, y_{f_u^2} \}
\end{equation*}
as required.

\begin{figure}[H]
    \centering
    \begin{tikzpicture}[node distance={25mm}, main/.style = {draw, circle}] 
        \begin{scope}[every node/.style={circle, thick, draw, minimum size=5mm}]
            \node[main] (1ld) {$v$}; 
            \node[main] (1lu) [above of=1ld] {};
            \node[main] (2lu) [above of=1lu] {}; 

            \node[main, draw=none] (bfake) [right of=1ld] {};

            \node[main] (1rd) [right of=bfake] {$u$}; 
            \node[main] (1ru) [above of=1rd] {};
            \node[main] (2rl) [left of=1ru] {}; 
            \node[main] (2rr) [right of=1ru] {}; 

            \node[main] (1rdd) [below of=1rd] {};
            \node[main] (2rdl) [left of=1rdd] {}; 
            \node[main] (2rdr) [right of=1rdd] {}; 
        \end{scope}

        \begin{scope}[>={Stealth[black]},
              every edge/.style={draw=black}, line width=2pt]
            \draw[my_green] (1ld) -- node[midway, right, pos=0.5, text=black] {$\widetilde{y}_{i}$} (1lu);
            \draw[my_green] (1lu) -- node[midway, left, pos=0.5, text=black] {} (2lu);

            \draw[my_red, decorate,decoration=snake] (1rd) -- node[midway, left, pos=0.5, text=black] {$y_{f_u^1}$} (1ru);
            \draw[my_green] (1ru) -- node[midway, left, pos=0.5, text=black] {} (2rl);
            \draw[my_green] (1ru) -- node[midway, left, pos=0.5, text=black] {} (2rr);

            \draw[my_red, decorate,decoration=snake] (1rd) -- node[midway, left, pos=0.5, text=black] {$y_{f_u^2}$} (1rdd);
            \draw[my_green] (1rdd) -- node[midway, left, pos=0.5, text=black] {} (2rdl);
            \draw[my_green] (1rdd) -- node[midway, left, pos=0.5, text=black] {} (2rdr);

            \draw[my_blue, dashed] (1ld.east) to ["$\widetilde{y}_{i + 1} - \max\{ y_{f_u^1} \text{,} \ y_{f_u^2} \}$" {above, font=\fontsize{7}{0}\selectfont, shift={(0, 0.15)}}, text = black] (1rd.west);

            \coordinate[xshift=-6mm] (1l) at ($(1ld)!0.5!(1lu)$);
            \coordinate[xshift=6mm] (1r) at ($(1rd)!0.5!(1ru)$);
            \coordinate[xshift=6mm] (1r') at ($(1rd)!0.5!(1rdd)$);

            \coordinate[yshift=-6mm] (bl) at ($(1ld)!0.45!(1rd)$);
            \coordinate[yshift=-6mm] (br) at ($(1ld)!0.55!(1rd)$);

        \end{scope}

        %dual construction
        \begin{scope}[>={Stealth[red]},
              every edge/.style={draw=red}]

            %path edges
            \draw[red, dashed] [->] (1l.north) to ["$c - \widetilde{y}_{i+1}$" {left, font=\fontsize{7}{0}\selectfont, shift={(-0.1cm, 0)}}, bend right=35] (1ld);

            \draw[red, dashed] [->] (1r.north) to ["$y_{f_u^1}$" {right, font=\fontsize{7}{0}\selectfont, shift={(0.1cm, 0)}}, bend left=35] (1rd);

            \draw[red, dashed] [->] (1r'.south) to ["$y_{f_u^2}$" {right, font=\fontsize{7}{0}\selectfont, shift={(0.1cm, 0)}}, bend right=35] (1rd);

            %bridge
            \draw[red] [->] (bl.west) to ["$y_e$" {below, font=\fontsize{7}{0}\selectfont, shift={(0, 0)}}, bend left=35] (1ld);

        \end{scope}

    \end{tikzpicture}
    \caption{
    Edge assignment and cover updates upon the arrival of a bridge $e$ with $\{\type(\delta(u) \setminus \{e\}), \type(\delta(v) \setminus \{e\})\} = \{(1,0,0), (0,2,0)\}$. The dashed cover assignments represent the values of $x^{old}$ under the assumption that prior to the arrival of $e$, assignments and cover construction are as per Figure~\ref{fig: cover construction}. Whereas, the solid cover assignments represent the contribution of $e$ to the updated cover $x^{new}$. In particular, the value of $x^{new}_{u}$ is the sum of the dashed and solid assignments into $u$, and the same holds for $x^{new}_v$.
    }
    \label{fig: cover construction for bridge 3}
\end{figure}
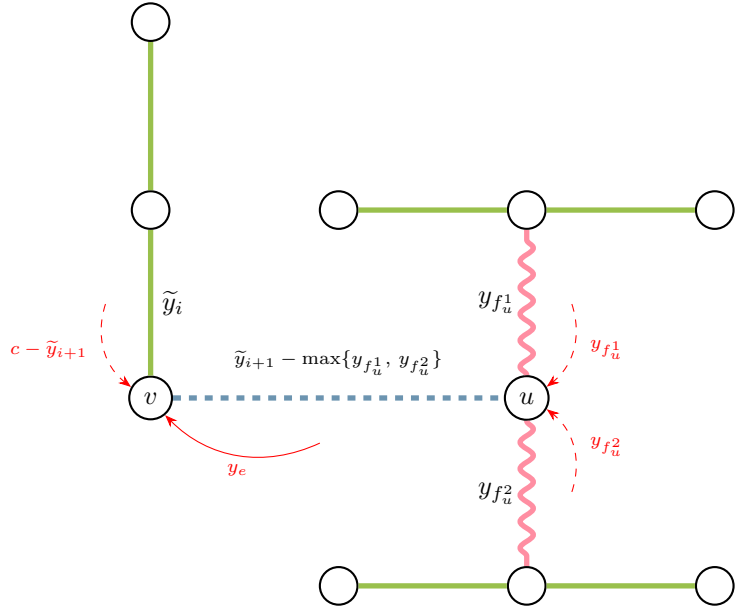

\subsubsection*{Case 4: $\{\type(\delta(u) \setminus \{e\}), \type(\delta(v) \setminus \{e\})\} = \{(1,1,0), (1,0,0)\}$}

Due to the assumption that $\degree(v) \leq \degree(u)$ we have $\type(\delta(v) \setminus \{e\}) = (1,0,0)$ and $\type(\delta(u) \setminus \{e\}) = (1,1,0)$. Let $f_v \in \delta(v) \setminus \{e\}$ and $f^1_u, f^2_u \in \delta(u) \setminus \{e\}$, that is $\type(f^1_u) = \type(f_v) = 1$ and $\type(f^2_u) = 2$. Let $i \coloneq n_{f^1_u}$ and $j \coloneq n_{f_v}$, thus as stated above we shall assume that $y_{f^1_u} = \widetilde{y}_i$ and $y_{f_v} = \widetilde{y}_j$, and by \eqref{ob: 2} in Observation~\ref{obs fractional matching consistent instances} $y_{f_u^2} \in [1-\frac{3c}{2}, \frac{c}{2}]$. Therefore, following the cover construction in Figure~\ref{fig: cover construction}, we have
\begin{equation*}
    x^{old}_u = c - \widetilde{y}_{i+1} + y_{f_u^2} \qquad \text{and} \qquad x^{old}_v = c - \widetilde{y}_{j+1}
\end{equation*}
To ensure $x^{new}_u + x^{new}_v\geq c$ as in~\ref{P2}, we require
\begin{equation*}
    x^{new}_u + x^{new}_v=y_e +  x^{old}_u + x^{old}_v  = y_e + (c - \widetilde{y}_{i+1} + y_{f_u^2}) + (c - \widetilde{y}_{j+1}) \geq c\,,
\end{equation*}
leading to $y_e \geq \widetilde{y}_{j+1} - y_{f_u^2} - (c - \widetilde{y}_{i+1})$.  Notice that after the arrival of $e$, the degree of $u$ becomes three, and so no further arriving edges are incident to $u$.  As in the previous case, we employ a similar strategy, although this time we no longer have $x^{new}_v=x^{old}_v+y_e$.

Let $y_e = a + b$, where $x^{new}_u = x^{old}_u + a$ and $x^{new}_v = x^{old}_v + b$. If we enforce
\begin{equation*}
    c - y_{f_u^2} = x^{new}_v = x^{old}_v + b = c - \widetilde{y}_{j+1} + b
\end{equation*}
then $x^{new}_v$ as in the previous case satisfies Property~\ref{P3}. Also if we enforce the above equality $c - y_{f_u^2}= c - \widetilde{y}_{j+1} + b$ then
we have $b = \widetilde{y}_{j+ 1} - y_{f_u^2}$, which is nonnegative by our assumptions on $y_{f_u^2}$. In a straightforward manner, from $b = \widetilde{y}_{j+ 1} - y_{f_u^2}$ and $y_e=a+b$ we also get  $a = y_e - b = y_e - (\widetilde{y}_{j+ 1} - y_{f_u^2})$, as seen in Figure~\ref{fig: cover construction for bridge 4}. So we have
\begin{equation*}
    x^{new}_u = x^{old}_u + a = c - \widetilde{y}_{i+1} + y_{f_u^2} + y_e - (\widetilde{y}_{j+1} - y_{f_u^2}) = c - \widetilde{y}_{i+1} - \widetilde{y}_{j+1} + 2y_{f_u^2} + y_e\,.
\end{equation*}

To ensure that~\ref{P2} holds for $f_u^1$ we require $x^{new}_u \geq c - \widetilde{y}_{i+1}$; similarly, to ensure~\ref{P2} for $f_u^2$ we require $x^{new}_u \geq y_{f_u^2}$ because these inequalities are guaranteeing~\ref{P2} in consistent instances. Therefore we require $x^{new}_u \geq \max\{ c - \widetilde{y}_{i+1}, y_{f_u^2} \}$. For this, we need
\begin{equation*}
    x^{new}_u = c - \widetilde{y}_{i+1} - \widetilde{y}_{j+1} + 2y_{f_u^2} + y_e \geq \max\{ c - \widetilde{y}_{i+1}, y_{f_u^2} \}
\end{equation*}
leading to $y_e \geq \max\{ c - \widetilde{y}_{i+1}, y_{f_u^2} \} + \widetilde{y}_{i+1} + \widetilde{y}_{j+1} - c - 2y_{f_u^2}$. Using $\max\{ c - \widetilde{y}_{i+1}, y_{f_u^2} \} = c - \widetilde{y}_{i+1} + y_{f_u^2} - \min\{c- \widetilde{y}_{i+1}, y_{f_u^2} \}$ we get
\begin{equation*}
    y_e \geq c - \widetilde{y}_{i+1} + y_{f_u^2} - \min\{c- \widetilde{y}_{i+1}, y_{f_u^2} \} + \widetilde{y}_{i+1} + \widetilde{y}_{j+1} - c - 2y_{f_u^2} = \widetilde{y}_{j+1} - y_{f_u^2} -  \min\{c- \widetilde{y}_{i+1}, y_{f_u^2} \}
\end{equation*}
taking the maximum with zero to ensure $y_e \geq 0$, we get precisely the assignment given in line~\ref{line: b4 y assigned} in Algorithm~\ref{alg: one}.

\begin{figure}[H]
    \centering
    \begin{tikzpicture}[node distance={25mm}, main/.style = {draw, circle}] 
        \begin{scope}[every node/.style={circle, thick, draw, minimum size=5mm}]
            \node[main] (1ld) {$v$}; 
            \node[main] (1lu) [above of=1ld] {};
            \node[main] (2lu) [above of=1lu] {}; 

            \node[main, draw=none] (bfake) [right of=1ld] {};

            \node[main] (1rd) [right of=bfake] {$u$}; 
            \node[main] (1ru) [below of=1rd, right of=1rd] {};
            \node[main] (2ru) [left of=1ru] {};

            \node[main] (1rdd) [above of=1rd] {};
            \node[main] (2rdl) [left of=1rdd] {}; 
            \node[main] (2rdr) [right of=1rdd] {}; 
        \end{scope}

        \begin{scope}[>={Stealth[black]},
              every edge/.style={draw=black}, line width=2pt]
            \draw[my_green] (1ld) -- node[midway, right, pos=0.5, text=black] {$\widetilde{y}_{j}$} (1lu);
            \draw[my_green] (1lu) -- node[midway, left, pos=0.5, text=black] {} (2lu);

            \draw[my_green] (1rd) -- node[midway, left, pos=0.5, text=black] {$\widetilde{y}_i$} (1ru);
            \draw[my_green] (1ru) -- node[midway, left, pos=0.5, text=black] {} (2ru);

            \draw[my_red, decorate,decoration=snake] (1rd) -- node[midway, left, pos=0.5, text=black] {$y_{f_u^2}$} (1rdd);
            \draw[my_green] (1rdd) -- node[midway, left, pos=0.5, text=black] {} (2rdl);
            \draw[my_green] (1rdd) -- node[midway, left, pos=0.5, text=black] {} (2rdr);

            \draw[my_blue, dashed] (1ld.east) to ["$\max\{ \widetilde{y}_{j+1} - y_{f_u^2}$" {above, font=\fontsize{7}{0}\selectfont, shift={(0, 0)}}, text = black] ["$- \min\{ c - \widetilde{y}_{i + 1} \text{,} \ y_{f_u^2} \} \text{,} \ 0 \}$" {below, font=\fontsize{7}{0}\selectfont, shift={(0, 0)}}, text = black] (1rd.west);

            \coordinate[xshift=-6mm] (1l) at ($(1ld)!0.5!(1lu)$);
            \coordinate[xshift=6mm] (1r) at ($(1rd)!0.5!(1ru)$);
            \coordinate[xshift=6mm] (1r') at ($(1rd)!0.5!(1rdd)$);

            \coordinate[yshift=-6mm] (bl) at ($(1ld)!0.45!(1rd)$);
            \coordinate[yshift=-6mm] (br) at ($(1ld)!0.55!(1rd)$);

        \end{scope}

        %dual construction
        \begin{scope}[>={Stealth[red]},
              every edge/.style={draw=red}]

            %path edges
            \draw[red, dashed] [->] (1l.north) to ["$c - \widetilde{y}_{j+1}$" {left, font=\fontsize{7}{0}\selectfont, shift={(-0.1cm, 0)}}, bend right=35] (1ld);

            \draw[red, dashed] [->] (1r.north) to ["$c - \widetilde{y}_{i + 1}$" {right, font=\fontsize{7}{0}\selectfont, shift={(0.1cm, 0)}}, bend right=35] (1rd);

            \draw[red, dashed] [->] (1r'.south) to ["$y_{f_u^2}$" {right, font=\fontsize{7}{0}\selectfont, shift={(0.1cm, 0)}}, bend left=35] (1rd);

            %bridge
            \draw[red] [->] (bl.west) to ["$\widetilde{y}_{j + 1} - y_{f_u^2}$" {below, font=\fontsize{7}{0}\selectfont, shift={(-1, 0)}}, bend left=35] (1ld);

            %bridge
            \draw[red] [->] (br.east) to ["$y_e - ( \widetilde{y}_{j+1} - y_{f_u^2})$" {below, font=\fontsize{7}{0}\selectfont, shift={(0.1, 0)}}, bend right=35] (1rd);

        \end{scope}

    \end{tikzpicture}
    \caption{
    Edge assignment and cover updates upon the arrival of a bridge $e$ with $\{\type(\delta(u) \setminus \{e\}), \type(\delta(v) \setminus \{e\})\} = \{(1,1,0), (1,0,0)\}$. The dashed cover assignments represent the values of $x^{old}$ under the assumption that prior to the arrival of $e$, assignments and cover construction are as per Figure~\ref{fig: cover construction}. Whereas, the solid cover assignments represent the contribution of $e$ to the updated cover $x^{new}$. In particular, the value of $x^{new}_{u}$ is the sum of the dashed and solid assignments into $u$, and the same holds for $x^{new}_v$.
    }
    \label{fig: cover construction for bridge 4}
\end{figure}

\section{Upper Bound for Integral Matchings for Maximum Degree Three} \label{sec: int alg}

In the previous sections, we showed that Algorithm~\ref{alg: one} achieves the guarantee $c$ for fractional matchings. In this section, we show that the guarantee $c$ cannot be achieved for integral matchings. In particular, in this section, we prove the following theorem.

\begin{theorem} \label{thm: int alg deg 3}
For integral matchings in the adversarial edge arrival model, no algorithm achieves a guarantee larger than $0.58065$ on graphs of maximum degree three.
\end{theorem}

Let $\mathcal{ALG}$ be a randomized algorithm for integral matchings. Let $\gamma$ be the competitive ratio achievable by $\mathcal{ALG}$.
Consider the following graph in Figure~\ref{fig: int deg 4 start}. Here, 
first edges to arrive are $e^u_1$ and $e^d_1$ followed by $e^u_2$ and $e^d_2$.

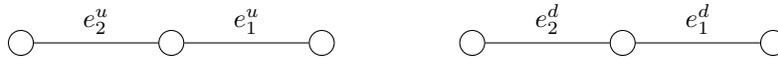
\begin{figure}[h!]
        \centering
        \begin{tikzpicture}[node distance={20mm}, main/.style = {draw, circle, scale=1}] 
            \node[main] (1ul) {};
            \node[main] (1ur) [right of=1ul] {};
            \node[main] (2ul) [left of=1ul] {};

            \node[main] (2dl) [right of=1ur] {};
            \node[main] (1dl) [right of=2dl] {};
            \node[main] (1dr) [right of=1dl] {};

            \draw (1ul) -- node[midway, above, pos=0.5] {$e^u_1$} (1ur);
            \draw (1ul) -- node[midway, above, pos=0.5] {$e^u_2$} (2ul);

            \draw (1dl) -- node[midway, above, pos=0.5] {$e^d_1$} (1dr);
            \draw (1dl) -- node[midway, above, pos=0.5] {$e^d_2$} (2dl);

        \end{tikzpicture}
        \caption{First edges in the considered instances.} \label{fig: int deg 4 start}
\end{figure}

Let $x^u_1$ and $x^u_2$ represent the probability that $e^u_1$ and $e^u_2$ are included in the matching by $\mathcal{ALG}$; similarly, let $x^d_1$ and $x^d_2$ be the probability that $e^d_1$ and $e^d_2$ are included in the matching. Due to the symmetry of the constructions in Figures~\ref{fig: int deg 4 start}, \ref{fig: int deg 4 first} and \ref{fig: int deg 4 second}, we may assume that $x^d_1$ equals $x^u_1$, and $x^d_2$ equals $x^u_2$. Indeed, for this we could see the process in the following way: first edges to arrive are $e^a_1$ and $e^b_1$ followed by $e^a_2$ and $e^b_2$ in Figure~\ref{fig: int deg 4 start labels}; then after $\mathcal{ALG}$ takes decisions related to these four edges, with probability $1/2$ we label the edges as $x^u_1:=x^a_1$, $x^d_1:=x^b_1$, $x^u_2:=x^a_2$, $x^d_2:=x^b_2$, and with probability $1/2$ we label the edges as $x^u_1:=x^b_1$, $x^d_1:=x^a_1$, $x^u_2:=x^b_2$, $x^d_2:=x^a_2$.  Note that the symmetry does not allow us to assume that the inclusion of $e^d_1$ and the inclusion of $e^u_1$ in the matching by $\mathcal{ALG}$ are independent events. Let us denote $x^d_1$ (and so $x^u_1$) by $x_1$, and similarly $x^d_2$ (and so $x^u_2$) by~$x_2$.

The constraints on $x_1$, $x_2$ and $\gamma$ that we can obtain from Figure~\ref{fig: int deg 4 start} are as follows: $2x_1\geq 2\gamma$ and $1-x_1-x_2\geq 0$.

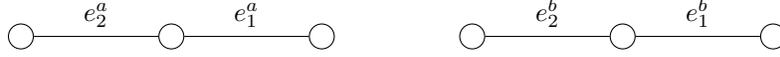
\begin{figure}[h!]
        \centering
        \begin{tikzpicture}[node distance={20mm}, main/.style = {draw, circle, scale=1}] 
            \node[main] (1ul) {};
            \node[main] (1ur) [right of=1ul] {};
            \node[main] (2ul) [left of=1ul] {};

            \node[main] (2dl) [right of=1ur] {};
            \node[main] (1dl) [right of=2dl] {};
            \node[main] (1dr) [right of=1dl] {};

            \draw (1ul) -- node[midway, above, pos=0.5] {$e^a_1$} (1ur);
            \draw (1ul) -- node[midway, above, pos=0.5] {$e^a_2$} (2ul);

            \draw (1dl) -- node[midway, above, pos=0.5] {$e^b_1$} (1dr);
            \draw (1dl) -- node[midway, above, pos=0.5] {$e^b_2$} (2dl);

        \end{tikzpicture}
        \caption{First edges in the constructed instances without ``labels''.} \label{fig: int deg 4 start labels}
\end{figure}

The first option for the further edge arrivals is depicted in Figure~\ref{fig: int deg 4 first}, where after the edges $e_1^u$, $e_1^d$, $e_2^u$, and $e_2^d$, all other edges in Figure~\ref{fig: int deg 4 first} arrive. In Figure~\ref{fig: int deg 4 first}, in the boxes, there are upper bounds on the edges' probability to be included in the matching by $\mathcal{ALG}$. Note that if $e_1^u$ and $e_2^u$ are included with probabilities $x_1$ and $x_2$, respectively, then by integrality of the constructed matching, neither of them is included with the probability $1-x_1-x_2$. An analogous statement holds for $e_1^d$, $e_2^d$. We obtain a new constraint on $x_1$, $x_2$ and $\gamma$ from Figure~\ref{fig: int deg 4 start}: $2x_1+2x_2+4(1-x_1-x_2)\geq 4\gamma$.

\begin{figure}
        \centering
        \begin{tikzpicture}[node distance={20mm}, main/.style = {draw, circle, scale=1}] 
            \node[main] (1ul) {};
            \node[main] (1ur) [right of=1ul] {};
            \node[main] (2ul) [left of=1ul] {};

            \node[main] (1ls) [below of=1ul] {};

            \node[main] (2dl) [right of=1ur] {};
            \node[main] (1dl) [right of=2dl] {};
            \node[main] (1dr) [right of=1dl] {};

            \node[main] (1ds) [below of=1dl] {};

            \draw (1ul) -- node[midway, above, pos=0.5] {$e^u_1$} (1ur);
            \draw (1ul) -- node[midway, above, pos=0.5] {$e^u_2$} (2ul);
            
            \draw[red, decorate,decoration=snake] (1ul) -- node[midway, right, pos=0.5] {\fbox{$\leq1 - x_1 - x_2$}} (1ls);

            \draw[red, decorate, decoration=snake] (2ul) to[out=45,in=135] node[at start, above, pos=0.5] {\fbox{$\leq1 - x_1 - x_2$}} (1ur);

            \draw (1dl) -- node[midway, above, pos=0.5] {$e^d_1$} (1dr);
            \draw (1dl) -- node[midway, above, pos=0.5] {$e^d_2$} (2dl);
            
            \draw[red, decorate,decoration=snake] (1dl) -- node[midway, right, pos=0.5] {\fbox{$\leq1 - x_1 - x_2$}} (1ds);

            \draw[red, decorate,decoration=snake] (2dl) to[out=45,in=135] node[at start, above, pos=0.5] {\fbox{$\leq1 - x_1 - x_2$}} (1dr);
        \end{tikzpicture}
        \caption{First option for the considered instances.} \label{fig: int deg 4 first}
\end{figure}
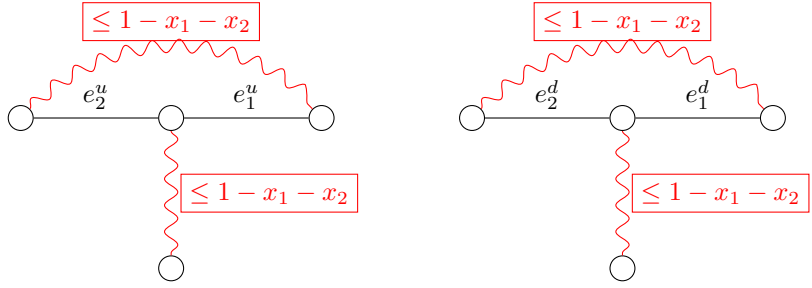

The second option for the further edge arrivals is depicted in Figure~\ref{fig: int deg 4 second}, where after the edges $e_1^u$, $e_1^d$, $e_2^u$ and $e_2^d$ the remaining edges depicted by solid straight lines arrive; then all curvy edges in Figure~\ref{fig: int deg 4 second} arrive; then all dashed edges arrive. Again, in Figure~\ref{fig: int deg 4 second},  in the boxes there are either the probabilities or the upper bounds on the edges' probability to be included in the matching by $\mathcal{ALG}$.  We obtain a new constraint on $x_1$, $x_2$, $x_3^u$, $x_3^d$, $x_4$ and $\gamma$ from Figure~\ref{fig: int deg 4 start}: \[2x_1+2 x_2+x_3^u+x_3^d+x_4+(1-x_1-x_3^u)+(1-x_1-x_3^d) \geq 5\gamma\] after the arrival of curvy edges; and \[2x_1+2 x_2+x_3^u+x_3^d+x_4+(1-x_1-x_3^u)+(1-x_1-x_3^d)+(1-x_3^u-x_4)+(1-x_3^d-x_4) \geq 6\gamma\] after the arrival of the dashed edge.  We also get constraints for each vertex in the graph.
     
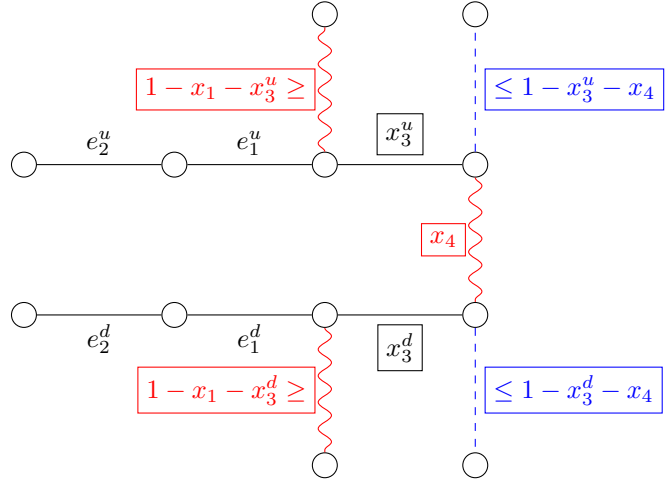
\begin{figure}
       \centering
        \begin{tikzpicture}[node distance={20mm}, main/.style = {draw, circle, scale=1}] 
            \node[main] (1ul) {};
            \node[main] (1ur) [right of=1ul] {};
            \node[main] (2ul) [left of=1ul] {};
            \node[main] (2ur) [right of=1ur] {};

            \node[main] (1dl) [below of=1ul] {};
            \node[main] (1dr) [right of=1dl] {};
            \node[main] (2dl) [left of=1dl] {};
            \node[main] (2dr) [right of=1dr] {};

            \node[main] (1us) [above of=1ur] {};
            \node[main] (1ds) [below of=1dr] {};

            \node[main] (2us) [above of=2ur] {};
            \node[main] (2ds) [below of=2dr] {};

            \draw (1ul) -- node[midway, above, pos=0.5] {$e^u_1$} (1ur);
            \draw (1ul) -- node[midway, above, pos=0.5] {$e^u_2$} (2ul);
            \draw (1ur) -- node[midway, above, pos=0.5] {\fbox{$x_3^u$}} (2ur);

            \draw (1dl) -- node[midway, below, pos=0.5] {$e^d_1$} (1dr);
            \draw (1dl) -- node[midway, below, pos=0.5] {$e^d_2$} (2dl);
            \draw (1dr) -- node[midway, below, pos=0.5] {\fbox{$x_3^d$}} (2dr);
            \draw[red, decorate,decoration=snake] (1ur) -- node[midway, left, pos=0.5] {\fbox{$1 - x_1 - x^u_3\geq$}} (1us);
            \draw[red, decorate,decoration=snake] (1dr) -- node[midway, left, pos=0.5] {\fbox{$1 - x_1 - x^d_3\geq$}} (1ds);

            \draw[red, decorate,decoration=snake] (2ur) -- node[midway, left, pos=0.5] {\fbox{$x_4$}} (2dr);

            \draw[blue, dashed] (2ur) -- node[midway, right, pos=0.5] {\fbox{$\leq1 - x^u_3 - x_4$}} (2us);
            \draw[blue, dashed] (2dr) -- node[midway, right, pos=0.5] {\fbox{$\leq 1 - x^d_3 - x_4$}} (2ds);
        \end{tikzpicture}
        \caption{Second option for the considered instances.} \label{fig: int deg 4 second}
\end{figure}

Thus, the following Linear Program provides an upper bound on $\gamma$.
\begin{maxi*}
    {}{\gamma}{}{}
    \addConstraint{2 x_1 \geq 2\gamma}
    \addConstraint{2x_1+2 x_2+4(1-x_1-x_2) \geq 4\gamma}
    \addConstraint{1-x_1-x_2 \geq 0}
    \addConstraint{2x_1+2 x_2+x_3^u+x_3^d \geq 4\gamma}
    \addConstraint{2x_1+2 x_2+x_3^u+x_3^d+x_4+(1-x_1-x_3^u)+(1-x_1-x_3^d) \geq 5\gamma}
    \addConstraint{2x_1+2 x_2+x_3^u+x_3^d+x_4+(1-x_1-x_3^u)+(1-x_1-x_3^d)}
    \addConstraint{}{ \qquad +(1-x_3^u-x_4)+(1-x_3^d-x_4) \geq 6\gamma }
    \addConstraint{1-x_1-x_3^u\geq 0}
    \addConstraint{1-x_1-x_3^d\geq 0}
    \addConstraint{1-x_3^u-x_4\geq 0}
    \addConstraint{1-x_3^d-x_4\geq 0}
    \addConstraint{0\leq x_1,\,x_2,\, x_3^u,\, x_3^d,\,x_4\leq 1}
\end{maxi*}
Solving the above linear program, we get $\gamma \leq 0.58065$ as required.

\section{Upper Bound for Fractional Matchings for Maximum Degree Four} \label{sec: deg 4 upper bound}

In this section, we show that for fractional matchings, the best possible guarantee deteriorates by going from graphs with maximum degree three to maximum degree four. In particular, we show the following theorem and note that we have $0.58884 < c$. 

\begin{theorem} \label{thm: fractional alg deg 4}
For fractional matchings in the adversarial edge arrival model, no algorithm achieves a guarantee larger than $0.58884$ on bipartite graphs of maximum degree four.
\end{theorem}

To prove Theorem~\ref{thm: fractional alg deg 4}, let us consider the graph in Figure~\ref{fig fractional deg 4}. The edges in Figure~\ref{fig fractional deg 4} arrive in 30 batches $B_1$, \ldots, $B_{30}$ of nine different types, according to the list below
    \begin{itemize}
        \item $B_1 := \left\{e_1\right\}$.
        \item $B_2  := \left\{e^l_2, \,e^r_2\right\}$.
        \item $B_i := \left\{e^l_i,\, e^r_i,\, \hat{e}^l_{i - 2}, \hat{e}^r_{i - 2} \right\}$ for $i = 3, \dots, 6$.
        \item $B_i := \left\{f_{i - 6},\, \hat{f}^l_{i - 6},\, \hat{f}^r_{i - 6} \right\}$ for $i = 7, \dots, 10$.
        \item $B_i := \left\{e^{l, c_{i - 10}}_2,\, e^{r, c_{i - 10}}_2\right\}$ for $i = 11, \dots, 14$.
        \item $B_i := \left\{e^{l, c_{i - 14}}_3,\, e^{r, c_{i - 14}}_3,\, \hat{e}^{l, c_{i - 14}}_1,\, \hat{e}^{r, c_{i - 14}}_1 \right\}$ for $i = 15, \dots, 18$.
        \item $B_i := \left\{e^{l, c_{i - 18}}_4,\, e^{r, c_{i - 18}}_4,\, \hat{e}^{l, c_{i - 18}}_2,\, \hat{e}^{r, c_{i - 18}}_2 \right\}$ for $i = 19, \dots, 22$.
        \item $B_i := \left\{e^{l, c_{i - 22}}_5,\, e^{r, c_{i - 22}}_5,\, \hat{e}^{l, c_{i - 22}}_3,\, \hat{e}^{r, c_{i - 22}}_3 \right\}$ for $i = 23, \dots, 26$.
        \item $B_i := \left\{e^{l, c_{i - 26}}_5,\, e^{r, c_{i - 26}}_5,\, \hat{e}^{l, c_{i - 26}}_3,\, \hat{e}^{r, c_{i - 26}}_3 \right\}$ for $i = 27, \dots, 30$.
    \end{itemize}
    where $B_i$ represents the $i$th edge batch to arrive. For example, first arrives the edges in $B_1$, then the edges in $B_2$, then the edges in $B_3$, and so on. In Figure~\ref{fig fractional deg 4}, the width of the edges and the looseness of the dashes/dots in the edge pattern indicates the order of arrival.
\begin{figure}
        \centering
        \begin{tikzpicture}[node distance={13mm}, main/.style = {draw, circle, scale=0.85},rotate=90,transform shape, scale=0.85]
            \node[main] (1l) {}; 
            \node[main] (1r) [right of=1l] {}; 
            \node[main] (2r) [right =30mm of 1r] {}; 
            \node[main] (3r) [right =30mm of 2r] {}; 
            \node[main] (4r) [right of=3r] {};
            \node[main] (5r) [right of=4r] {};
            \node[main] (6r) [right of=5r] {}; 
            \node[main] (2l) [left =30mm of 1l] {}; 
            \node[main] (3l) [left =30mm of 2l] {}; 
            \node[main] (4l) [left of=3l] {};
            \node[main] (5l) [left of=4l] {};
            \node[main] (6l) [left of=5l] {};
    
            \node (2lsi) [below of=2l] {};
            \node (2rsi) [below of=2r] {};
            \node (1lsi) [above of=1l] {};
            \node (1rsi) [above of=1r] {};
    
            \node (5ls') [below of=4l] {};
            \node (5rs') [below of=4r] {};
    
            \node[main] (1ls) [above of=1lsi] {};
            \node[main] (1rs) [above of=1rsi] {};
            \node[main] (1ls') [below of=1l] {};
            \node[main] (1rs') [below of=1r] {};
            \node[main] (2ls) [below right=15mm and 30mm of 2lsi] {};
            \node[main] (2rs) [below left=15mm and 30mm of 2rsi] {};
            \node[main] (2ls') [above of =2l] {};
            \node[main] (2rs') [above of =2r] {};
            \node[main] (3ls') [below of=3l] {};
            \node[main] (3rs') [below of=3r] {};
            \node[main] (4ls') [above of=4l] {};
            \node[main] (4rs') [above of=4r] {};
    
            \node[main] (3ls) [above of=4ls'] {};
            \node[main] (3rs) [above of=4rs'] {};
            \node[main] (4ls) [below of=5ls'] {};
            \node[main] (4rs) [below of=5rs'] {e};
    
            %child 1 nodes
            \node[main] (2rc1) [right of=1rs] {};
            \node[main] (2lc1) [left of=1ls] {}; 
            \node[main] (3rc1) [right of=2rc1] {};
            \node[main] (3lc1) [left of=2lc1] {};
            \node[main] (1rsc1) [above of=1rs] {};
            \node[main] (1lsc1) [above of=1ls] {};
            %new nodes
            \node[main] (4rc1) [right of=3rc1] {};
            \node[main] (4lc1) [left of=3lc1] {};
            \node[main] (5rc1) [right of=4rc1] {};
            \node[main] (5lc1) [left of=4lc1] {};
            \node[main] (6rc1) [right of=5rc1] {};
            \node[main] (6lc1) [left of=5lc1] {};
            %new spokes
            \node[main] (2rsc1) [above of=2rc1] {};
            \node[main] (2lsc1) [above of=2lc1] {};
            \node[main] (3rsc1) [above of=3rc1] {};
            \node[main] (3lsc1) [above of=3lc1] {};
            \node[main] (4rsc1) [above of=4rc1] {};
            \node[main] (4lsc1) [above of=4lc1] {};
    
            %child 2 nodes
            \node[main] (2rc2) [right of=2rs] {};
            \node[main] (2lc2) [left of=2ls] {}; 
            \node[main] (3rc2) [right of=2rc2] {};
            \node[main] (3lc2) [left of=2lc2] {};
            \node[main] (1rsc2) [below of=2rs] {};
            \node[main] (1lsc2) [below of=2ls] {};
            %new nodes
            \node[main] (4rc2) [right of=3rc2] {};
            \node[main] (4lc2) [left of=3lc2] {};
            \node[main] (5rc2) [right of=4rc2] {};
            \node[main] (5lc2) [left of=4lc2] {};
            \node[main] (6rc2) [right of=5rc2] {};
            \node[main] (6lc2) [left of=5lc2] {};
            %new spokes
            \node[main] (2rsc2) [below of=2rc2] {};
            \node[main] (2lsc2) [below of=2lc2] {};
            \node[main] (3rsc2) [below of=3rc2] {};
            \node[main] (3lsc2) [below of=3lc2] {};
            \node[main] (4rsc2) [below of=4rc2] {};
            \node[main] (4lsc2) [below of=4lc2] {};
    
            %child 3 nodes
            \node[main] (2rc3) [right of=3rs] {};
            \node[main] (2lc3) [left of=3ls] {}; 
            \node[main] (3rc3) [right of=2rc3] {};
            \node[main] (3lc3) [left of=2lc3] {};
            \node[main] (1rsc3) [above of=3rs] {};
            \node[main] (1lsc3) [above of=3ls] {};
            %new nodes
            \node[main] (4rc3) [above of=3rc3] {};
            \node[main] (4lc3) [above of=3lc3] {};
            \node[main] (5rc3) [above of=4rc3] {};
            \node[main] (5lc3) [above of=4lc3] {};
            \node[main] (6rc3) [above of=5rc3] {};
            \node[main] (6lc3) [above of=5lc3] {};
            %new spokes
            \node[main] (2rsc3) [above of=2rc3] {};
            \node[main] (2lsc3) [above of=2lc3] {};
            \node[main] (3rsc3) [right of=3rc3] {};
            \node[main] (3lsc3) [left of=3lc3] {};
            \node[main] (4rsc3) [right of=4rc3] {};
            \node[main] (4lsc3) [left of=4lc3] {};
    
            %child 4 nodes
            \node[main] (2rc4) [right of=4rs] {};
            \node[main] (2lc4) [left of=4ls] {}; 
            \node[main] (3rc4) [right of=2rc4] {};
            \node[main] (3lc4) [left of=2lc4] {};
            \node[main] (1rsc4) [below of=4rs] {};
            \node[main] (1lsc4) [below of=4ls] {};
            %new nodes
            \node[main] (4rc4) [below of=3rc4] {};
            \node[main] (4lc4) [below of=3lc4] {};
            \node[main] (5rc4) [below of=4rc4] {};
            \node[main] (5lc4) [below of=4lc4] {};
            \node[main] (6rc4) [below of=5rc4] {};
            \node[main] (6lc4) [below of=5lc4] {};
            %new spokes
            \node[main] (2rsc4) [below of=2rc4] {};
            \node[main] (2lsc4) [below of=2lc4] {};
            \node[main] (3rsc4) [right of=3rc4] {};
            \node[main] (3lsc4) [left of=3lc4] {};
            \node[main] (4rsc4) [right of=4rc4] {};
            \node[main] (4lsc4) [left of=4lc4] {};
    
            \draw[line width=4.1pt] (1l) -- node[midway, below, pos=0.5] {$e_1$} (1r);
            \draw[line width=3.8pt] (1l) -- node[midway, below, pos=0.5] {$e^l_2$} (2l);
            \draw[line width=3.8pt] (1r) -- node[midway, below, pos=0.5] {$e^r_2$} (2r);
            \draw[line width=3.5pt] (2l) -- node[midway, below, pos=0.5] {$e^l_3$} (3l);
            \draw[line width=3.5pt] (2r) -- node[midway, below, pos=0.5] {$e^r_3$} (3r);
            \draw[line width=3.2pt] (3l) -- node[midway, below, pos=0.5] {$e^l_4$} (4l);
            \draw[line width=3.2pt] (3r) -- node[midway, below, pos=0.5] {$e^r_4$} (4r);
            \draw[line width=2.9pt] (4l) -- node[midway, below, pos=0.5] {$e^l_5$} (5l);
            \draw[line width=2.9pt] (4r) -- node[midway, below, pos=0.5] {$e^r_5$} (5r);
            \draw[line width=2.6pt] (5l) -- node[midway, below, pos=0.5] {$e^l_6$} (6l);
            \draw[line width=2.6pt] (5r) -- node[midway, below, pos=0.5] {$e^r_6$} (6r);
    
            \draw[line width=3.5pt, decorate,decoration=snake]  (1l) -- node[midway, left, pos=0.5] {$\hat{e}^l_1$} (1ls);
            \draw[line width=3.5pt, decorate,decoration=snake]  (1r) -- node[midway, right, pos=0.5] {$\hat{e}^r_1$} (1rs);
            \draw[line width=3.2pt, decorate,decoration=snake] (2l) -- node[midway, above right, pos=0.5] {$\hat{e}^l_2$} (2ls);
            \draw[line width=3.2pt, decorate,decoration=snake] (2r) -- node[midway,  above left, pos=0.5] {$\hat{e}^r_2$} (2rs);
            \draw[line width=2.9pt, decorate,decoration=snake] (3l) -- node[midway, right, pos=0.5] {$\hat{e}^l_3$} (3ls);
            \draw[line width=2.9pt, decorate,decoration=snake] (3r) -- node[midway, left, pos=0.5] {$\hat{e}^r_3$} (3rs);
            \draw[line width=2.6pt, decorate,decoration=snake] (4l) -- node[midway, left, pos=0.5] {$\hat{e}^l_4$} (4ls);
            \draw[line width=2.6pt, decorate,decoration=snake] (4r) -- node[midway, right, pos=0.5] {$\hat{e}^r_4$} (4rs);
    
            \draw[line width=1.6pt, dashed] (1l) -- node[midway, left, pos=0.5] {$\hat{f}^l_1$} (1ls');
            \draw[line width=1.6pt,  dashed]  (1r) -- node[midway, right, pos=0.5] {$\hat{f}^r_1$} (1rs');
            \draw[line width=1.6pt,  dashed]  (1ls) -- node[midway, below, pos=0.5] {$f_1$} (1rs);
    
            \draw[line width=1.2pt, dashed] (2l) -- node[midway, left, pos=0.5] {$\hat{f}^l_2$} (2ls');
            \draw[line width=1.2pt,  dashed] (2r) -- node[midway, right, pos=0.5] {$\hat{f}^r_2$} (2rs');
            \draw[line width=1.2pt,  dashed] (2ls) -- node[midway, below, pos=0.5] {$f_2$} (2rs);
    
            \draw[line width=.9pt, dashed] (3l) -- node[midway, left, pos=0.5] {$\hat{f}^l_3$} (3ls');
            \draw[line width=.9pt, dashed] (3r) -- node[midway, right, pos=0.5] {$\hat{f}^r_3$} (3rs');
            \draw[line width=.9pt, dashed] (3ls) to[out=50,in= 130] node[midway, below, pos=0.5] {$f_3$} (3rs);
    
            \draw[line width=.6pt,  dashed] (4l) -- node[midway, left, pos=0.5] {$\hat{f}^l_4$} (4ls');
            \draw[line width=.6pt, dashed] (4r) -- node[midway, right, pos=0.5] {$\hat{f}^r_4$} (4rs');
            \draw[line width=.6pt, dashed] (4ls) to[out=-45,in= -135] node[midway, below, pos=0.5] {$f_4$} (4rs);
    
            %child 1 edges
            \draw[line width=.7pt, densely dotted] (1ls) -- node[midway, below, pos=0.5] {$e^{l, c_1}_2$} (2lc1);
            \draw[line width=.7pt, densely dotted] (1rs) -- node[midway, below, pos=0.5] {$e^{r, c_1}_2$} (2rc1);
            \draw[line width=.7pt, dotted] (2lc1) -- node[midway, below, pos=0.5] {$e^{l, c_1}_3$} (3lc1);
            \draw[line width=.7pt, dotted] (2rc1) -- node[midway, below, pos=0.5] {$e^{r, c_1}_3$} (3rc1);
            \draw[line width=.7pt, dotted] (1ls) -- node[midway, left, pos=0.5] {$\hat{e}^{l, c_1}_1$} (1lsc1);
            \draw[line width=.7pt, dotted] (1rs) -- node[midway, right, pos=0.5] {$\hat{e}^{r, c_1}_1$} (1rsc1);   

            %new edges
            \draw[line width=.7pt, dotted] (3lc1) -- node[midway, below, pos=0.5] {$e^{l, c_1}_4$} (4lc1);
            \draw[line width=.7pt, dotted] (3rc1) -- node[midway, below, pos=0.5] {$e^{r, c_1}_4$} (4rc1);
            \draw[line width=.7pt, dotted] (4lc1) -- node[midway, below, pos=0.5] {$e^{l, c_1}_5$} (5lc1);
            \draw[line width=.7pt, dotted] (4rc1) -- node[midway, below, pos=0.5] {$e^{r, c_1}_5$} (5rc1);
            \draw[line width=.7pt, dotted] (5lc1) -- node[midway, below, pos=0.5] {$e^{l, c_1}_6$} (6lc1);
            \draw[line width=.7pt, dotted] (5rc1) -- node[midway, below, pos=0.5] {$e^{r, c_1}_6$} (6rc1);

            \draw[line width=.7pt, dotted] (2lc1) -- node[midway, left, pos=0.5] {$\hat{e}^{l, c_1}_2$} (2lsc1);
            \draw[line width=.7pt, dotted] (2rc1) -- node[midway, right, pos=0.5] {$\hat{e}^{r, c_1}_2$} (2rsc1);
            \draw[line width=.7pt, dotted] (3lc1) -- node[midway, left, pos=0.5] {$\hat{e}^{l, c_1}_3$} (3lsc1);
            \draw[line width=.7pt, dotted] (3rc1) -- node[midway, right, pos=0.5] {$\hat{e}^{r, c_1}_3$} (3rsc1);  
            \draw[line width=.7pt, dotted] (4lc1) -- node[midway, left, pos=0.5] {$\hat{e}^{l, c_1}_4$} (4lsc1);
            \draw[line width=.7pt, dotted] (4rc1) -- node[midway, right, pos=0.5] {$\hat{e}^{r, c_1}_4$} (4rsc1);  
    
            %child 2 edges
            \draw[line width=.6pt, densely dotted] (2ls) -- node[midway, below, pos=0.5] {$e^{l, c_2}_2$} (2lc2);
            \draw[line width=.6pt, densely dotted] (2rs) -- node[midway, below, pos=0.5] {$e^{r, c_2}_2$} (2rc2);
            \draw[line width=.6pt, dotted] (2lc2) -- node[midway, above, pos=0.5] {$e^{l, c_2}_3$} (3lc2);
            \draw[line width=.6pt, dotted] (2rc2) -- node[midway, above, pos=0.5] {$e^{r, c_2}_3$} (3rc2);
            \draw[line width=.6pt, dotted] (2ls) -- node[midway, left, pos=0.7] {$\hat{e}^{l, c_2}_1$} (1lsc2);
            \draw[line width=.6pt, dotted] (2rs) -- node[midway, right, pos=0.7] {$\hat{e}^{r, c_2}_1$} (1rsc2); 

            %new edges
            \draw[line width=.7pt, dotted] (3lc2) -- node[midway, above, pos=0.5] {$e^{l, c_2}_4$} (4lc2);
            \draw[line width=.7pt, dotted] (3rc2) -- node[midway, above, pos=0.5] {$e^{r, c_2}_4$} (4rc2);
            \draw[line width=.7pt, dotted] (4lc2) -- node[midway, above, pos=0.5] {$e^{l, c_2}_5$} (5lc2);
            \draw[line width=.7pt, dotted] (4rc2) -- node[midway, above, pos=0.5] {$e^{r, c_2}_5$} (5rc2);
            \draw[line width=.7pt, dotted] (5lc2) -- node[midway, above, pos=0.5] {$e^{l, c_2}_6$} (6lc2);
            \draw[line width=.7pt, dotted] (5rc2) -- node[midway, above, pos=0.5] {$e^{r, c_2}_6$} (6rc2);

            \draw[line width=.7pt, dotted] (2lc2) -- node[midway, left, pos=0.5] {$\hat{e}^{l, c_2}_2$} (2lsc2);
            \draw[line width=.7pt, dotted] (2rc2) -- node[midway, right, pos=0.5] {$\hat{e}^{r, c_2}_2$} (2rsc2);
            \draw[line width=.7pt, dotted] (3lc2) -- node[midway, left, pos=0.5] {$\hat{e}^{l, c_2}_3$} (3lsc2);
            \draw[line width=.7pt, dotted] (3rc2) -- node[midway, right, pos=0.5] {$\hat{e}^{r, c_2}_3$} (3rsc2);  
            \draw[line width=.7pt, dotted] (4lc2) -- node[midway, left, pos=0.5] {$\hat{e}^{l, c_2}_4$} (4lsc2);
            \draw[line width=.7pt, dotted] (4rc2) -- node[midway, right, pos=0.5] {$\hat{e}^{r, c_2}_4$} (4rsc2);  
    
            %child 3 edges
            \draw[line width=.5pt, densely dotted] (3ls) -- node[midway, below, pos=0.5] {$e^{l, c_3}_2$} (2lc3);
            \draw[line width=.5pt, densely dotted] (3rs) -- node[midway, below, pos=0.5] {$e^{r, c_3}_2$} (2rc3);
            \draw[line width=.5pt, dotted] (2lc3) -- node[midway, below, pos=0.5] {$e^{l, c_3}_3$} (3lc3);
            \draw[line width=.5pt, dotted]  (2rc3) -- node[midway, below, pos=0.5] {$e^{r, c_3}_3$} (3rc3);
            \draw[line width=.5pt, dotted]  (3ls) -- node[midway, left, pos=0.5] {$\hat{e}^{l, c_3}_1$} (1lsc3);
            \draw[line width=.5pt, dotted]  (3rs) -- node[midway, right, pos=0.5] {$\hat{e}^{r, c_3}_1$} (1rsc3);  

            %new edges
            \draw[line width=.7pt, dotted] (3lc3) -- node[midway, left, pos=0.5] {$e^{l, c_3}_4$} (4lc3);
            \draw[line width=.7pt, dotted] (3rc3) -- node[midway, right, pos=0.5] {$e^{r, c_3}_4$} (4rc3);
            \draw[line width=.7pt, dotted] (4lc3) -- node[midway, right, pos=0.5] {$e^{l, c_3}_5$} (5lc3);
            \draw[line width=.7pt, dotted] (4rc3) -- node[midway, left, pos=0.5] {$e^{r, c_3}_5$} (5rc3);
            \draw[line width=.7pt, dotted] (5lc3) -- node[midway, right, pos=0.5] {$e^{l, c_3}_6$} (6lc3);
            \draw[line width=.7pt, dotted] (5rc3) -- node[midway, left, pos=0.5] {$e^{r, c_3}_6$} (6rc3);

            \draw[line width=.7pt, dotted] (2lc3) -- node[midway, left, pos=0.5] {$\hat{e}^{l, c_3}_2$} (2lsc3);
            \draw[line width=.7pt, dotted] (2rc3) -- node[midway, right, pos=0.5] {$\hat{e}^{r, c_3}_2$} (2rsc3);
            \draw[line width=.7pt, dotted] (3lc3) -- node[midway, below, pos=0.5] {$\hat{e}^{l, c_3}_3$} (3lsc3);
            \draw[line width=.7pt, dotted] (3rc3) -- node[midway, below, pos=0.5] {$\hat{e}^{r, c_3}_3$} (3rsc3);  
            \draw[line width=.7pt, dotted] (4lc3) -- node[midway, above, pos=0.5] {$\hat{e}^{l, c_3}_4$} (4lsc3);
            \draw[line width=.7pt, dotted] (4rc3) -- node[midway, above, pos=0.5] {$\hat{e}^{r, c_3}_4$} (4rsc3);  
    
            %child 4 edges
            \draw[line width=.4pt, densely dotted] (4ls) -- node[midway, above, pos=0.5] {$e^{l, c_4}_2$} (2lc4);
            \draw[line width=.4pt, densely dotted] (4rs) -- node[midway, above, pos=0.5] {$e^{r, c_4}_2$} (2rc4);
            \draw[line width=.5pt, dotted]  (2lc4) -- node[midway, above, pos=0.5] {$e^{l, c_4}_3$} (3lc4);
            \draw[line width=.5pt, dotted] (2rc4) -- node[midway, above, pos=0.5] {$e^{r, c_4}_3$} (3rc4);
            \draw[line width=.5pt, dotted] (4ls) -- node[midway, left, pos=0.5] {$\hat{e}^{l, c_4}_1$} (1lsc4);
            \draw[line width=.5pt, dotted] (4rs) -- node[midway, right, pos=0.5] {$\hat{e}^{r, c_4}_1$} (1rsc4);  

            \draw[line width=.7pt, dotted] (3lc4) -- node[midway, left, pos=0.5] {$e^{l, c_4}_4$} (4lc4);
            \draw[line width=.7pt, dotted] (3rc4) -- node[midway, right, pos=0.5] {$e^{r, c_4}_4$} (4rc4);
            \draw[line width=.7pt, dotted] (4lc4) -- node[midway, right, pos=0.5] {$e^{l, c_4}_5$} (5lc4);
            \draw[line width=.7pt, dotted] (4rc4) -- node[midway, left, pos=0.5] {$e^{r, c_4}_5$} (5rc4);
            \draw[line width=.7pt, dotted] (5lc4) -- node[midway, right, pos=0.5] {$e^{l, c_4}_6$} (6lc4);
            \draw[line width=.7pt, dotted] (5rc4) -- node[midway, left, pos=0.5] {$e^{r, c_4}_6$} (6rc4);

            \draw[line width=.7pt, dotted] (2lc4) -- node[midway, left, pos=0.5] {$\hat{e}^{l, c_4}_2$} (2lsc4);
            \draw[line width=.7pt, dotted] (2rc4) -- node[midway, right, pos=0.5] {$\hat{e}^{r, c_4}_2$} (2rsc4);
            \draw[line width=.7pt, dotted] (3lc4) -- node[midway, above, pos=0.5] {$\hat{e}^{l, c_4}_3$} (3lsc4);
            \draw[line width=.7pt, dotted] (3rc4) -- node[midway, above, pos=0.7] {$\hat{e}^{r, c_4}_3$} (3rsc4);  
            \draw[line width=.7pt, dotted] (4lc4) -- node[midway, below, pos=0.5] {$\hat{e}^{l, c_4}_4$} (4lsc4);
            \draw[line width=.7pt, dotted] (4rc4) -- node[midway, below, pos=0.5] {$\hat{e}^{r, c_4}_4$} (4rsc4);  
        \end{tikzpicture}
        \caption{Considered instance of a graph of maximum degree four.}
        \label{fig fractional deg 4}
    \end{figure}

    Furthermore, given the sequence of arrivals, the maximum matching cardinality of the matching in the graph increases. The increase of the maximum matching cardinality is as follows: the arrival of $B_1$ and $B_2$ increase the cardinality by $1$ each, the arrival of  $B_3$, \ldots, $B_6$ increase the cardinality by $2$, $B_7$, \ldots, $B_{10}$ increase the cardinality by $1$, $B_{11}$, \ldots, $B_{14}$ increase by $1$, and $B_{15}$, \dots, $B_{30}$ increase the cardinality by $2$.

    Let $\gamma$ denote the guarantee achieved on graphs with maximum degree four. Let $\mu_i$ denote the cardinality of the maximum matching after the arrival of batch $B_i$ for $i = 1, \dots, 30$. Due to the above discussion, we have $\mu_1=1$, $\mu _2=2$; we have $\mu_3=4$,\ldots, $\mu _6=10$; we have $\mu_7=11$,\ldots, $\mu _{14}=18$; we have $\mu_{15}=20$,\ldots, $\mu _{30}=50$  
    Consider the following Linear Program to determine an upper bound on~$\gamma$.
    \begin{maxi*}
        {}{\gamma}{}{}
        \addConstraint{\sum_{e \in \bigcup_{j = 1}^i B_j} y_e \geq \gamma\cdot \mu_i}{}{\mkern40mu \text{for all} \ i = 1 \dots, 30}
        \addConstraint{\sum_{e \in \delta(u)} y_e \leq 1}{}{\mkern40mu \text{for all}\ u \in V}
    \end{maxi*}
    
     Solving the above Linear Program, we obtain $\gamma \leq 0.58884 < c \approx 0.5914$ as required. 

We note that the intent of Theorem~\ref{thm: fractional alg deg 4} is not to optimize the bound on bipartite graphs of maximum degree four but to provide a gap on the guarantees achievable for graphs of maximum degree three and four. In particular, by generalizing the instance in Figure~\ref{fig fractional deg 4} by increasing the number of rounds, and treating the $f_i$ edges as $e_1$ for a recursive process, one can improve upon this bound.

\section{Open Questions}

Let us point to further directions and open questions related to our work. The work of~\cite{Buchbinder} shows that the best possible guarantee of an online algorithm equals $2/3$ for graphs of maximum degree two, both for integral and fractional matchings in both vertex arrival and edge arrival models. The upper bound from~\cite{Buchbinder} and our work show that the possible guarantee of an online algorithm for fractional matchings equals $c=4/(9-\sqrt 5)\approx 0.5914$ for graphs of maximum degree three in both vertex arrival and edge arrival models. Our work leads to the following open question: What is the smallest value $d$ such that online algorithms for fractional matchings achieve different best possible guarantees in vertex and edge arrival models for graphs of maximum degree $d$?

Note that our algorithm achieves the guarantee $c=4/(9-\sqrt 5)\approx 0.5914$ for fractional matchings in both bipartite and non-bipartite graphs. Also, the work of~\cite{Buchbinder} shows that $c$ is an upper bound on the guarantee of any online algorithm for bipartite graphs. In general, for fractional matchings, the construction in~\cite{Gamlath} shows that the best possible guarantee is $1/2$ for both bipartite and non-bipartite graphs in the edge arrival model.  Is there $d$ such that online algorithms for fractional matchings achieve different best possible guarantees in the edge arrival model for bipartite and non-bipartite graphs of maximum degree $d$? If the answer is positive, then what is the smallest such $d$? Apart from the maximum degree, what other parameters of the underlying graphs have a crucial role in the difference of guarantees in bipartite and non-bipartite graphs? 
Is it possible to obtain the guarantee $c$ for integral matchings in the case of bipartite graphs with maximum degree three?

\subsection*{Acknowledgments.} The research of Thomas Snow was supported by NSERC Undergraduate Student Research Award (USRA) and Math Undergraduate Research Award (MURA). The research of Kanstantsin Pashkovich was supported by NSERC Discovery Grants Program RGPIN-2020-04346. The authors are grateful to Ricardo Fukasawa for many helpful discussions. Kanstantsin Pashkovich is grateful to Laurent Poirrier for the help with encoding linear programs to obtain upper bounds.

\bibliographystyle{alpha}
\bibliography{bibliography.bib}

\newpage

% Appendix
\begin{appendix}

\section{Properties of Values in Consistent Instances} \label{appendix: fractional matching in consistent instances}
In our analysis, we need the following facts. 

\begin{lemma} \label{lma fractional matching consistent instances} The following properties hold:
    \begin{enumerate} 
        \item \label{lma 1.1} for all natural $n$, $n\geq 4$ we have $2\widetilde{y}_{n + 1} - 2\widetilde{y}_n = (4c - 2)F_{n-1} - cF_{n-2}$.

        \item \label{lma 1.2}
   for all natural $n$, $n\geq 4$ we have $2\widetilde{y}_{n+1} = 2\widetilde{y}_n - c\psi^{n-1} = 2\widetilde{y}_n + c(-1)^n \phi^{1-n}$.

   \item \label{lma 1.3}
    the subsequence $(\widetilde{y}_{2k})_{k \in \mathbb{N}}$ of $(\widetilde{y}_k)_{k \in \mathbb{N}}$ is a strictly increasing sequence.

    \item \label{lma 1.4}   the subsequence $(\widetilde{y}_{2k+1})_{k \in \mathbb{N}}$ of $(\widetilde{y}_k)_{k \in \mathbb{N}}$ is a strictly decreasing sequence.    

    \item \label{lma 1.5}
    for every natural $k$ and $n$ we have that $\widetilde{y}_{2n+1} > \widetilde{y}_{2k}$.

    \item \label{lma 1.6}
    for all natural $n$ we have $\widetilde{y}_n + 2\widetilde{y}_{n+1} < 1 + \frac{c}{2}$.

    \item \label{lma 1.7}
    for all natural $n$ we have $1 - \widetilde{y}_n - \widetilde{y}_{n+1} = c - \widetilde{y}_{n+2}$.
    \end{enumerate}
\end{lemma}

Note that~\eqref{lma 1.7} in the above lemma is used throughout our work. In particular, we use both $1 - \widetilde{y}_n - \widetilde{y}_{n+1}$ and $c - \widetilde{y}_{n+2}$ interchangeably.

For the sake of the proof let us provide the first six values of $\widetilde{y}$, $\widetilde{y}_1 = c\approx 0.5914$, $\widetilde{y}_2 = \frac{c}{2}\approx 0.2957$, $\widetilde{y}_3 = \frac{5c - 2}{2}\approx 0.4784$, $\widetilde{y}_4 = 4c - 2\approx 0.3655$, $\widetilde{y}_5 = \frac{15c - 8}{2}\approx 0.4353$, and $\widetilde{y}_6 = \frac{25c}{2} - 7\approx 0.3921$.
\begin{proof}[Proof of part~\eqref{lma 1.1} in Lemma~\ref{lma fractional matching consistent instances}]
    Let $n$ be a natural number with $n\geq 4$, then we have
    \begin{align*}
        2(\widetilde{y}_{n + 1} - \widetilde{y}_n) &= 3cF_{n+1} + cF_{n-1} - 2c -2F_{n+1} + 2 - 
        3cF_n - cF_{n-2} \\
        & \qquad + 2c + 2F_n - 2 \\
        &= 3cF_n + 3cF_{n-1} + cF_{n-1} - 2F_n - 2F_{n-1} - 3cF_n - cF_{n-2} + 2F_n \\
        &= (4c - 2)F_{n-1} - cF_{n-2} \,,
    \end{align*}
    where the first equality follows from the definition of $\widetilde{y}_n$, $\widetilde{y}_{n+1}$, and the second equality follows from the property $F_n = F_{n-1} + F_{n-2}$ of the Fibonacci numbers.
    
\end{proof}

\begin{proof}[Proof of part~\eqref{lma 1.2} in Lemma~\ref{lma fractional matching consistent instances}]
    By part~\eqref{lma 1.1} in Lemma~\ref{lma fractional matching consistent instances} we have
    \begin{align*}
        2\widetilde{y}_{n+1} &= 2\widetilde{y}_n + (4c-2)F_{n-1} - cF_{n-2} \\
        &= 2\widetilde{y}_n + c(4F_{n-1} - F_{n-2}) - 2F_{n-1} \\
        &= 2\widetilde{y}_n + c(4\frac{\phi^{n-1} - \psi^{n-1}}{\sqrt{5}} - \frac{\phi^{n-2} - \psi^{n-2}}{\sqrt{5}}) - 2F_{n-1} \\
        &= 2\widetilde{y}_n + c(\frac{\phi^{n-1}(4 + \psi) - \psi^{n-1}(4 + \psi + \sqrt{5})}{\sqrt{5}}) - 2F_{n-1} \\
        &= 2\widetilde{y}_n + c(F_{n-1}(4 + \psi) - \psi^{n-1}) - 2F_{n-1} \qquad\quad(\star) \\
        &= 2\widetilde{y}_n - c\psi^{n-1} \qquad\qquad\qquad\qquad\qquad\qquad\quad(\star\star)\\
        &= 2\widetilde{y}_n - c(-\frac{1}{\phi})^{n-1} \\
        &= 2\widetilde{y}_n - c(-1)^{n-1}\phi^{1-n} \\
        &= 2\widetilde{y}_n + c(-1)^n \phi^{1-n} \,,
    \end{align*}
    where $(\star\star)$ follows from $(\star)$ as $4c + \psi c - 2 = 0$, which is easily verifiable through computation.
    
\end{proof}

\begin{proof}[Proof of parts~\eqref{lma 1.3} and~\eqref{lma 1.4} in Lemma~\ref{lma fractional matching consistent instances}]

    Both statements are corollaries of \eqref{lma 1.2} in Lemma~\ref{lma fractional matching consistent instances}.
    
\end{proof}

\begin{proof}[Proof of part~\ref{lma 1.5} in Lemma~\ref{lma fractional matching consistent instances}]
    We consider the following two cases. \\
    \underline{Case 1:} $k \leq n$ \\
    \begin{align*}
        2\widetilde{y}_{2n+1} &= 2\widetilde{y}_{2n} + c(-1)^{2n}\phi^{1-2n} \\
        &= 2\widetilde{y}_{2n} + c\phi^{1-2n} \\
        &> 2\widetilde{y}_{2n} \\
        &> 2\widetilde{y}_{2k}\,,
    \end{align*}
     where the final inequality follows from the fact that $k \leq n$ and by~\eqref{lma 1.3} in Lemma~\ref{lma fractional matching consistent instances} $(\widetilde{y}_{2m})_{m \in \mathbb{N}}$ is a strictly increasing sequence, hence $2\widetilde{y}_{2n} > 2\widetilde{y}_{2k}$. \\
    \underline{Case 2:} $k > n$ \\
    \indent By~\eqref{lma 1.4} in Lemma~\ref{lma fractional matching consistent instances} $(\widetilde{y}_{2m+1})_{m \in \mathbb{N}}$ is a strictly decreasing sequence and hence, as $k > n$ we have that $2\widetilde{y}_{2n + 1} > 2\widetilde{y}_{2k + 1}$. So, \\
    \begin{align*}
        2\widetilde{y}_{2n+1} &> 2\widetilde{y}_{2k + 1} \\
        &= 2\widetilde{y}_{2k} + c(-1)^{2k}\phi^{1-2k} \\
        &= 2\widetilde{y}_{2k} + c\phi^{1-2k} \\
        &> 2\widetilde{y}_{2k}
    \end{align*}
    So, for all natural $k$ and $n$ we have that $\widetilde{y}_{2n+1} > \widetilde{y}_{2k}$ as required.
    
\end{proof}

\begin{proof}[Proof of part~\ref{lma 1.6} in Lemma~\ref{lma fractional matching consistent instances}]
    It is not hard to see that for $n \leq 6$ the statement holds. So assume $n \geq 7$.
    \begin{align*}
        2\widetilde{y}_n + 4\widetilde{y}_{n+1} &= 2\widetilde{y}_n + 4\widetilde{y}_n + 2c(-1)^n\phi^{1-n} \\
        &= 6\widetilde{y}_n + 2c(-1)^n\phi^{1-n} \\
        &\leq 6\widetilde{y}_n + 2c\phi^{1-n} \\
        &\leq 6\widetilde{y}_7 + 2c\phi^{1-7} \\
        &< 3\frac{21}{25} + 2c(\frac{1 + \sqrt{5}}{2})^{-6} \\
        &< 2 + c\,,
    \end{align*}
    where the second inequality follows from~\eqref{lma 1.4} and~\eqref{lma 1.5} in Lemma~\ref{lma fractional matching consistent instances}. So it follows that $y_n \leq y_7$ for all $n \geq 7$ as required.
    
\end{proof}

\begin{proof}[Proof of part~\ref{lma 1.7} in Lemma~\ref{lma fractional matching consistent instances}]
    We proceed by induction on $n$. The case where $n=1$ holds and is easily verifiable. So assume the statement holds for arbitrary natural $n$ that is, $1 - \widetilde{y}_n - \widetilde{y}_{n+1} + \widetilde{y}_{n+2} = c$. So by the inductive hypothesis and \eqref{lma 1.2} in Lemma~\ref{lma fractional matching consistent instances} we have,
    \begin{align*}
        1 - &\widetilde{y}_{n+1} - \widetilde{y}_{n+2} + \widetilde{y}_{n+3} \\
        &= c + \widetilde{y}_n - 2\widetilde{y}_{n+2} + \widetilde{y}_{n+3} \\
        &= c + \widetilde{y}_n - \widetilde{y}_{n+2} + \frac{1}{2}c(-1)^{n+2}\phi^{-n-1} \\
        &=c + \widetilde{y}_n + \frac{1}{2}c(-1)^{n+2}\phi^{-n-1} - \widetilde{y}_{n+1} - \frac{1}{2}c(-1)^{n+1}\phi^{-n} \\
        &=c + \frac{1}{2}c(-1)^{n+2}\phi^{-n-1} - \frac{1}{2}c(-1)^{n+1}\phi^{-n} - \frac{1}{2}c(-1)^{n}\phi^{-n + 1} \\
        &=c + \frac{1}{2}c(-1)^{n+2}\phi^{-n-1} + \frac{1}{2}c(-1)^{n+2}\phi^{-n} + \frac{1}{2}c(-1)^{n+1}\phi^{-n + 1} \\
        &= c + \frac{1}{2}c(-1)^{n+1}\phi^{-n - 1}(\phi^2 - \phi - 1) \\
        &= c\,,
    \end{align*}
    where the last statement holds as $\phi^2 - \phi - 1 = 0$ a well known identity of the golden ratio.
    
\end{proof}
\newpage

\section{Proof of Properties in Main Lemmas} \label{sec: proof properties}
Let us now prove both Lemma~\ref{lma main lemma} and Lemma~\ref{lma main lemma additional}. We prove the statement by induction on our time point. The base case, i.e. the case when no edge has arrived, can be verified in a straightforward way. Let us assume that the statement, i.e. all properties \ref{P1}, \dots, \ref{P6} hold before the arrival of an edge $e = uv$, where $\degree(u) \geq \degree(v)$. 
    
    Let $G$ be the graph after the arrival of $e$ and let $x^{old}$ be produced by the algorithm immediately before the arrival of edge $e$, i.e.  $x^{old}$ and $y$ satisfy \ref{P1}, \ldots, \ref{P6} with respect to the graph $G \setminus e$. Let $x^{new}_u$ and $x^{new}_v$ be the values assigned to $x_u$ and $x_v$ immediately after the edge $e$ arrived.

    We now consider the following case study:
    \begin{enumerate}[label=\arabic*.]
        \item \label{Case 1} $\type(e)=1$
        \begin{enumerate}[label*=\roman*.]
            \item\label{case: main 1.1} $n_e=1$
            \item \label{case: main 1.2} $n_e>1$
        \end{enumerate}
        \item \label{case: main 2} $\type(e)=2$
        \item $\type(e)=3$
        \begin{enumerate}[label*=\roman*.]
            \item \label{case: main 3.1} $\{\type(\delta(u) \setminus \{e\}), \type(\delta(v) \setminus \{e\})\} = \{(1,0,0)\}$
            \item \label{case: main 3.2} $\{\type(\delta(u) \setminus \{e\}), \type(\delta(v) \setminus \{e\})\} = \{(1,0,0), (0,1,0)\}$
            \item \label{case: main 3.3} $\{\type(\delta(u) \setminus \{e\}), \type(\delta(v) \setminus \{e\})\} = \{(0,2,0), (1,0,0)\}$
            \item \label{case: main 3.4} $\{\type(\delta(u) \setminus \{e\}), \type(\delta(v) \setminus \{e\})\} = \{(1,1,0), (1,0,0)\}$
        \end{enumerate}
    \end{enumerate}

\textbf{Case~\ref{case: main 1.1}: $\type(e) = 1$ and $n_e=1$.}

By line~\ref{line: test spoke edge}, we have $\degree(z(e)) < 3$ or $\type(\delta(z(e)) \setminus \{e\}) \in \{(0,2,0), (1,1,0)\}$.
 
\textbf{Checking~\ref{P1}.} In this case, the assignment of new values $x^{new}_{z(e)}$, $x^{new}_{w(e)}$ and $y_e$ happens in lines \ref{line: x path edge z}, \ref{line: x path edge w} and \ref{line: y path edge}, so it is straightforward to check that the property~\ref{P1} holds for $G$. 

\textbf{Checking~\ref{P2}.} Note, that if in line~\ref{line: y path edge} $y_e$ is assigned the value $\widetilde{y}_{n_e}$, i.e. the value $\widetilde{y}_{1}=c$, then the property~\ref{P2} holds in $G$, since
\[x^{new}_{z(e)} + x^{new}_{w(e)} = x^{old}_{z(e)} + x^{old}_{w(e)} + y_e\geq y_e\,.\]
Let us consider the case when $y_e$ is assigned the value smaller than $\widetilde{y}_{n_e}$, i.e. $c=\widetilde{y}_1> 1 - \sum_{f \in \delta(z(e)) \setminus \{e\}} y_f$. Thus, $\degree(z(e))=2$ or $\degree(z(e))=3$. Due to $n_e=1$ the first entry in $\type(\delta(z(e)))$ equals $1$, and so we have only two possible cases $\type(\delta(z(e)))=(1,1,0)$ and $\type(\delta(z(e)))=(1,2,0)$ by Table~\ref{tab: edge-partitions}. In the case when $\type(\delta(z(e)))=(1,1,0)$, by property~\ref{P3}, we have that $y_f \leq 1-c$ for $f \in \delta(z(e)) \setminus\{e\}$ and thus, $y_e$ is assigned the value $\widetilde{y}_{n_1}=c$ and~\ref{P2} holds due to the same arguments as above.
In the case when $\type(\delta(z(e)))=(1,2,0)$, we have
\begin{align*}
    x^{new}_{z(e)}=\left(x^{old}_{z(e)} + y_e - (c-\widetilde{y}_2)\right)=x^{old}_{z(e)}+\left(1-\sum_{f\in \delta(z(e)),\, f\neq e}y_f\right)-c/2=1-c/2\,,
\end{align*}
where the first and second equalities are due to assignments in Algorithm~\ref{alg: one} and that $\widetilde{y}_2=c/2$, and the third equality is \eqref{ob: 8} in Observation~\ref{obs vertex neighborhood}.
Similarly, we have
\begin{align*}
    x^{new}_{w(e)} = x^{old}_{w(e)} + c/2\,.
\end{align*}
Thus, we have
\begin{align*}
x^{new}_{z(e)} + x^{new}_{w(e)} = 1 + x^{old}_{w(e)} \geq c\,.
\end{align*}
For every $f\in \delta(z(e))$, $f\neq e$ we have
\[x^{new}_{z(f)}+x^{new}_{w(f)}=(c-y_f)+x^{new}_{z(e)}\geq \left(c-(1-c)\right)+1-c/2=3c/2\,,\]
where the first equality holds as $w(f) = z(e)$ and $x^{new}_{z(f)} = x_{z(f)}^{old} = c - y_f$ by \ref{P4} and line~\ref{line: y spoke edge} in Algorithm~\ref{alg: one}, and the inequality is due to \ref{P3}. Since $x^{new}_{w(e)}\geq x^{old}_{w(e)}$ and $x^{new}_a+x^{new}_b\geq c$ for every $ab=f\in \delta(z(e))$, we get that $x^{new}$ satisfies~\ref{P2}.

\textbf{Checking~\ref{P3}.} The values of $x^{new}$ and $x^{old}$ vary only for $z(e)$ and $w(e)$. If $\degree(z(e)) = 2$, then by $n_e = 1$ we have that $\type(\delta(z(e))) = (1,1,0)$. Similarly, if $\degree(w(e)) = 2$ then by \eqref{obs 2.1} in Observation~\ref{obs edge types} and Table~\ref{tab: edge-partitions}, we have $\type(\delta(w(e))) = (1,1,0)$. Thus, both $w(e)$ and $z(e)$ do not satisfy the premise of~\ref{P3}, and so \ref{P3} continues to hold by the inductive hypothesis.

\textbf{Checking~\ref{P4}.} Here, we follow the same arguments as we used for verifying~\ref{P2} above. We have $x^{new}_{w(e)}=x^{old}_{w(e)}+c/2$, so the value of $x_{w(e)}$ is non-decreasing. We also have that $y_e=c$ unless $\type(\delta(z(e))\setminus e)=(0,2,0)$ as shown above when checking \ref{P1}; hence $x^{new}_{z(e)}=x^{old}_{z(e)}+c/2$ unless $\type(\delta(z(e))\setminus e)=(0,2,0)$. This shows that~\ref{P4} holds upon the arrival of $e$. 

\textbf{Checking~\ref{P5}.} Here, we follow the same arguments as we used for verifying~\ref{P2} above. If $y_e = c$ then $x^{new}_{z(e)} \geq \frac{c}{2} = \widetilde{y}_{n_e + 1}$. If $y_e < c$ then $x^{new}_{z(e)} = 1 - \frac{c}{2} \geq \frac{c}{2} = \widetilde{y}_{n_e + 1}$.
In all cases, $x^{new}_{w(e)} \geq \frac{c}{2} = c - \widetilde{y}_{n_e + 1}$ which holds with equality if $\degree(w(e)) = 1$. So \ref{P5} holds.

\textbf{Checking~\ref{P6}.} Here, we again follow the same arguments as we used for verifying~\ref{P2} above. We have $x^{new}_{w(e)}=x^{old}_{w(e)}+c/2$, so the value of $x_{w(e)}$ is non-decreasing. Finally, if $x^{new}_{z(e)}<x^{old}_{z(e)}$ then we have $x^{new}_{z(e)} = 1 - c/2 \geq c \geq 1 - c$; and so hence~\ref{P6} holds inductively by~\ref{P3}.

\textbf{Checking~\ref{P7}.} This holds straightforwardly with respect to $z(e)$ by the assignment done in line~\ref{line: y path edge} in Algorithm~\ref{alg: one}. We now consider $w(e)$. First, note that if $\degree(w(e)) = 3$ then this follows by choice of $w(e)$ in line~\ref{line: choice of z}, so we may assume $\degree(w(e)) < 3$. So, by \eqref{obs 2.1} in Observation~\ref{obs edge types} we have that for $f \in \delta(w(e)) \setminus \{e\}$, $\type(f) \neq 1$; furthermore, as $\degree_{G \setminus e}(w(e)) < 2$ we have by \eqref{ob: 3} in Observation~\ref{obs edge types} that $\type(f) \neq 3$. So $\type(f) = 2$ and hence inductively by \ref{P3} we have that $y_f \leq 1 - c$; therefore, $\sum_{f \in \delta(w(e))} y_f \leq y_e + 1 - c \leq c + 1 - c = 1$s
where the first inequality is due to $\degree(w(e)) \in \{1,2\}$, the second inequality holds as $y_e \leq c$. So \ref{P7} holds as required. \\

\textbf{Case~\ref{case: main 1.2}: $\type(e) = 1$ and $n_e > 1$.}
By line~\ref{line: test spoke edge}, we have $\degree(z(e)) < 3$ or $\type(\delta(z(e)) \setminus \{e\}) \in \{(0,2,0), (1,1,0)\}$. However, as $n_e \neq 1$, by line~\ref{line: def of n} we have that there exists $ \ f_p \in \delta(z) \setminus \{e\}$ with $\type(f_p) = 1$, hence $\type(\delta(z(e)) \setminus \{e\}) \neq (0,2,0)$. Moreover, by \eqref{ob: 6} in Observation~\ref{obs special endpoints} and line~\ref{line: def of n} in Algorithm~\ref{alg: one} we have $n_e = n_{f_p} + 1$.

\textbf{Checking~\ref{P1}.} In this case, the assignment of new values $x^{new}_{z(e)}$, $x^{new}_{w(e)}$ and $y_e$ happens in lines \ref{line: x path edge z}, \ref{line: x path edge w} and \ref{line: y path edge} of Algorithm~\ref{alg: one}, so it is straightforward to check that the property~\ref{P1} holds for $G$. 

\textbf{Checking~\ref{P2}.} Note, that if in line~\ref{line: y path edge} of Algorithm~\ref{alg: one}, $y_e$ is assigned the value $\widetilde{y}_{n_e}$, then the property~\ref{P2} holds in $G$, since
\[x^{new}_{z(e)} + x^{new}_{w(e)} = x^{old}_{z(e)} + x^{old}_{w(e)} + y_e \geq c - \widetilde{y}_{n_e} + y_e = c\,,\]
where the inequality holds as $x^{old}_{z(e)} \geq c - \widetilde{y}_{n_{f_p} + 1} = c - \widetilde{y}_{n_e}$ inductively by \ref{P5} on $f_p$ and \eqref{remark: 2} in Observation~\ref{obs special endpoints}.

Let us consider the case when $y_e < \widetilde{y}_{n_e}$, i.e. $\widetilde{y}_{n_e} > 1 - \sum_{f \in \delta(z(e)) \setminus \{e\}} y_f$. Thus, $\degree(z(e))=2$ or $\degree(z(e))=3$. Due to $n_e > 1$ the first entry in $\type(\delta(z(e)))$ equals $2$, and so we have only two possible cases $\type(\delta(z(e)))=(2,0,0)$ and $\type(\delta(z(e)))=(2,1,0)$ by Table~\ref{tab: edge-partitions}. 

In the case when $\type(\delta(z(e)))=(2,0,0)$, by line~\ref{line: y path edge} in Algorithm~\ref{alg: one} we have that for $f \in \delta(z(e)) \setminus \{e\}$, $y_f \leq \widetilde{y}_{n_f} = \widetilde{y}_{n_e - 1}$. So by \eqref{ob: 3c/2 bound} in Observation~\ref{obs fractional matching consistent instances} we have that $\widetilde{y}_{n_e} + \widetilde{y}_{n_e - 1} <= 3c/2 < 1$ thus, we have that $y_e$ is assigned the value $\widetilde{y}_{n_e}$; and so~\ref{P2} holds due to the same arguments as above.

In the case when $\type(\delta(z(e)))=(2,1,0)$, let $f_s \in \delta(z(e))$ with $\type(f_s) = 2$, hence we have
\begin{equation} \label{case 1.ii eq 1}
    \begin{split}
    x^{new}_{z(e)} &= \left(x^{old}_{z(e)} + y_e - (c-\widetilde{y}_{n_e + 1})\right) \\ 
    &= c - \widetilde{y}_{n_{f_p} + 1} + y_{f_s} + \left(1-\sum_{f\in \delta(z(e)),\, f\neq e}y_f\right) - (c - \widetilde{y}_{n_e + 1}) \\ 
    &= 1 - y_{f_p} - \widetilde{y}_{n_e} + \widetilde{y}_{n_e + 1} \\
    &\geq 1 - \widetilde{y}_{n_e - 1} - \widetilde{y}_{n_e} + \widetilde{y}_{n_e + 1} \\
    &= c\,,
    \end{split}
\end{equation}
where the first equality is due to assignments in Algorithm~\ref{alg: one}, the second equality is due to Claim~\ref{claim 1}, the inequality is due to \eqref{ob: 4} in Observation~\ref{obs edge types} and $n_e = n_{f_p} + 1$, and the final equality holds by~\eqref{lma 1.7} in Lemma~\ref{lma fractional matching consistent instances}. Thus, we have,
\begin{align*}
x^{new}_{z(e)} + x^{new}_{w(e)} \geq c \,.
\end{align*}
Similarly, we have
\begin{align*}
    x^{new}_{w(e)} = x^{old}_{w(e)} + (c - \widetilde{y}_{n_e + 1}) \geq x^{old}_{w(e)}\,,
\end{align*}
where the inequality holds by~\eqref{ob: 1} in Observation~\ref{obs fractional matching consistent instances}.
So as $x^{new}_{z(e)} \geq c$ and $x^{new}_{w(e)} \geq x^{old}_{w(e)}$ we get that $x^{new}$ satisfies~\ref{P2}.

\textbf{Checking~\ref{P3}.} The values of $x^{new}$ and $x^{old}$ vary only for $z(e)$ and $w(e)$. If $\degree(z(e))$ equals $2$, then by $n_e > 1$ we have that $\type(\delta(z(e))) = (2,0,0)$, hence by the same argument as earlier, by line~\ref{line: y path edge} in Algorithm~\ref{alg: one} and \eqref{ob: 3c/2 bound} in Observation~\ref{obs fractional matching consistent instances} we have that $y_e = \widetilde{y}_{n_e}$.

So inductively by \ref{P5} and line~\ref{line: x path edge z} we have that,
\[x^{new}_{z(e)} = (c - \widetilde{y}_{n_e}) + \widetilde{y}_{n_e} - (c - \widetilde{y}_{n_e + 1}) = \widetilde{y}_{n_e + 1} \in [\frac{c}{2}, \frac{5c - 2}{2}]\,,\]
where the inclusion is due to~\eqref{ob: 1} in Observation~\ref{obs fractional matching consistent instances}, and as $\frac{c}{2} > 2c - 1$ we have that \ref{P3} holds with respect to $z(e)$. Moreover,
\[x^{new}_{z(e)} = \widetilde{y}_{n_e + 1} = c - 1 + \widetilde{y}_{n_e - 1} + \widetilde{y}_{n_e} \geq  c - 1 + y_e + y_{f_p} \,,\]
where the second equality holds by~\eqref{lma 1.7} in Lemma~\ref{lma fractional matching consistent instances} and the inequality holds by \eqref{ob: 4} in Observation~\ref{obs edge types} and as $y_e = \widetilde{y}_{n_e}$. 
Similarly, if $\degree(w(e)) = 2$ then by~\eqref{obs 2.1} in Observation~\ref{obs edge types} and Table~\ref{tab: edge-partitions}, we have $\type(\delta(w(e))) = (1,1,0)$. Thus, $w(e)$ does not satisfy the premise of~\ref{P3}, and so \ref{P3} continues to hold with respect to $w(e)$.

\textbf{Checking~\ref{P4}.} Here, we follow the same arguments as we used for verifying~\ref{P2} above. We have $x^{new}_{w(e)} = x^{old}_{w(e)} + (c - \widetilde{y}_{n_e + 1}) \geq x^{old}_{w(e)}$, where the inequality holds by~\eqref{ob: 1} in Observation~\ref{obs fractional matching consistent instances}. So the value of $x_{w(e)}$ is non-decreasing.

We also have that $y_e = \widetilde{y}_{n_e}$ unless $\type(\delta(z(e))\setminus e)=(1,1,0)$; hence $x^{new}_{z(e)} = x^{old}_{z(e)} + y_e - (c - \widetilde{y}_{n_e + 1}) \geq x^{old}_{z(e)}$ unless $\type(\delta(z(e))\setminus e)=(1,1,0)$, where the inequality holds by~\eqref{ob: 1} and \eqref{ob: 2} in Observation~\ref{obs fractional matching consistent instances}. This shows that~\ref{P4} holds upon the arrival of $e$. 

\textbf{Checking~\ref{P5}.} Here, we follow the same arguments as we used for verifying~\ref{P2} above. If $y_e = \widetilde{y}_{n_e}$ then inductively by \ref{P5} and line~\ref{line: x path edge z} we have that 
\[x^{new}_{z(e)} \geq (c - \widetilde{y}_{n_e}) + \widetilde{y}_{n_e} - (c - \widetilde{y}_{n_e + 1}) = \widetilde{y}_{n_e + 1}\,.\]
If $y_e < \widetilde{y}_{n_e}$ then by \eqref{case 1.ii eq 1} we have $x^{new}_{z(e)} \geq c > \widetilde{y}_{n_e + 1}$, where the inequality is due to~\eqref{ob: 1} in Observation~\ref{obs fractional matching consistent instances}.
In all cases, by line~\ref{line: x path edge w} we have $x^{new}_{w(e)} \geq (c - \widetilde{y}_{n_e + 1})$ which holds with equality if $\degree(w(e)) = 1$. So \ref{P5} holds.

\textbf{Checking~\ref{P6}.} Here, we again follow the same arguments as we used for verifying~\ref{P2} above. We have $x^{new}_{w(e)} = x^{old}_{w(e)} + (c - \widetilde{y}_{n_e + 1})$, so the value of $x_{w(e)}$ is non-decreasing and hence inductively \ref{P6} holds with respect to $w(e)$. Finally, if $x^{new}_{z(e)} < x^{old}_{z(e)}$ then we have $x^{new}_{z(e)} \geq c \geq 1 - c$; and so~\ref{P6} holds inductively by~\ref{P3}.

\textbf{Checking~\ref{P7}.} This holds straightforwardly with respect to $z(e)$ by the assignment done in line~\ref{line: y path edge} in Algorithm~\ref{alg: one}. We now consider $w(e)$. First, note that if $\degree(w(e)) = 3$ then this follows by choice of $w(e)$ in line~\ref{line: choice of z}, so we may assume $\degree(w(e)) < 3$. So, by \eqref{obs 2.1} in Observation~\ref{obs edge types} we have that for $f \in \delta(w(e)) \setminus \{e\}$, $\type(f) \neq 1$; furthermore, as $\degree_{G \setminus e}(w(e)) < 2$ we have by \eqref{ob: 3} in Observation~\ref{obs edge types} that $\type(f) \neq 3$. So $\type(f) = 2$ and hence inductively by \ref{P3} we have that $y_f \leq 1 - c$; therefore, $\sum_{f \in \delta(w(e))} y_f \leq y_e + 1 - c \leq \widetilde{y}_{n_e} + 1 - c \leq 1$
where the first inequality is due to $\degree(w(e)) \in \{1,2\}$, the second inequality holds as $y_e \leq \widetilde{y}_{n_e}$, and the last inequality is due to~\eqref{ob: 1} in Observation~\ref{obs fractional matching consistent instances}. So \ref{P7} holds as required. \\

\textbf{Case~\ref{case: main 2}: $\type(e) = 2$.}
By line~\ref{line: test spoke edge}, we have $\degree(z(e)) = 3$ and $\type(\delta(z) \setminus \{e\}) \notin \{(0,2,0), (1,1,0)\}$.

\textbf{Checking~\ref{P1}.} In this case, the assignment of new values $x^{new}_{z(e)}$, $x^{new}_{w(e)}$ and $y_e$ happens in lines \ref{line:  spoke cover change} and \ref{line: y spoke edge}, so it is straightforward to check that the property~\ref{P1} holds for $G$. 

\textbf{Checking~\ref{P2}.} By line~\ref{line: y spoke edge} we have that $y_e = c - x^{old}_{z(e)}$ hence,
\[x^{new}_{z(e)} + x^{new}_{w(e)} = x^{old}_{z(e)} + x^{old}_{w(e)} + y_e \geq x^{old}_{z(e)} + c - x^{old}_{z(e)} = c\,,\]
where the first equality holds by line~\ref{line:  spoke cover change}. Moreover by line~\ref{line:  spoke cover change} we have that $x^{new}_{z(e)} = x^{old}_{z(e)}$ and $x^{new}_{w(e)} = x^{old}_{w(e)} + y_e > x^{old}_{w(e)}$ where the inequality holds as $y_e \in [1 - \frac{3c}{2}, 1 - c]$ inductively by \ref{P3}. So we have that $x^{new}$ satisfies~\ref{P2}.

\textbf{Checking~\ref{P3}.} First, by line~\ref{line: test spoke edge}, we have $\degree(z(e)) = 3$ so $z(e)$ does not satisfy the premise of \ref{P3}. If $\degree(w(e)) = 2$ then by Table~\ref{tab: edge-partitions} we have that $\type(\delta(w(e))) \in \{(1,1,0), (0,2,0)\}$ hence $w(e)$ does not satisfy the premise of \ref{P3}. So \ref{P3} holds inductively.

\textbf{Checking~\ref{P4}.} Similarly to the proof of \ref{P2}, we have by line~\ref{line:  spoke cover change} that $x^{new}_{z(e)} = x^{old}_{z(e)}$ and $x^{new}_{w(e)} = x^{old}_{w(e)} + y_e > x^{old}_{w(e)}$, hence \ref{P4} holds inductively as all other values of $x^{new}$ remain unchanged.

\textbf{Checking~\ref{P5}.} Following the proof of \ref{P4} above, we have that $x^{new}_a \geq x^{old}_a$ for all $a \in V$ hence, inductively \ref{P5} holds with respect to $x^{new}$.

\textbf{Checking~\ref{P6}.} By line~\ref{line:  spoke cover change} we have that $x^{new}_{w(e)} = x^{old}_{w(e)} + y_e \geq y_e$ and as all other values of $x^{new}$ remain unchanged we have that \ref{P6} holds inductively.

\textbf{Checking~\ref{P7}.} By line~\ref{line: test spoke edge} we have that $\type(\delta(z(e)) \setminus \{e\}) \notin \{(1,1,0), (0,2,0)\}$ and $\degree_{G \setminus e}(z(e)) = 2$, so inductively by \ref{P3} we have that $x^{old}_{z(e)} \geq c - 1 + \sum_{f \in \delta(z(e)) \setminus \{e\}} y_f$. Therefore, as $y_e = c - x^{old}_{z(e)}$ we have that $\sum_{f \in \delta(z(e)} y_f \leq 1$ as required. Now, if $\degree(w(e)) = 3$ then $\sum_{f \in \delta(w(e))} y_f \leq 1$ by choice of $w(e)$ in line~\ref{line: choice of z}. If $\degree(w(e)) = 2$ then by Table~\ref{tab: edge-partitions} we have that $\type(\delta(w)) \in \{(1,1,0), (0,2,0)\}$ so for $f \in \delta(z(e)) \setminus \{e\}$ we have inductively by \ref{P3} and \eqref{ob: 1} in Observation~\ref{obs fractional matching consistent instances} that $y_f \leq \max\{c, 1-c\} = c$. Hence, as $y_e \leq 1 - c$ by~\ref{P3}, we have that $\sum_{f \in \delta(w(e))} y_f \leq c + 1 - c = 1$. Finally, if $\degree(w(e)) = 1$ then $\sum_{f \in \delta(w(e))} y_f = y_e < 1$. So \ref{P7} holds with respect to $x^{new}$ as required. \\

\textbf{Case~\ref{case: main 3.1}: $\type(e) = 3$ and $\{\type(\delta(u) \setminus \{e\}), \type(\delta(v) \setminus \{e\})\} = \{(1,0,0)\}$.}
Let $f_u, f_v, z(e)$, and $w(e)$ be defined as in lines~\ref{line: b1 f_u, f_v} and \ref{line: b1 z, w} in Algorithm~\ref{alg: one}. 
We first compute the updated values of $x^{new}_{z(e)}$ and $x^{new}_{w(e)}$.
\begin{equation} \label{case 3.1 x_z new} 
    \begin{split}
        x^{new}_{z(e)} = x^{old}_{z(e)} + \widetilde{y}_{n_{f_{z(e)}} + 1} - \min\left\{\frac{c}{2}, 1 - y_{f_{z(e)}} - y_{e}\right\} = c - \min\left\{\frac{c}{2}, 1 - y_{f_{z(e)}} - y_{e}\right\} \,,
    \end{split}
\end{equation}
where the first equality is due to line~\ref{line: b1 z cover change} and the second equality holds inductively by \ref{P5} as $\type(\delta(z(e)) \setminus \{e\}) = (1,0,0)$ we have $x^{old}_{z(e)} = c - \widetilde{y}_{f_{z(e)} + 1}$.
Similarly,
\begin{equation} \label{case 3.1 x_w new} 
    \begin{split}
        x^{new}_{w(e)} &= x^{old}_{w(e)} + \widetilde{y}_{n_{f_{w(e)}} + 1} - \max\left\{\frac{c}{2}, c - ( 1 - y_{f_{z(e)}} - y_{e})\right\} \\ 
        &= c - \max\left\{\frac{c}{2}, c - ( 1 - y_{f_{z(e)}} - y_{e})\right\} \,,
    \end{split}
\end{equation}
where the first equality is due to line~\ref{line: b1 w cover change} and the second equality holds inductively by \ref{P5} as $\type(\delta(w(e)) \setminus \{e\}) = (1,0,0)$ we have $x^{old}_{w(e)} = c - \widetilde{y}_{f_{w(e)} + 1}$.

\textbf{Checking~\ref{P1}.} As $\min\{\frac{c}{2}, 1 - y_{f_{z(e)}} - y_{e}\} + \max\{\frac{c}{2}, c - ( 1 - y_{f_{z(e)}} - y_{e})\} = c$ and the assignments the new values $x^{new}_{z(e)}, x^{new}_{w(e)}$, and $y_e$ happen in lines~\ref{line: b1 z cover change}, \ref{line: b1 w cover change}, and \ref{line: b1 y assigned}, it is straightforward to check that \ref{P1} holds for $G$.

\textbf{Checking~\ref{P2}.} By \eqref{case 3.1 x_z new} and \eqref{case 3.1 x_w new} we have the following,
\[x^{new}_{z(e)} + x^{new}_{w(e)} = 2c - \left(\min\left\{\frac{c}{2}, 1 - y_{f_{z(e)}} - y_{e}\right\} + \max\left\{\frac{c}{2}, c - ( 1 - y_{f_{z(e)}} - y_{e})\right\}\right) = c \,,\]
where the second equality holds as $\min\left\{\frac{c}{2}, 1 - y_{f_{z(e)}} - y_{e}\right\} + \max\left\{\frac{c}{2}, c - ( 1 - y_{f_{z(e)}} - y_{e})\right\} = c$.
We will now show that $x^{new}_{z(e)} \geq x^{old}_{z(e)}$ and $x^{new}_{w(e)} \geq x^{old}_{w(e)}$ which inductively implies that \ref{P2} holds with respect to $x^{new}$.
First, by~\eqref{ob: 1} in Observation~\ref{obs fractional matching consistent instances} we have that $\widetilde{y}_{f_{z(e)} + 1} \geq \frac{c}{2}$, therefore,
\[x^{new}_{z(e)} = x^{old}_{z(e)} + \widetilde{y}_{n_{f_{z(e)}} + 1} - \min\left\{\frac{c}{2}, 1 - y_{f_{z(e)}} - y_{e}\right\} \geq x^{old}_{z(e)} \,.\]
Similarly, if $\max\left\{\frac{c}{2}, c - ( 1 - y_{f_{z(e)}} - y_{e})\right\} = \frac{c}{2}$ then,
\[x^{new}_{w(e)} = x^{old}_{w(e)} + \widetilde{y}_{n_{f_{w(e)}} + 1} - \frac{c}{2} \geq x^{old}_{w(e)} \,.\]
Now, if $\max\left\{\frac{c}{2}, c - ( 1 - y_{f_{z(e)}} - y_{e})\right\} = c - ( 1 - y_{f_{z(e)}} - y_{e})$ then we have,
\begin{align*}
    x^{new}_{w(e)} &= x^{old}_{w(e)} + \widetilde{y}_{n_{f_{w(e)}} + 1} - \left(c - ( 1 - y_{f_{z(e)}} - y_{e})\right) = x^{old}_{w(e)} + 1 - y_{f_{z(e)}} - \widetilde{y}_{n_{f_{z(e)}} + 1} \\ &\geq x^{old}_{w(e)} + 1 - \widetilde{y}_{f_{z(e)}} - \widetilde{y}_{n_{f_{z(e)}} + 1} \geq x^{old}_{w(e)} \,,
\end{align*}
where the second equality is due to line~\ref{line: b1 y assigned}, the first inequality is due to \eqref{ob: 4} in Observation~\ref{obs edge types}, and the final inequality is due to~\eqref{lma 1.7} in Lemma~\ref{lma fractional matching consistent instances} and  \eqref{ob: 2} in Observation~\ref{obs fractional matching consistent instances}.
So \ref{P2} holds with respect to $x^{new}$ as required.

\textbf{Checking~\ref{P3}.} We first show $1 - y_{f_z} - y_e \geq 2 - 3c$ as follows,
\begin{align*}
    1 - y_{f_z} - y_e \geq 1 - \widetilde{y}_{n_{f_{z(e)}}} - \widetilde{y}_{n_{f_{z(e)}} + 1} + (c - \widetilde{y}_{n_{f_{w(e)}} + 1}) \geq 2(1 - \frac{3c}{2}) = 2 - 3c \,,
\end{align*}
where the fist inequality holds by line~\ref{line: b1 y assigned} and \eqref{ob: 4} in Observation~\ref{obs edge types} and the second inequality holds by~\eqref{lma 1.7} in Lemma~\ref{lma fractional matching consistent instances} and  \eqref{ob: 2} in Observation~\ref{obs fractional matching consistent instances}. So, by \eqref{case 3.1 x_z new} we have,
\[x_z = c - \min\left\{\frac{c}{2}, 1 - y_{f_z} - y_{e}\right\} \in [\frac{c}{2}, 4c - 2] \subset [2c - 1, \frac{5c - 2}{2}] \,.\]
Furthermore, 
\[x_z = c - \min\left\{\frac{c}{2}, 1 - y_{f_z} - y_{e}\right\} \geq c - 1 + y_{f_z} + y_{e} \,.\]
So \ref{P3} holds with respect to $z(e)$. We now check $w(e)$, first we have,
\[x_w = c - \max\left\{\frac{c}{2}, c - ( 1 - y_{f_z} - y_{e})\right\} \in [2-3c, \frac{c}{2}] \subset [2c - 1, \frac{5c - 2}{2}] \,,\]
as $1 - y_{f_z} - y_e \geq 2 - 3c$ from above. Now, to show $x^{new}_{w(e)} \geq c - 1 + \sum_{f \in \delta(w(e))} y_f$ we consider the following case study. If $\max\left\{\frac{c}{2}, c - ( 1 - y_{f_{z(e)}} - y_{e})\right\} = \frac{c}{2}$ then we have the following,
\begin{align*}
    c - x^{new}_{w(e)} + y_e + y_{f_{w(e)}} \leq 2 \widetilde{y}_{n_{f_{w(e)}} + 1} + \widetilde{y}_{n_{f_{w(e)}}} - \frac{c}{2} \leq 1
\end{align*}
where the first inequality holds as $x^{new}_{w(e)} = \frac{c}{2}$ by \eqref{case 3.1 x_w new} and \eqref{ob: 4} in Observation~\ref{obs edge types} and the second inequality holds by~\eqref{lma 1.6} in Lemma~\ref{lma fractional matching consistent instances}.
We now consider the case where $\max\left\{\frac{c}{2}, c - ( 1 - y_{f_{z(e)}} - y_{e})\right\} = c - ( 1 - y_{f_{z(e)}} - y_{e})$. First, we have the following,
\begin{align*}
    c - ( 1 - y_{f_{z(e)}} - y_{e}) \leq \widetilde{y}_{n_{f_{z(e)}}} + \widetilde{y}_{n_{f_{z(e)}} + 1} + \widetilde{y}_{n_{f_{w(e)}} + 1} - 1 = \widetilde{y}_{{n_{f_z}} + 2} + \widetilde{y}_{n_{f_w} + 1} - c \,,
\end{align*} 
where the inequality is due to line~\ref{line: b1 y assigned} and \eqref{ob: 4} in Observation~\ref{obs edge types} and the equality is due to~\eqref{lma 1.7} in Lemma~\ref{lma fractional matching consistent instances}. So we have,
\begin{align*}
    c - x_w + y_e + y_{f_w} &\leq 2 \widetilde{y}_{n_{f_w} + 1} + \widetilde{y}_{n_{f_w}} + \widetilde{y}_{{n_{f_z}} + 2} + \widetilde{y}_{{n_{f_z}} + 1} - 2c \\
    &\leq 1 + \frac{c}{2} + \frac{3c}{2} - 2c = 1 \,,
\end{align*}
where the first inequality holds by line~\ref{line: b1 y assigned} and as $c - x_{w(e)} = c - ( 1 - y_{f_{z(e)}} - y_{e}) \leq \widetilde{y}_{{n_{f_z}} + 2} + \widetilde{y}_{n_{f_w} + 1} - c$ by \eqref{case 3.1 x_w new} and the second inequality holds by~\eqref{lma 1.6} in Lemma~\ref{lma fractional matching consistent instances} and because $\widetilde{y}_{{n_{f_z}} + 2} + \widetilde{y}_{{n_{f_z}} + 1} \leq \frac{3c}{2}$ which follows from~\eqref{lma 1.7} in Lemma~\ref{lma fractional matching consistent instances} as well as \eqref{ob: 2} in Observation~\ref{obs fractional matching consistent instances}. So \ref{P3} holds with respect to $w(e)$ and therefore with respect to $G$ as required.

\textbf{Checking~\ref{P4}, \ref{P5}, and \ref{P6}.} Following the proof of \ref{P2} above we have shown that both $x^{new}_{z(e)} \geq x^{old}_{z(e)}$ and $x^{new}_{w(e)} \geq x^{old}_{w(e)}$ and so as all other values of $x$ remain unchanged we have that properties \ref{P4}, \ref{P5}, and \ref{P6} hold inductively. 

\textbf{Checking~\ref{P7}.} Following the proof of \ref{P3} above we have that $x^{new}_{z(e)} \geq c - 1 + \sum_{f \in \delta(z(e))} y_f$, $x^{new}_{w(e)} \geq c - 1 + \sum_{f \in \delta(w(e))} y_f$, and $x^{new}_{z(e)}, x^{new}_{w(e)} \leq \frac{5c - 2}{2} < c$ hence, $\sum_{f \in \delta(z(e))} y_f < 1$ and $\sum_{f \in \delta(w(e))} y_f < 1$ so \ref{P7} holds. \\

\textbf{Case~\ref{case: main 3.2}: $\type(e) = 3$ and $\{\type(\delta(u) \setminus \{e\}), \type(\delta(v) \setminus \{e\})\} = \{(1,0,0), (0,1,0)\}$.}
Let us define $f_1, f_2, z(e),$ and $w(e)$ as in line~\ref{line: b2 definitions}. 
We first compute the updated values of $x^{new}_{z(e)}$ and $x^{new}_{w(e)}$.
\begin{equation} \label{case 3.2 x_z new} 
    \begin{split} 
        x^{new}_{z(e)} &= x^{old}_{z(e)} + y_e - \max\left\{(2c-1) - y_{f_2}, 0\right\} \\
        &= c - \widetilde{y}_{n_{f_1} + 1} + \max\left\{\widetilde{y}_{n_{f_1} + 1} - y_{f_2}, 0\right\} - \max\left\{(2c-1) - y_{f_2}, 0\right\} \\
        &= \begin{cases}
            c - \widetilde{y}_{n_{f_1} + 1} & \mbox{ if $ y_{f_2} \geq \widetilde{y}_{n_{f_1} + 1}$}\\
            c - y_{f_2} & \mbox{ if $ y_{f_2} \in [2c-1,  \widetilde{y}_{n_{f_1} + 1})$} \\
            1 - c & \mbox{ if $ y_{f_2} < 2c-1 $} \\
        \end{cases}
    \end{split}   
\end{equation}
where the first equality is due to line~\ref{line: b2 z cover change} and the second equality is due to line~\ref{line: b2 y assigned} and \ref{P5} inductively as $\type(\delta(z(e)) \setminus \{e\}) = (1,0,0)$.
Similarly,
\begin{equation} \label{case 3.2 x_w new} 
    \begin{split}
        x^{new}_{w(e)} &= x^{old}_{w(e)} + \max\left\{(2c-1) - y_{f_2}, 0\right\} \\
        &= y_{f_2} + \max\left\{(2c-1) - y_{f_2}, 0\right\} \\
        &= \begin{cases}
            y_{f_2} & \mbox{ if $ y_{f_2} > 2c - 1$}\\
            2c-1 & \mbox{ if $ y_{f_2} \leq 2c - 1$} \\
        \end{cases}
    \end{split}
\end{equation}
where the first equality is due to line~\ref{line: b2 w cover change} and the second equality follows from \eqref{ob: 8} in Observation~\ref{obs vertex neighborhood}.

\textbf{Checking~\ref{P1}.} In this case, the assignment of new values $x^{new}_{z(e)}$, $x^{new}_{w(e)}$ and $y_e$ happen in lines \ref{line: b2 y assigned}, \ref{line: b2 z cover change}, and \ref{line: b2 w cover change}, so it is straightforward to check that the property~\ref{P1} holds for $G$. 

\textbf{Checking~\ref{P2}.} By \eqref{case 3.2 x_z new} and \eqref{case 3.2 x_w new} we have,
\begin{align*}
    x^{new}_{z(e)} + x^{new}_{w(e)} = c - \widetilde{y}_{n_{f_1} + 1} + \max\left\{\widetilde{y}_{n_{f_1} + 1} - y_{f_2}, 0\right\} + y_{f_2} \geq c \,,
\end{align*}
where the inequality follows directly from the previous line through a case analysis of $\max\left\{\widetilde{y}_{n_{f_1} + 1} - y_{f_2}, 0\right\}$. We proceed by showing that $x^{new}_{z(e)} \geq x^{old}_{z(e)}$ and $x^{new}_{w(e)} \geq x^{old}_{w(e)}$ which will imply inductively that \ref{P2} hols on $G$. First, let us consider $x^{new}_{z(e)}$. By \eqref{case 3.2 x_z new} it suffices to show that $y_e - \max\left\{(2c-1) - y_{f_2}, 0\right\} \geq 0$. Consider the following,
\begin{align*}
    y_e - \max\left\{(2c-1) - y_{f_2}, 0\right\} &= \max\left\{\widetilde{y}_{n_{f_1} + 1} - y_{f_2}, 0\right\} - \max\left\{(2c-1) - y_{f_2}, 0\right\} \\
    &\geq \begin{cases}
        \widetilde{y}_{n_{f_1} + 1} - (2c - 1) & \mbox{ if $ y_e = \widetilde{y}_{n_{f_1} + 1} - y_{f_2}$}\\
        0 & \mbox{ if $ y_e = 0 $} \\
    \end{cases} \,,
\end{align*}
where the equality is due to line~\ref{line: b2 y assigned} and the inequality holds as if $y_e = 0$ then $y_{f_2} \geq \widetilde{y}_{n_{f_1} + 1} \geq \frac{c}{2} > 2c - 1$ by~\eqref{ob: 1} in Observation~\ref{obs fractional matching consistent instances} and so $\max\left\{(2c-1) - y_{f_2}, 0\right\} = 0$. So $x^{new}_{z(e)} \geq x^{old}_{z(e)}$; moreover, by \eqref{case 3.2 x_w new} we have that $x^{new}_{w(e)} \geq x^{old}_{w(e)}$ hence \ref{P2} holds as required.

\textbf{Checking~\ref{P3}.} We first show $x^{new}_{z(e)} \in [2c-1, \frac{5c-2}{2}]$. By \ref{case 3.2 x_z new} we have,
\begin{equation*}
    x^{new}_{z(e)} = 
    \begin{cases}
        c - \widetilde{y}_{n_{f_1} + 1} & \mbox{ if $y_{f_2} \geq \widetilde{y}_{n_{f_1} + 1}$ }\\
        c - y_{f_2} & \mbox{ if $ y_{f_2} \in [2c-1,  \widetilde{y}_{n_{f_1} + 1})$} \\
        1 - c & \mbox{ if $ y_{f_2} < 2c-1 $} \\
    \end{cases} \ \ 
    \in [2c-1, \frac{5c-2}{2}] \,,
\end{equation*}
where the inclusion holds as $2c-1 < 1-c < \frac{5c-2}{2}$, $\widetilde{y}_{n_{f_1} + 1} \in [\frac{c}{2}, \frac{5c - 2}{2}]$ by~\eqref{ob: 1} in Observation~\ref{obs fractional matching consistent instances}, and $y_{f_2} \in [1 - \frac{3c}{2}, 1-c]$ inductively by \ref{P3}, so if $y_{f_2} \geq \widetilde{y}_{n_{f_1} + 1}$ then $\widetilde{y}_{n_{f_1} + 1} \leq 1 - c$ hence, $c - \widetilde{y}_{n_{f_1} + 1} \in [2c-1 \frac{c}{2}]$. We will now show that $x^{new}_{z(e)} \geq c - 1 + \sum_{f \in \delta(z(e))} y_f$ through a case analysis on the value of $x^{new}_{z(e)}$ as in \eqref{case 3.2 x_z new}. If $x^{new}_{z(e)} =  c - \widetilde{y}_{n_{f_1} + 1}$ then we have that $y_{f_2} \geq \widetilde{y}_{n_{f_1} + 1}$ and hence by line~\ref{line: b2 y assigned} we have that $y_e = 0$ so,
\[c - x^{new}_{z(e)} + y_{f_1} + y_e = \widetilde{y}_{n_{f_1} + 1} + y_{f_1} \leq \widetilde{y}_{n_{f_1} + 1} + \widetilde{y}_{n_{f_1}} \leq 1 \,, \]
where the first inequality follows from \eqref{ob: 4} in Observation~\ref{obs edge types} and the second inequality is due to \eqref{ob: 3c/2 bound} in Observation~\ref{obs fractional matching consistent instances}.
If $x^{new}_{z(e)} =  c - y_{f_2}$ then we have that $y_{f_2} \in [2c-1,  \widetilde{y}_{n_{f_1} + 1})$ and hence by line~\ref{line: b2 y assigned} $y_e = \widetilde{y}_{n_{f_1} + 1} - y_{f_2}$. So,
\[c - x^{new}_{z(e)} + y_{f_1} + y_e = \widetilde{y}_{n_{f_1} + 1} + y_{f_1} \leq \widetilde{y}_{n_{f_1} + 1} + \widetilde{y}_{n_{f_1}} \leq 1 \,, \]
where the inequalities hold for the same reason as in the case where $x^{new}_{z(e)} =  c - \widetilde{y}_{n_{f_1} + 1}$.
Finally, if $x^{new}_{z(e)} = 1 - c$ then we have that $y_{f_2} < 2c-1 < \widetilde{y}_{n_{f_1} + 1}$ and hence by line~\ref{line: b2 y assigned} we have $y_e = \widetilde{y}_{n_{f_1} + 1} - y_{f_2}$ so,
\[c - x^{new}_{z(e)} + y_{f_1} + y_e \leq \widetilde{y}_{n_{f_1} + 1} + \widetilde{y}_{n_{f_1}} + 2c - 1 - y_{f_2} \leq 5c - 2 < 1 \,, \]
where the first inequality is due to \eqref{ob: 4} in Observation~\ref{obs edge types} and the second inequality holds as $\widetilde{y}_{n_{f_1} + 1} + \widetilde{y}_{n_{f_1}} \leq \frac{3c}{2}$ by \eqref{ob: 3c/2 bound} in Observation~\ref{obs fractional matching consistent instances} and $y_{f_2} \geq 1 - \frac{3c}{2}$ inductively by \ref{P3}. So \ref{P3} holds with respect to $z(e)$. We now check $w(e)$, by \eqref{case 3.2 x_w new} we have that,
\begin{equation*}
    x^{new}_{w(e)} = 
    \begin{cases}
        y_{f_2} & \mbox{ if $ y_{f_2} > 2c - 1$} \\
        2c-1 & \mbox{ if $ y_{f_2} \leq 2c - 1$} \\
    \end{cases} \ \ 
    \in [2c-1, \frac{5c-2}{2}] \,,
\end{equation*}
where the inclusion holds as $y_{f_2} \leq 1 - c$ inductively by \ref{P3}. We will now show that $x^{new}_{w(e)} \geq c - 1 + \sum_{f \in \delta(w(e))} y_f$ through a case analysis on the value of $x^{new}_{w(e)}$ as in \eqref{case 3.2 x_z new}. If $x^{new}_{w(e)} = y_{f_2}$ we have that $y_{f_2} > 2c - 1$ and hence,
\[c - x^{new}_{w(e)} + y_{f_2} + y_e = c - y_e < 1 \,.\]
Moreover, if $x^{new}_{w(e)} = 2c-1$ then $y_{f_2} < 2c-1 < \widetilde{y}_{n_{f_1} + 1}$ and hence by line~\ref{line: b2 y assigned} we have $y_e = \widetilde{y}_{n_{f_1} + 1} - y_{f_2}$ so,
\[c - x^{new}_{w(e)} + y_{f_2} + y_e = 1 - c + \widetilde{y}_{n_{f_1} + 1} \leq 1 \,,\]
where the inequality holds by \eqref{ob: 1} in Observation~\ref{obs fractional matching consistent instances}. So \ref{P3} holds with respect to $w(e)$ and hence on $G$ inductively. 

\textbf{Checking~\ref{P4}, \ref{P5}, and \ref{P6}.} Following the proof of \ref{P2} above we have shown that both $x^{new}_{z(e)} \geq x^{old}_{z(e)}$ and $x^{new}_{w(e)} \geq x^{old}_{w(e)}$ and so as all other values of $x$ remain unchanged we have that properties \ref{P4}, \ref{P5}, and \ref{P6} hold inductively. 

\textbf{Checking~\ref{P7}.} Following the proof of \ref{P3} we have that $x^{new}_{z(e)} \geq c - 1 + \sum_{f \in \delta(z(e))} y_f$, $x^{new}_{w(e)} \geq c - 1 + \sum_{f \in \delta(w(e))} y_f$, and $x^{new}_{z(e)}, x^{new}_{w(e)} \leq \frac{5c - 2}{2} < c$ hence, $\sum_{f \in \delta(z(e))} y_f < 1$ and $\sum_{f \in \delta(w(e))} y_f < 1$ so \ref{P7} holds. \\

\textbf{Case~\ref{case: main 3.3}: $\type(e) = 3$ and $\{\type(\delta(u) \setminus \{e\}), \type(\delta(v) \setminus \{e\})\} = \{(0,2,0), (1,0,0)\}$.}
So as $\degree(u) \geq \degree(v)$ we have that $\type(\delta(u) \setminus \{e\}) = (0,2,0)$ and $\type(\delta(v) \setminus \{e\}) = (1,0,0)$. Let $f_v \in \delta(v) \setminus \{e\}$ that is $\type(f_v) = 1$ and let $f_1, f_2 \in \delta(u) \setminus \{e\}$ that is $\type(f_1) = \type(f_2) = 2$, without loss of generality we may assume that $y_{f_1} \geq y_{f_2}$. 
So, by line~\ref{line: b3 v cover change} we have that,
\begin{equation} \label{case 3.3 x_u new} 
    x^{new}_u = x^{old}_u = y_{f_1} + y_{f_2}
\end{equation}
where the second equality is due to \eqref{ob: 8} in Observation~\ref{obs vertex neighborhood} and,
\begin{equation} \label{case 3.3 x_v new} 
    \begin{split}
        x^{new}_{v} &= x^{old}_{v} + y_e \\
        &= c - \widetilde{y}_{n_{f_v} + 1} + \max\left\{\widetilde{y}_{n_{f_v} + 1} - y_{f_1}, 0\right\} \\
        &= \begin{cases}
            c - \widetilde{y}_{n_{f_v} + 1} & \mbox{ if $ y_{f_1} > \widetilde{y}_{n_{f_v} + 1}$}\\
            c - y_{f_1} & \mbox{ if $ y_{f_2} \leq y_{f_1} \leq \widetilde{y}_{n_{f_v} + 1}$} \\
        \end{cases}
    \end{split}
\end{equation}
where the second equality holds by line~\ref{line: b3 y assigned} and inductively by \ref{P5} as $\type(\delta(v) \setminus\{e\}) = (1,0,0)$.

\textbf{Checking~\ref{P1}.} In this case, the assignment of new values $x^{new}_{u}$, $x^{new}_{v}$ and $y_e$ happens in lines \ref{line: b3 y assigned} and \ref{line: b3 v cover change}, so it is straightforward to check that the property~\ref{P1} holds for $G$. 

\textbf{Checking~\ref{P2}.} By \eqref{case 3.3 x_u new} and \eqref{case 3.3 x_v new} we have,
\[x^{new}_{u} + x^{new}_{v} = c - \widetilde{y}_{n_{f_v} + 1} + \max\left\{\widetilde{y}_{n_{f_v} + 1} - y_{f_1}, 0\right\} + y_{f_1} + y_{f_2} \geq c\]
where the inequality holds by a case analysis of $\max\left\{\widetilde{y}_{n_{f_v} + 1} - y_{f_1}, 0\right\}$. Moreover, by \eqref{case 3.3 x_u new} and \eqref{case 3.3 x_v new} we have, that $x^{new}_u = x^{old}_u$ and $x^{new}_v = x^{old}_v + y_e = x^{old}_v + \max\left\{\widetilde{y}_{n_{f_v} + 1} - y_{f_1}, 0\right\} \geq x^{old}_v$, therefore we have inductively that \ref{P2} holds for $G$ as required.

\textbf{Checking~\ref{P3}.} First, as $\degree(u) = 3$ we have that $u$ does not satisfy the premise of~\ref{P3}. Now, by \eqref{case 3.3 x_v new} we have that,
\begin{equation*}
    x^{new}_{v} 
    = \begin{cases}
        c - \widetilde{y}_{n_{f_v} + 1} & \mbox{ if $ y_{f_1} > \widetilde{y}_{n_{f_v} + 1}$}\\
        c - y_{f_1} & \mbox{ if $ y_{f_2} \leq y_{f_1} \leq \widetilde{y}_{n_{f_v} + 1}$} \\
    \end{cases}
\end{equation*}
and hence we proceed by case analysis on the value of $x^{new}_v$. 
If $x^{new}_v = c - \widetilde{y}_{n_{f_v} + 1}$ then $y_{f_1} > \widetilde{y}_{n_{f_v} + 1}$ and hence by line~\ref{line: b3 y assigned} we have that $y_e = 0$. Moreover, inductively by \ref{P3} we have $y_{f_1} \leq 1 - c$ and so $x^{new}_v = c - \widetilde{y}_{n_{f_v} + 1} \geq 2c - 1$ and hence by  \eqref{ob: 2} in Observation~\ref{obs fractional matching consistent instances} we have that $x^{new}_v \in [2c-1, \frac{c}{2}] \subset [2c-1, \frac{5c-2}{2}]$. Also,
\[c - x^{new}_v + y_e + y_{f_v} = \widetilde{y}_{n_{f_v} + 1} + y_{f_v} \leq \widetilde{y}_{n_{f_v} + 1} + \widetilde{y}_{n_{f_v}} \leq 1 \,\]
where the first inequality holds by \eqref{ob: 4} in Observation~\ref{obs edge types} and the second holds by \eqref{ob: 3c/2 bound} in Observation~\ref{obs fractional matching consistent instances}.
If $x^{new}_v = c - y_{f_1}$ then $y_{f_1} \leq \widetilde{y}_{n_{f_v} + 1}$ and hence by line~\ref{line: b3 y assigned} we have that $y_e = \widetilde{y}_{n_{f_v} + 1} - y_{f_1}$. Moreover, inductively by \ref{P3} we have that $y_{f_1} \in [1 - \frac{3c}{2}, 1 - c]$ and so, $x^{new}_v = c - y_{f_1} \in [2c - 1, \frac{5c - 2}{2}]$. Also,
\[c - x^{new}_v + y_e + y_{f_v} = \widetilde{y}_{n_{f_v} + 1} + y_{f_v} \leq \widetilde{y}_{n_{f_v} + 1} + \widetilde{y}_{n_{f_v}} \leq 1 \,\]
where the first inequality holds by \eqref{ob: 4} in Observation~\ref{obs edge types} and the second holds by \eqref{ob: 3c/2 bound} in Observation~\ref{obs fractional matching consistent instances}.
So \ref{P3} holds as required.

\textbf{Checking~\ref{P4}, \ref{P5}, and \ref{P6}.} Following the proof of \ref{P2} above we have shown that both $x^{new}_{u} \geq x^{old}_{u}$ and $x^{new}_{v} \geq x^{old}_{v}$ and so as all other values of $x$ remain unchanged we have that properties \ref{P4}, \ref{P5}, and \ref{P6} hold inductively. 

\textbf{Checking~\ref{P7}.} Following the proof of \ref{P3} we have that $x^{new}_v \geq c - 1 + \sum_{f \in \delta(v)} y_f$ and $x^{new}_v \leq \frac{5c-2}{2} < c$ hence, $\sum_{f \in \delta(v)} y_f \leq 1$ as required. Moreover,
\begin{align*}
    \sum_{f \in \delta(u)} y_f = y_{f_1} + y_{f_2} + y_e 
    = \begin{cases}
        y_{f_1} + y_{f_2} &\mbox{ if $ y_e = 0 $} \\
        \widetilde{y}_{n_{f_v} + 1} + y_{f_2} &\mbox{ if $ y_e = \widetilde{y}_{n_{f_v} + 1} - y_{f_1} $} 
    \end{cases}
    \leq 1 \,,
\end{align*}
where the second equality follows from line~\ref{line: b3 y assigned} and the inequality holds as $y_{f_1}, y_{f_2} \leq 1 - c$ inductively by \ref{P3} and $ \widetilde{y}_{n_{f_v} + 1} \leq \frac{5c-2}{2}$ by \eqref{ob: 1} in Observation~\ref{obs fractional matching consistent instances}. So \ref{P7} holds as required. \\

\textbf{Case~\ref{case: main 3.4}: $\type(e) = 3$ and $\{\type(\delta(u) \setminus \{e\}), \type(\delta(v) \setminus \{e\})\} = \{(1,1,0), (1,0,0)\}$.}
First, as $\degree(u) \geq \degree(v)$ we have that $\type(\delta(u) \setminus \{e\}) = (1,1,0)$ and $\type(\delta(v) \setminus \{e\}) = (1,0,0)$. Let $f_1, f_2 \in \delta(u) \setminus \{e\}$ with $\type(f_i) = i$ for $i = 1, 2$, and let $f_v \in \delta(v) \setminus \{e\}$ hence $\type(f_v) = 1$. By line~\ref{line: b4 u cover change} we have,
\begin{equation} \label{case 3.4 x_u new} 
    \begin{split}
        x^{new}_{u} &= x^{old}_{u} + y_e - \max\left\{\widetilde{y}_{n_{f_v} + 1} - y_{f_2}, 0\right\} \\
        &= c - \widetilde{y}_{n_{f_1} + 1} + y_{f_2} + \max\{\widetilde{y}_{n_{f_v} + 1} - y_{f_2} - \min\{c - \widetilde{y}_{n_{f_1} + 1}, y_{f_2}\}, 0\} \\
        &\qquad - \max\{\widetilde{y}_{n_{f_v} + 1} - y_{f_2}, 0\} \\
        &= \begin{cases}
            c - \widetilde{y}_{n_{f_1} + 1} + y_{f_2} - \min\{c - \widetilde{y}_{n_{f_1} + 1}, y_{f_2}\} & \mbox{ if $ y_e \neq 0 $} \\
            c - \widetilde{y}_{n_{f_1} + 1} + y_{f_2} & \mbox{ if $ y_e = 0 > \widetilde{y}_{n_{f_v} + 1} - y_{f_2}$} \\
            c - \widetilde{y}_{n_{f_1} + 1} + y_{f_2} - (\widetilde{y}_{n_{f_v} + 1} - y_{f_2}) & \mbox{ if $ y_e = 0 \leq \widetilde{y}_{n_{f_v} + 1} - y_{f_2}$} \\
        \end{cases} \\
        &\geq c - \widetilde{y}_{n_{f_1} + 1} + y_{f_2} - \min\{c - \widetilde{y}_{n_{f_1} + 1}, y_{f_2}\} \\
        &= \max\{c - \widetilde{y}_{n_{f_1} + 1}, y_{f_2}\}
    \end{split}
\end{equation}
where the second equality holds by line~\ref{line: b4 y assigned} and as $x^{old}_{u} = c - \widetilde{y}_{n_{f_1} + 1} + y_{f_2}$ by Claim~\ref{claim 1} and the third equality holds from a case analysis of the two max functions. The inequality holds as if $y_e = 0$ then $\widetilde{y}_{n_{f_v} + 1} - y_{f_2} \leq \min\{c - \widetilde{y}_{n_{f_1} + 1}, y_{f_2}\}$ and the final equality is due to a case analysis of $\min\{c - \widetilde{y}_{n_{f_1} + 1}, y_{f_2}\}$. 
Moreover, by line~\ref{line: b4 v cover change} we have,
\begin{equation} \label{case 3.4 x_v new} 
    \begin{split}
        x^{new}_{v} &= x^{old}_{v} + \max\left\{\widetilde{y}_{n_{f_v} + 1} - y_{f_2}, 0\right\} = \max\left\{c - \widetilde{y}_{n_{f_v} + 1}, c - y_{f_2}\right\}
    \end{split}
\end{equation}
where the second equality holds as $x^{old}_{v} = c - \widetilde{y}_{n_{f_v} + 1}$ inductively by \ref{P5} and by a case analysis of $\max\left\{\widetilde{y}_{n_{f_v} + 1} - y_{f_2}, 0\right\}$.

\textbf{Checking~\ref{P1}.} In this case, the assignment of new values $x^{new}_{u}$, $x^{new}_{v}$ and $y_e$ happens in lines \ref{line: b4 y assigned}, \ref{line: b4 u cover change} and \ref{line: b4 v cover change}, so it is straightforward to check that the property~\ref{P1} holds for $G$. 

\textbf{Checking~\ref{P2}.} Following \eqref{case 3.4 x_u new} and \eqref{case 3.4 x_v new} we have that,
\[x^{new}_{u} + x^{new}_{v} \geq \max\{c - \widetilde{y}_{n_{f_1} + 1}, y_{f_2}\} + \max\left\{c - \widetilde{y}_{n_{f_v} + 1}, c - y_{f_2}\right\} \geq y_{f_2} + c - y_{f_2} = c\]
Furthermore, by \eqref{case 3.4 x_u new} we have that $x^{new}_u \geq \max\{c - \widetilde{y}_{n_{f_1} + 1}, y_{f_2}\}$ so inductively by \ref{P5} we have that $\sum_{a \in \Eends(f_1)} x^{new}_{a} \geq \widetilde{y}_{n_{f_1} + 1} + c - \widetilde{y}_{n_{f_1} + 1} = c$ also, as $x^{new}_u \geq y_{f_2}$ it follows by \eqref{ob: 7} in Observation~\ref{obs edge types} that $\sum_{a \in \Eends(f_2)} x^{new}_{a} \geq c$.
Furthermore, by \eqref{case 3.4 x_v new} we have that $x^{new}_v \geq x^{old}_v$ and hence we have inductively that \ref{P2} holds.

\textbf{Checking~\ref{P3}.} First we have that as $\degree(u) = 3$ that $u$ fails the premise of \ref{P3}. Now, following \eqref{case 3.4 x_v new} we have that,
\begin{align*}
    x^{new}_{v} = \max\left\{c - \widetilde{y}_{n_{f_v} + 1}, c - y_{f_2}\right\} \in [2c-1, \frac{5c-2}{2}]
\end{align*}
where the inclusion holds as if $x^{new}_{v} = c - \widetilde{y}_{n_{f_v} + 1}$ then $\widetilde{y}_{n_{f_v} + 1} \leq y_{f_2} \leq 1 - c$ inductively by \ref{P3} so along with  \eqref{ob: 2} in Observation~\ref{obs fractional matching consistent instances} we have that $x^{new}_{v} \in [2c-1, \frac{c}{2}] \subset [2c-1, \frac{5c-2}{2}]$. Moreover, if $x^{new}_{v} = c - y_{f_2}$ then inductively by \ref{P3} we have that $x^{new}_{v} \in [2c-1, \frac{5c-2}{2}]$.
To show $x^{new}_{v} \geq c - 1 + \sum_{f \in \delta(v)} y_f$ we consider the following case study based on the value of $x^{new}_v$ as in \eqref{case 3.4 x_v new}. If $x^{new}_{v} = c - \widetilde{y}_{n_{f_v} + 1}$ then we have that $\widetilde{y}_{n_{f_v} + 1} \leq y_{f_2}$ hence $y_e = 0$ therefore,
\[c - x^{new}_v + y_{f_v} + y_e \leq \widetilde{y}_{n_{f_v} + 1} + \widetilde{y}_{n_{f_v}} < 1 \,,\]
where the first inequality follows from \eqref{ob: 4} in Observation~\ref{obs edge types} and the second inequality follows from \eqref{ob: 3c/2 bound} in Observation~\ref{obs fractional matching consistent instances}.
If $x^{new}_{v} = c - y_{f_2}$ then,
\begin{align*}
    c - x^{new}_v + y_{f_v} + y_e &= y_{f_2} + y_{f_v} + \max\{\widetilde{y}_{n_{f_v} + 1} - y_{f_2} - \min\{c - \widetilde{y}_{n_{f_1} + 1}, y_{f_2}\}, 0\} \\
    &= \begin{cases}
        y_{f_2} + y_{f_v} & \mbox{ if $ y_e = 0 $} \\
        \widetilde{y}_{n_{f_v} + 1} + y_{f_v} - \min\{c - \widetilde{y}_{n_{f_1} + 1}, y_{f_2}\} & \mbox{ if $ y_e \neq 0 $}
    \end{cases} \\
    &\leq 1
\end{align*}
where the first equality holds by the assignment of $y_e$ in line~\ref{line: b4 y assigned} and the inequality holds as $y_{f_2} \leq 1 - c$ inductively by \ref{P3} and $y_{f_v} \leq c$ by \eqref{ob: 1} in Observation~\ref{obs fractional matching consistent instances} so $y_{f_2} + y_{f_v} \leq 1 - c + c = 1$ and by \eqref{ob: 4} in Observation~\ref{obs edge types} and \eqref{ob: 3c/2 bound} in Observation~\ref{obs fractional matching consistent instances} we have that $\widetilde{y}_{n_{f_v} + 1} + y_{f_v} \leq \widetilde{y}_{n_{f_v} + 1} + \widetilde{y}_{n_{f_v}} < 1$.
So \ref{P3} holds on $G$ as required.

\textbf{Checking~\ref{P4}.} As $\type(\delta(u) \setminus \{e\}) = (1,1,0)$ we have that the premise of \ref{P4} does not hold with respect to $u$. Moreover, by \eqref{case 3.4 x_v new} we have that $x^{new}_{v} \geq x^{old}_v$ as required.

\textbf{Checking~\ref{P5} and \ref{P6}.} By \eqref{case 3.4 x_u new} and \eqref{case 3.4 x_v new} we have that $x^{new}_{u} \geq \max\left\{c - \widetilde{y}_{n_{f_1} + 1}, y_{f_2}\right\}$ and $x^{new}_v \geq x^{old}_v$ and so both \ref{P5} and \ref{P6} holds.

\textbf{Checking~\ref{P7}.} Following the proof of \ref{P3} we have that $x^{new}_v \geq c - 1 + \sum_{f \in \delta(v)} y_f$ and $x^{new}_v \leq \frac{5c-2}{2} < 1$ and hence, $\sum_{f \in \delta(v)} y_f \leq 1$ as required. 
To show $\sum_{f \in \delta(u)} y_f \leq 1$, we consider a case study on the value of $y_e$ as in line~\ref{line: b4 y assigned}. If $y_e = 0$ then,
\[\sum_{f \in \delta(u)} y_f = y_{f_1} + y_{f_2} \leq c + 1 - c = 1 \,,\]
where the inequality holds by \eqref{ob: 1} in Observation~\ref{obs fractional matching consistent instances} and inductively by \ref{P3}. If we have \[y_e = \widetilde{y}_{n_{f_v} + 1} - y_{f_2} - \min\left\{c - \widetilde{y}_{n_{f_1} + 1}, y_{f_2}\right\}\] then we have
\begin{align*}
    \sum_{f \in \delta(u)} y_f &= y_{f_1} + y_{f_2} + \widetilde{y}_{n_{f_v} + 1} - y_{f_2} - \min\left\{c - \widetilde{y}_{n_{f_1} + 1}, y_{f_2}\right\} \\
    &\leq \widetilde{y}_{n_{f_1}} + \widetilde{y}_{n_{f_v} + 1} - \min\left\{c - \widetilde{y}_{n_{f_1} + 1}, y_{f_2}\right\} \\
    &\leq c + \frac{5c-2}{2} - (1 - \frac{3c}{2})
    = 5c - 2 
    < 1
\end{align*}
where the first inequality holds by \eqref{ob: 4} in Observation~\ref{obs edge types}. The second inequality holds as $ \widetilde{y}_{n_{f_1}} \leq c$ and $\widetilde{y}_{n_{f_v} + 1} \leq \frac{5c-2}{2}$ by \eqref{ob: 1} in Observation~\ref{obs fractional matching consistent instances} and $\min\left\{c - \widetilde{y}_{n_{f_1} + 1}, y_{f_2}\right\} \geq 1 - \frac{3c}{2}$ by  \eqref{ob: 2} in Observation~\ref{obs fractional matching consistent instances} and inductively by \ref{P3}. So \ref{P7} holds as required.

So all the properties hold by induction.

\section{Upper Bound for MinIndex for Maximum Degree Three}
First, let us introduce a framework developed by Buchbinder, Segev, and Tkach \cite{Buchbinder}, so-called MinIndex algorithm. This framework produces an integral matching within the general adversarial edge arrivals model. MinIndex is parametrized by a natural number $k$ and $k$ nonnegative numbers $p_1$, \ldots, $p_k$ such that $p_1+\ldots+p_k=1$. The framework functions by maintaining a distribution of matchings where each matching in this distribution is returned with a pre-determined probability. Once an edge arrives, it is greedily added to the first matching for which it is feasible, see Algorithm~\ref{alg: MinIndex}. 

\begin{algorithm}[H]
    \caption{MinIndex($k, p_1, \dots, p_k$)} \label{alg: MinIndex}
    \begin{algorithmic}
        \State Initialize: $M_i \gets \emptyset$ for all $i = 1, \dots, k$
        \State When $e$ arrives:
        \If{$M_i \cup \{e\}$ is not feasible for all $i = 1, \dots, k$}
            \State Reject $e$.
        \Else
            \State $M_i \gets M_i \cup \{e\}$ where $i$ is the minimal index for which $M_i \cup \{e\}$ is feasible.
        \EndIf
        \State \Return $M_i$ with probability $p_i$.
    \end{algorithmic}
\end{algorithm}

\cite{Buchbinder} shows that MinIndex with $k=3$ and $(p_1,p_2,p_3)=(5/9,3/9,1/9)$ achieves the guarantee $5/9$ in the adversarial edge arrival model when the underlying graph is restricted to be a forest. Here, we show that $5/9$ is the best guarantee achievable by MinIndex even on forests with maximum degree three. Note, \cite{Buchbinder} demonstrated that MinIndex cannot achieve a guarantee larger than $5/9$ on forests of maximum degree four, so our results improves this bound both in terms of the guarantee and in terms of the permitted maximum degree.

\begin{theorem} \label{thm: upper bound for MinIndex}
    For no selection of parameters $k$ and $p_1$, $p_2$, \ldots $p_k$, MinIndex achieves a guarantee larger than $5/9$ on forests of maximum degree three.
\end{theorem}

In the remaining part of this appendix, we prove Theorem~\ref{thm: upper bound for MinIndex}. For this, we consider two families of instances. The first family is constructed below. The second family is constructed based on consistent instances from Section~\ref{sec: consistent instances} but with a modified edge arrival order. The instances in both families are parametrized by a parameter $n$.

    Let $k$ be a natural number and $p_1$, \dots, $p_k$ be nonnegative numbers such that $p_1+p_2+\ldots +p_k=1$. Let $M_1$, \dots, $M_k$ be the matchings computed by MinIndex, i.e. by Algorithm~\ref{alg: MinIndex}. Let~$M$ be a random variable indicating the matching output by Algorithm~\ref{alg: MinIndex}. We denote by $\gamma$ the guarantee achieved by MinIndex with the parameters $k$ and $p_1$, \dots, $p_k$.

    \subsection*{First Family}

    Let $n$ be a natural number. Let us describe the edges that are going to arrive at the beginning. The first edges to arrive form a path~$P$, consisting of the edges $e_1$, $e_2$,\ldots, $e_{3n+3}$. Here, the edges $e_j$ and $e_{j+1}$ are incident for every $j=1,\ldots,3n+2$.
    The first three batches to arrive are as follows:
    \begin{itemize}
        \item $B_1:=\{e_i \mid i \equiv 2 \mod 3\}=\{e_2, e_5, e_8,\ldots, e_{3n+2}\}$
        \item $B_2:=\{e_i \mid i \equiv 1 \mod 3\} \,\cup \,\{e_{3n+3)}\} =\{e_1, e_4, e_7,\ldots, e_{3n+1}\}\,\cup\, \{e_{3n+3}\}$
        \item $B_3:=\{e_i \mid i \equiv 0 \mod 3,\, i \leq 3n\} = \{e_3, e_6, e_9, \dots, e_{3n}\}$.
    \end{itemize}

    Let us describe the batches $B_4$, $B_5$, $B_6$, $B_7$. Let us first describe the structure of the edges in these batches. For this, we iterate over the vertices $u$ on the path $P$ which are not incident to $e_1$ nor to $e_{2n+3}$. For each such vertex $u$ we construct the following edges:
    \begin{itemize}
        \item if $\delta(u)$ has no edges in $B_3$ then we construct an edge $uw^u$ and place it in $B_6$
        \item if $\delta(u)$ has no edges in $B_1$ then we construct edges $uw^u$, $w^uv^u$, $v^ut^u$ and place them in $B_7$, $B_4$ and $B_5$, respectively
        \item if $\delta(u)$ has no edges in $B_2$ then we construct edges $uw^u$, $w^uv^u$, $v^ut^u$, $t^ur^u$, $v^uq^u$ and place them in $B_7$, $B_5$, $B_4$, $B_5$, $B_6$, respectively.
    \end{itemize}
An example of the underlying graph for the case $n=2$ is shown in Figure~\ref{graph: gadget}. 

It is straightforward to verify that with this arrival order, we have $M_1=B_1\cup B_4$, $M_2=B_2\cup B_5$, $M_3=B_1\cup B_6$, and $M_4=B_7$ as in Figure~\ref{graph: gadget}. It is also straightforward to verify that the constructed graph always has a perfect matching, showing thus that the cardinality of a maximum matching equals $6n+2$.

    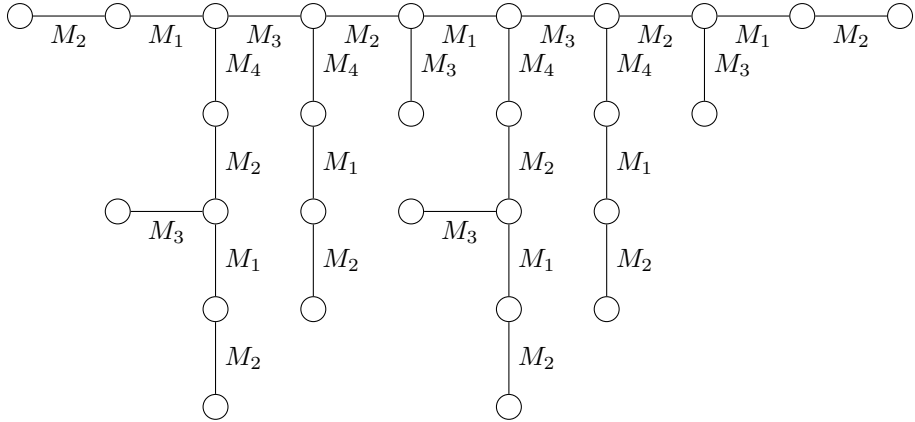
\begin{figure}[H]
        \centering
        \begin{tikzpicture}[node distance={13mm}, main/.style = {draw, circle}] 
            \node[main] (v1) {}; 
            \node[main] (v2) [right of=v1] {};
            \node[main] (v3) [right of=v2] {};
            \node[main] (v4) [right of=v3] {}; 
            \node[main] (v5) [right of=v4] {}; 
            \node[main] (v6) [right of=v5] {}; 
            \node[main] (v7) [right of=v6] {}; 
            \node[main] (v8) [below of=v3] {}; 
            \node[main] (v9) [below of=v8] {}; 
            \node[main] (v10) [below of=v9] {}; 
            \node[main] (v11) [below of=v10] {}; 
            \node[main] (v12) [left of=v9] {}; 
    
            \node[main] (v13) [below of=v4] {}; 
            \node[main] (v14) [below of=v13] {}; 
            \node[main] (v15) [below of=v14] {}; 
    
            \node[main] (v16) [below of=v5] {}; 

            \node[main] (v17) [right of=v7] {}; 
            \node[main] (v18) [right of=v17] {}; 
            \node[main] (v19) [right of=v18] {}; 

            \node[main] (v8') [below of=v6] {}; 
            \node[main] (v9') [below of=v8'] {}; 
            \node[main] (v10') [below of=v9'] {}; 
            \node[main] (v11') [below of=v10'] {}; 
            \node[main] (v12') [left of=v9'] {}; 
    
            \node[main] (v13') [below of=v7] {}; 
            \node[main] (v14') [below of=v13'] {}; 
            \node[main] (v15') [below of=v14'] {}; 
    
            \node[main] (v16') [below of=v17] {};

            \draw (v1) -- node[midway, below, pos=0.5] {$M_2$} (v2);
            \draw (v2) -- node[midway, below, pos=0.5] {$M_1$} (v3);
            \draw (v3) -- node[midway, below, pos=0.5] {$M_3$} (v4);
            \draw (v4) -- node[midway, below, pos=0.5] {$M_2$} (v5);
            \draw (v5) -- node[midway, below, pos=0.5] {$M_1$} (v6);
            \draw (v6) -- node[midway, below, pos=0.5] {$M_3$} (v7);
    
            \draw (v3) -- node[midway, right, pos=0.5] {$M_4$} (v8);
            \draw (v8) -- node[midway, right, pos=0.5] {$M_2$} (v9);
            \draw (v9) -- node[midway, right, pos=0.5] {$M_1$} (v10);
            \draw (v10) -- node[midway, right, pos=0.5] {$M_2$} (v11);
            \draw (v9) -- node[midway, below, pos=0.5] {$M_3$} (v12);
    
            \draw (v4) -- node[midway, right, pos=0.5] {$M_4$} (v13);
            \draw (v13) -- node[midway, right, pos=0.5] {$M_1$} (v14);
            \draw (v14) -- node[midway, right, pos=0.5] {$M_2$} (v15);
    
            \draw (v5) -- node[midway, right, pos=0.5] {$M_3$} (v16);

            \draw (v7) -- node[midway, below, pos=0.5] {$M_2$} (v17);
            \draw (v17) -- node[midway, below, pos=0.5] {$M_1$} (v18);
            \draw (v18) -- node[midway, below, pos=0.5] {$M_2$} (v19);

            \draw (v6) -- node[midway, right, pos=0.5] {$M_4$} (v8');
            \draw (v8') -- node[midway, right, pos=0.5] {$M_2$} (v9');
            \draw (v9') -- node[midway, right, pos=0.5] {$M_1$} (v10');
            \draw (v10') -- node[midway, right, pos=0.5] {$M_2$} (v11');
            \draw (v9') -- node[midway, below, pos=0.5] {$M_3$} (v12');
    
            \draw (v7) -- node[midway, right, pos=0.5] {$M_4$} (v13');
            \draw (v13') -- node[midway, right, pos=0.5] {$M_1$} (v14');
            \draw (v14') -- node[midway, right, pos=0.5] {$M_2$} (v15');
    
            \draw (v17) -- node[midway, right, pos=0.5] {$M_3$} (v16');
            
        \end{tikzpicture}
        \caption{The instance from the first family for $n = 2$.}
        \label{graph: gadget}
    \end{figure}

    Moreover, the expected cardinality of the matching produced by Algorithm~\ref{alg: MinIndex} is as follows,
    \[\mathbb{E}[|M|] = \sum_{i = 1}^k |M_i|p_i = (3n + 1)p_1 + (4n + 2)p_2 + 3np_3 + 2np_4\]
    So we get the following constraint on the competitiveness $c$ of MinIndex~\ref{alg: MinIndex},
    \begin{equation*}
        \gamma \leq \frac{\mathbb{E}[|M|]}{6n+2} = \frac{3n+1}{6n+2}p_1 + \frac{4n+2}{6n+2}p_2 + \frac{3n}{6n+2}p_3 + \frac{2n}{6n+2}p_4\,,
    \end{equation*}
     taking the limit as $n \rightarrow \infty$ we get,

    \begin{equation} \label{c constraint 1}
        \gamma \leq \frac{1}{2}p_1 + \frac{2}{3}p_2 + \frac{1}{2}p_3 + \frac{1}{3}p_4\,.
    \end{equation}

\subsection*{Second Family}
    We now consider the second family of instances defined on the same graphs as the consistent instances from Section~\ref{sec: consistent instances} but with a different arrival order. Let $n$ be an even natural number. We have the following batches $B_1$, $B_2$ and $B_3$
    \begin{itemize}
        \item $B_1:=\{e^l_i, e^r_i \mid i \equiv 1 \mod 2,\,\}\cup\{e_1\}$
        \item $B_2:=\{e^l_i, e^r_i \mid i \equiv 0 \mod 2\}$
        \item $B_3:=\{\hat{e}^l_i, \hat{e}^r_i \mid i = 1, \dots, n-2\}$.
    \end{itemize}
    So by Algorithm~\ref{alg: MinIndex} we have $M_1=B_1$, $M_2=B_2$ and $M_3=B_3$. The underlying graph again has a perfect matching, so the cardinality of a maximum matching is $2n-2$. Thus, we have the following constraint on the guarantee $\gamma$ achieved by MinIndex,
    \begin{equation*}
        \frac{\mathbb{E}[|M|]}{2n} = \frac{n-1}{2(n-1)}p_1 + \frac{n}{2(n-1)}p_2 + \frac{2(n-2)}{2(n-1)}p_3 \geq \gamma
    \end{equation*}
     and therefore taking the limit as $n \rightarrow \infty$, we get the following constraint
    \begin{equation} \label{c constraint 2}
        \frac{1}{2}p_1 + \frac{1}{2}p_2 + p_3 \geq \gamma\,.
    \end{equation}

    Finally, the trivial constraints based on consistent instances with $n=1$ and $n=2$ give the following constraints
    \begin{equation} \label{c constraint 3}
        p_1 \geq \gamma \ \ \ \ \ \text{and} \ \ \ \ \ \frac{1}{2}p_1 + p_2 \geq \gamma.
    \end{equation}

\subsection*{Linear Program}
    So combining constraints~\eqref{c constraint 1}, \eqref{c constraint 2}, and \eqref{c constraint 3} along with probability constraints gives the following Linear Program bounding $\gamma$
    \begin{maxi*}
    {}{\gamma}{}{}
    \addConstraint{p_1 \geq \gamma}
    \addConstraint{\frac{1}{2}p_1 + p_2 \geq \gamma}
    \addConstraint{\frac{1}{2}p_1 + \frac{1}{2}p_2 + p_3 \geq \gamma}
    \addConstraint{\frac{1}{2}p_1 + \frac{2}{3}p_2 + \frac{1}{2}p_3 + \frac{1}{3}p_4 \geq \gamma}
    \addConstraint{p_1 + p_2 + p_3 + p_4 \leq 1}
    \addConstraint{p_1,\, p_2,\, p_3,\, p_4 \geq 0}\,.
    \end{maxi*}
 The above Linear Program achieves the optimal values $ 5/9$, where the optimal solution sets the parameters $p_1 = 5/9$, $p_2 = 3/9$, $p_3 = 1/9$, and $p_4 = 0$. Note that these parameters are exactly the parameters for which MinIndex from~\cite{Buchbinder} achieves the guarantee $5/9$ on all forests.

\end{appendix}

\end{document}